\documentclass{article}

\usepackage{preamblebb}
\usepackage{caption}
\usepackage{subcaption}
\usepackage{refcount}

\usepackage[title,titletoc]{appendix}

\usepackage[export]{adjustbox}

\title{3D Tensor Renormalisation Group at High Temperatures}
\date{\normalsize{\textit{Institut des Hautes Études Scientiﬁques, 91440 Bures-sur-Yvette, France, ebel@ihes.fr}}}

\usepackage{amsthm}

\theoremstyle{plain}
\newtheorem{trm}{Theorem}[section]
\newtheorem{lem}[trm]{Lemma}
\newtheorem{cor}[trm]{Corollary}

\theoremstyle{definition}
\newtheorem{dfn}[trm]{Definition}

\theoremstyle{remark}
\newtheorem{rmk}[trm]{Remark}

\usepackage[colorlinks=true]{hyperref}   
\hypersetup{
    unicode=false,          % non-Latin characters 
    pdftoolbar=true,        % show Acrobat
    pdfmenubar=true,        % show Acrobat
    pdffitwindow=false,     % window fit to page when opened
    pdfstartview={FitH},    % fits the width of the page to the window
    pdftitle={3D Tensor RG at High Temperatures},    % title
    pdfauthor={Nikolay Ebel},     % author
    pdfnewwindow=true,      % links in new window
    colorlinks=true,       % false: boxed links; true: colored links
    linkcolor=blue, %red,          % color of internal links (change box color with linkbordercolor)
    citecolor=blue,        % color of links to bibliography
    filecolor=blue,      % color of file links
    urlcolor=blue        % color of external links
}

\makeatletter
\providecommand*{\hrefurl}{\hyper@normalise\hrefurl@}
\providecommand*{\hrefurl@}[2]{\hyper@linkurl{#2}{#1}}
\makeatother

\usepackage{tikz}
\usepackage{fnpct}

\newcommand{\diag}{\hyperlink{targ:diags}{diag}}
\newcommand{\graph}{\hyperlink{targ:graphs}{g}}

\newcommand{\deltaA}{\mathbf{a}}
\newcommand{\deltaB}{\mathbf{b}}
\newcommand{\z}{z}
\newcommand{\s}{s}
\newcommand{\deltaU}{\mathbf{u}}

\usepackage{cite}

\begin{document}

\maketitle

\numberwithin{equation}{section}

\begin{abstract}
    Building upon previous $2D$ studies, this research focuses on describing $3D$ tensor renormalisation group (RG) flows for lattice spin systems, such as the Ising model. We present a novel RG map, which operates on tensors with infinite-dimensional legs and does not involve truncations, in contrast to numerical tensor RG maps. To construct this map, we developed new techniques for analysing tensor networks. Our analysis shows that the constructed RG map contracts the region around the tensor $A_*$, corresponding to the high-temperature phase of the $3D$ Ising model. This leads to the iterated RG map convergence in the Hilbert-Schmidt norm to $A_*$ when initialised in the vicinity of $A_*$. This work provides the first steps towards the rigorous understanding of tensor RG maps in $3D$.
\end{abstract}

\tableofcontents

\section{Introduction}

In this paper, we study the tensor network approach to the renormalisation group (RG) in the high-temperature phase of $3D$ lattice models. The transition from lattice models to tensor networks has been discussed in \cite{Kennedy2022a, Kennedy2022, Gilt, Xie}. Section~3 in \cite{Kennedy2022a} is of particular interest. While the main focus in \cite{Kennedy2022a} is on the $2D$ case, methods from Section~3 can be extended to $3D$. In the language of tensor networks, the "high-temperature phase" (or "high-temperature limit") means that all tensors in the network are near the fixed point tensor $A_*$, with a single non-zero tensor element $(A_*)_{000000}=1$.

The main result of our work is formulated in Theorem~\ref{maintheor} provided at the end of Section~\ref{basics}. Roughly speaking, Theorem~\ref{maintheor} states that a transformation of a tensor network (RG map) exists that:
\begin{itemize}
    \item reduces the size of the network;
    \item preserves the physical properties (i.e., it preserves the partition function);
    \item brings the tensors in the network closer to the fixed point tensor $A_*$.
\end{itemize}
Since the existing techniques of rigorous RG analysis \cite{Kennedy2022a},\cite{Kennedy2022} appear to be impossible to generalise to the $3D$ case, we develop a novel technique to prove Theorem~\ref{maintheor}. This technique is called rearrangement RG. It was first proposed by Tom Kennedy for $2D$ tensor networks in his unpublished notes \cite{Kennedy2022b} (reviewed in Appendix~\ref{2Drd}). From these notes, we take some ideas and adapt them to $3D$. Similar techniques were also applied in numerical studies, see the loop-TNR algorithm \cite{PhysRevLett.118.110504}.

Although we work in the high-temperature regime, we hope that, in the future, the methods introduced in this paper will help investigate other parts of the phase diagram. Our optimism is grounded in analytical and numerical advances in applying tensor RG to $2D$ systems. Specifically, it was shown in \cite{Kennedy2022} and \cite{Kennedy2022a} that the particular tensor RG techniques may be applied for rigorous study of high- and low-temperature limits of $2D$ lattice systems. Numerical studies go even further, showing that tensor RG algorithms are capable of yielding critical exponents in close agreement with known values for the $2D$ Ising model (see, e.g. \cite{Gilt},\cite{PhysRevLett.118.110504}).

Regarding 3D studies, to our knowledge, there are no analytical results for 3D tensor RG. On the numerical side, two algorithms, HOTRG \cite{Xie} and Gilt \cite{Gilt}, may be employed for 3D tensor networks. The HOTRG algorithm yielded critical exponents (and other physical quantities) consistent with Monte Carlo simulations and high-temperature expansions of the $3D$ Ising model. However, it has been reported that the critical temperature obtained by HOTRG may converge with increasing bond dimension slower than one might expect based on the results from \cite{Xie}.\footnote{This issue was noted in personal communications with Antoine Tilloy and Clément Delcamp.} Presumably, the issue is that HOTRG performs a coarse-graining step but not a proper RG transformation, i.e., some details of the UV physics survive the coarse-graining.

Some of such UV physics details can be described by so-called corner-double-line (CDL) tensors \cite{Gilt}, \cite{Kennedy2022} and their generalisation to $3D$ discussed in \cite{Gilt} (see Section VII) and \cite{lyu2023essential}. We call corner-multi-line (CML) tensors the CDL tensors and their $3D$ generalisation. The CML tensors carry irrelevant UV information, making retrieving any relevant IR information from the network complicated. For example, CML tensors complicate the identification of the IR fixed point, as there are equivalent tensors that differ only by CML contributions. Additionally, irrelevant UV information consumes valuable computational resources. If this issue is not addressed, the tensor network will mainly include UV information after several RG iterations. The filtration of UV information/CML tensors, called disentanglement, is an essential part of any proper RG transformation.

The Gilt algorithm performs a proper RG transformation of a tensor network. It appears to offer adequate tools for the disentanglement in various dimensions. It demonstrates impressive results in $2D$, yielding the scaling dimensions of the critical $2D$ Ising model with good precision (see Table I in \cite{Gilt}). However, no $3D$ results have been published yet. The reason is that even though Gilt has a quite low computational cost\footnote{The Gilt's computational cost in $3D$ is $O(\chi^{12})$, where $\chi$ is the bond dimension. This is close to $O(\chi^{11})$ required to contract two $6$-leg tensors.}, it is still too high for appropriate numerical computations in $3D$ as for the same accuracy as in $2D$, these require higher bond dimensions (see section VII in \cite{Gilt}).

There are currently few papers on the topic of rigorous tensor RG approach. Because of this, we will keep our discussion self-contained and provide all the necessary definitions. However, a curious reader may find some additional background information (and interesting $2D$ results) in \cite{Kennedy2022a},\cite{Kennedy2022}. We also refer the reader to \cite{TRG, TNR, PhysRevLett.118.110504, EntRen, Gilt} for more details on the numerical tensor RG.

This paper is organised as follows. In Section~\ref{basics}, we provide the basic definitions and formulate the main theorem. In Section~\ref{cornstr}, we consider the simplest RG transformations and use these to reduce the Theorem~\ref{maintheor} to a simpler statement. In Sections~\ref{rRG},\ref{prf}, we give a detailed proof of Theorem~\ref{maintheor}. In Section~\ref{finalremarks}, we make final remarks and formulate a few problems for the future.

\section{Basic definitions}\label{basics}

Many definitions discussed in this and the following section were introduced in \cite{Kennedy2022a}. We provide them here for convenience.

\begin{dfn}\label{tensd}
    Let $n \geq 1$ be an integer and $\mathcal{I}$ a finite or countable nonempty set (called an index set). An $n$\textbf{-leg tensor} $A$ with legs indexed by $\mathcal{I}$ is defined as a map from $\mathcal{I}^n$ to $\mathbb{C}$: $(i_1,\ldots,i_n) \mapsto A_{i_1,\ldots,i_n}$.
\end{dfn}
An $n$-leg tensor will also be called "a tensor with $n$ legs/indices" or "$n$-tensor".

\begin{dfn}\label{HSn}
    The \textbf{Hilbert-Schmidt norm} of an $n$-leg tensor $A$ is defined as
    \begin{equation}
        \|A\| = \sqrt{\sum_{i_1,\ldots,i_n \in \mathcal{I}} |A_{i_1,\ldots, i_n}|^2}.
    \end{equation}
\end{dfn}
This is the norm used most often throughout the work. It will be referred to simply as "norm".

\begin{dfn}\label{Def_of_sum}
    For two $n$-leg tensors $A$ and $B$ indexed by $\mathcal{I}_A$ and $\mathcal{I}_B$ respectively we define the sum as follows:

    \begin{itemize}
        \item If $\mathcal{I}_A=\mathcal{I}_B$, then $A+B$ is defined as the usual element-wise sum.
        \item If $\mathcal{I}_A \neq \mathcal{I}_B$, we first extend each tensor by zeros to a larger tensor indexed by $\mathcal{I}_A \cup \mathcal{I}_B$, and then take the element-wise sum of extended tensors.
    \end{itemize}
\end{dfn}

\begin{rmk}
    The above definitions naturally generalize to cases where legs have distinct index sets.
\end{rmk}

In graphical notation, a 6-leg tensor is represented as six lines meeting at a point in 3D space. We will label tensors with small figures at the midpoint or their names nearby. Here are two examples:
\begin{equation}\label{solid}
    A = \includegraphics[valign=c]{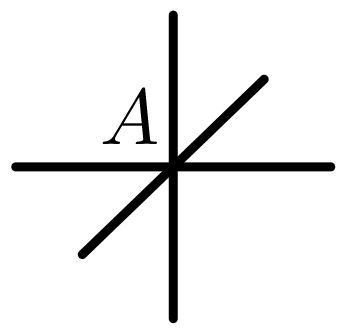}, \ \deltaA = \includegraphics[valign=c]{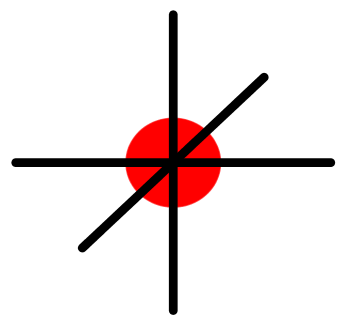}.
\end{equation}
No symmetry is assumed for $n$-leg tensors. Thus it is necessary to fix the correspondence between the indices $i_1,\ldots,i_n$ of a tensor component written in the standard notation: $A_{i_1 i_2 \ldots i_n}$, and the legs of the tensor component in the graphical notation. For a $6$-leg tensor, this correspondence is fixed by the following convention:
\begin{equation}\label{solidel}
    A_{i_1 j_1 i_2 j_2 i_3 j_3}=\includegraphics[valign=c,scale=0.35]{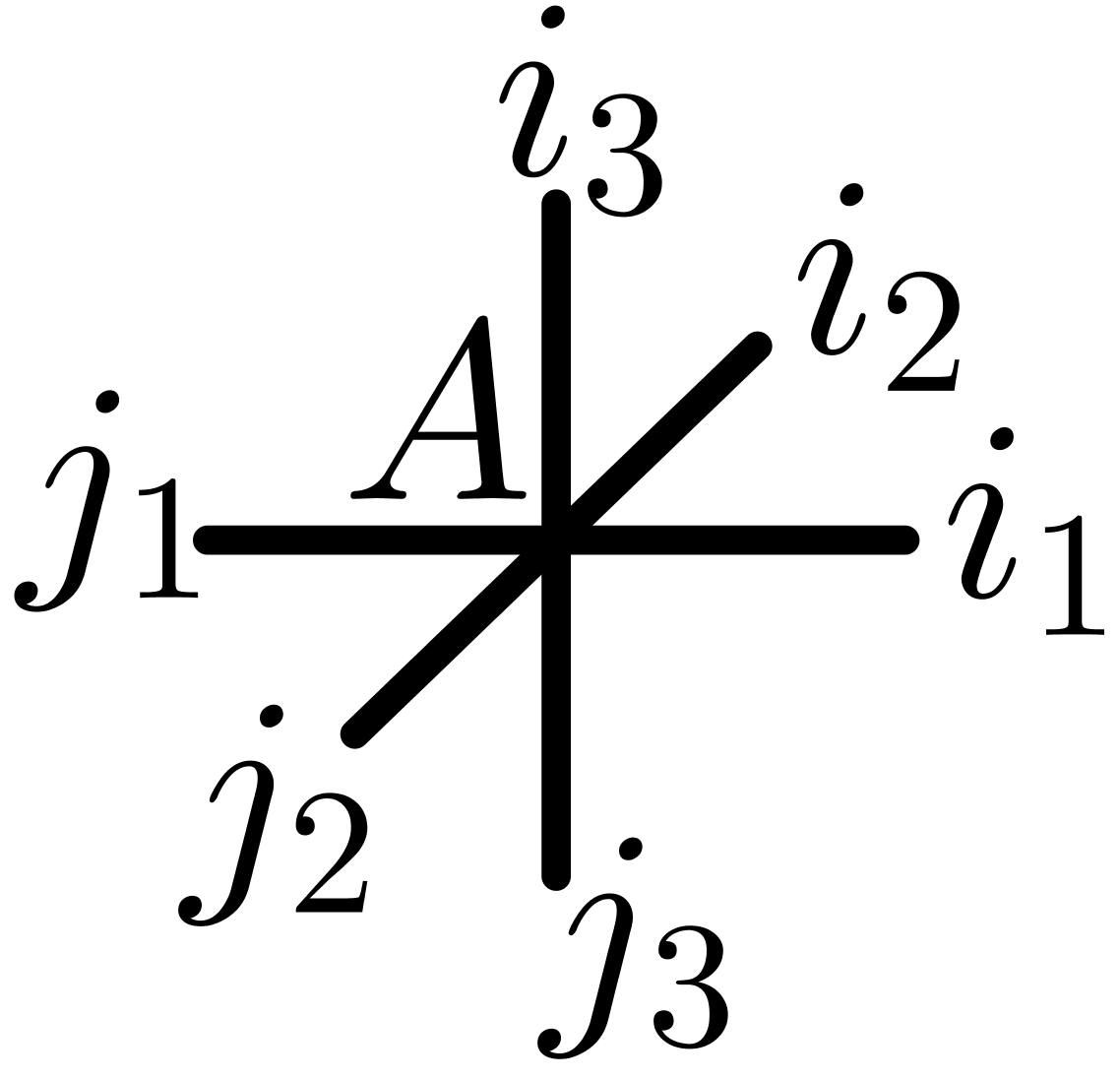}.
\end{equation}

The contraction of tensors is a necessary notion in this paper, which can be explained through the following example. Let $A$ and $B$ be $6$-tensors with the same index space $\mathcal{I}$. By contracting the first leg of $A$ with the second leg of $B$ it is possible to define the $11$-tensor $C$ as:
\begin{equation}\label{contraction}
    C_{i_1, \ldots, i_{11}} = \sum_{j\in \mathcal{I}} A_{j i_1 \ldots i_{5}} B_{ i_6 j i_7 \ldots i_{11}}.
\end{equation}
In the graphical notation, contraction is represented by glueing the corresponding legs of tensors:
\begin{equation}\label{contractionG}
    C= \includegraphics[valign=c, scale=0.66]{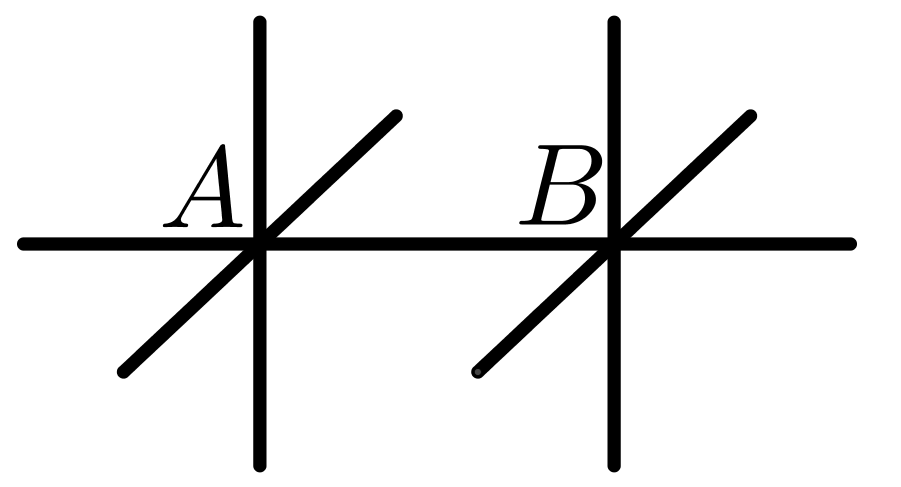}.
\end{equation}
Two contracted legs will be referred to as a \textbf{bond}. The contraction is defined only for tensors with the same index sets for each pair of contracted legs.

A \textbf{tensor network} is defined as a graph, every vertex of which is associated with a tensor and every edge with contracted legs. This association allows evaluation of the network by performing all contractions. This evaluation results in a number called the \textbf{partition function} of the tensor network.

Partition functions of many lattice models can be represented in the form of a tensor network partition function (see \cite{Kennedy2022}, \cite{Kennedy2022a}). Hence, the term "partition function" is suitable in this case.

\begin{figure}
    \centering
    \includegraphics[scale=1]{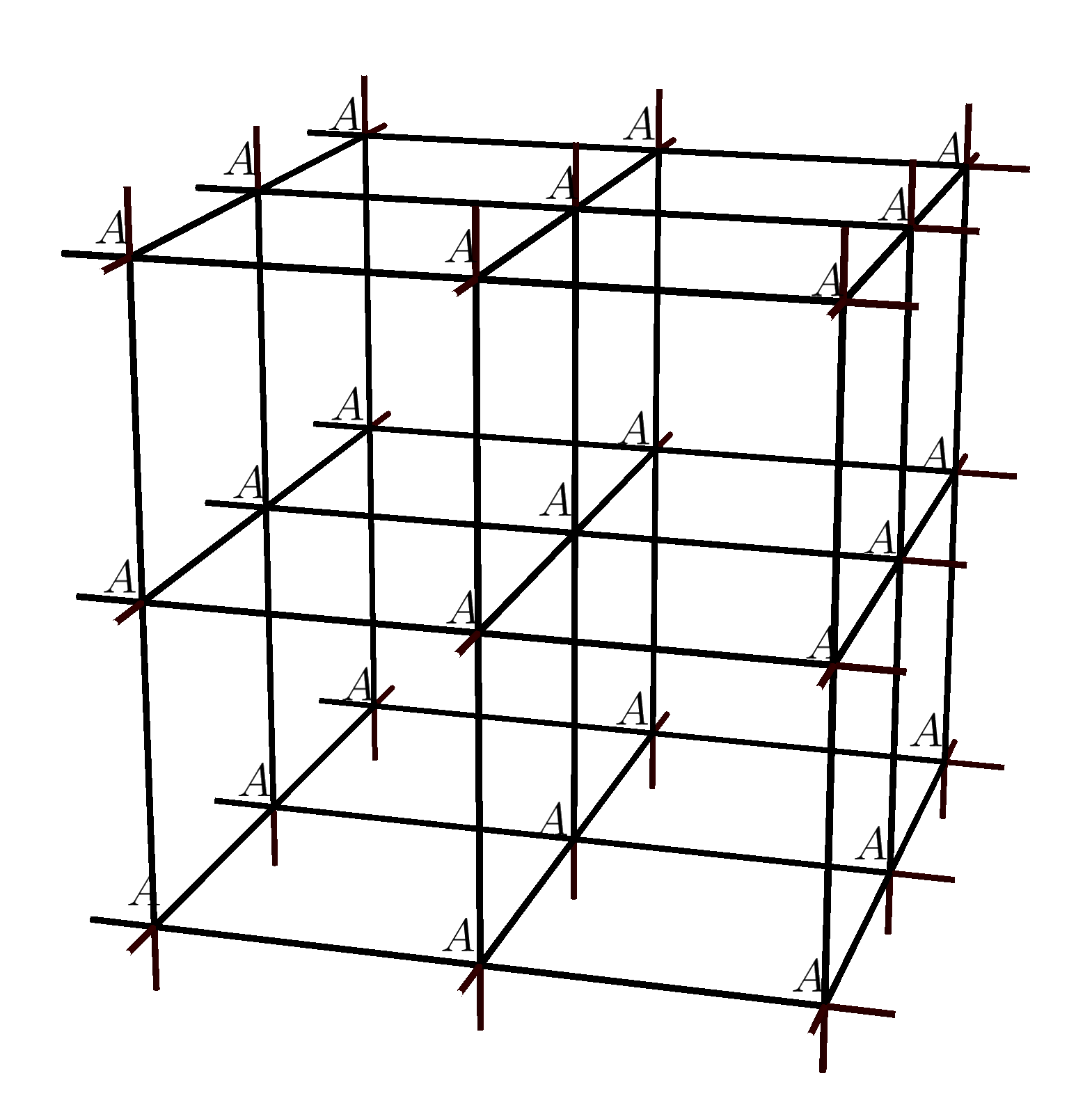}
    \caption{A cubic tensor network of size $3 \times 3 \times 3$. Boundary conditions are assumed to be periodic, i.e., legs of the tensors at opposite sides of the boundary are contracted}
    \label{fig:tensornetwork}
\end{figure}

In this work, we consider cubic networks with the same $6$-leg tensor at each vertex, see \figref{fig:tensornetwork}. A cubic network with tensor $A$ will be referred to as the network built of $A$. The partition function of such a network of size $N \times N \times N $ built of $A$ is denoted as $Z(A, N)$.

We are interested in tensors with a countable index set. We denote by $\mathcal{H}$ as the complex Hilbert space of all $6$-leg tensors of finite norm indexed by $\N_0=\N \cup \{0\}$, and by $\tilde{\mathcal{H}}$ as the subspace of $\mathcal{H}$ which consists of tensors $A$ satisfying the condition
\begin{equation}\label{devin0}
    A_{000000}=0.
\end{equation}

The central role will be played by the tensor $A_*$, which is called the \textbf{fixed point tensor} and defined as follows. It is the element of $\mathcal{H}$ with only one nonzero tensor component given by:
\begin{equation}\label{fipin0}
    (A_*)_{000000}=1.
\end{equation}
The term "fixed point" will be clarified later. It comes from the fact that there is a network transformation which maps $A_*$ to itself.

Let $U_\epsilon$ be the ball $\|A\|< \epsilon$ in $\tilde{\mathcal{H}}$. We study here the following domain of tensors:
\begin{equation}\label{ATens}
    U_\epsilon(A_*) = \{A_*+\deltaA \mid  \deltaA \in U_\epsilon\}.
\end{equation}

\begin{rmk}\label{strt}
    For any $6$-leg tensor $A \in \mathcal{H}$ with nonzero $000000$ component, there exists $\deltaA \in \tilde{\mathcal{H}}$ such that:
    \begin{equation}\label{scalingA}
        A=\z  (A_* + \deltaA),
    \end{equation}
    where $\z = A_{000000}$. Equation \eqref{scalingA} will be referred to as the \textbf{normal decomposition} of $A$. Tensor $\deltaA$ will be referred to as \textbf{normal deviation} of $A$, and the factor $\z$ as \textbf{normal factor} of A.
\end{rmk}

\begin{rmk}\label{factorrem}
    The normal factor can always be factored out from the partition function:
    \begin{equation}\label{factorout}
        Z(A,N) = \z^{N^3}Z(A_*+\deltaA, N).
    \end{equation}
    That makes $U_\epsilon (A_*)$ worth studying in our context.
\end{rmk}

Let $\mathcal{I}$ and $\mathcal{J}$ be some index sets. Let $\chi_{\mathcal{J}}$ be the characteristic function of $\mathcal{J}$ in $\mathcal{I}$. Its domain is $\mathcal{I}$ and it is defined by the following formula:
\begin{equation}
    \chi_{\mathcal{J}}(i) = \begin{cases}
        1, & \text{if } i \in \mathcal{J} \cap \mathcal{I}    \\
        0, & \text{if } i \notin \mathcal{J} \cap \mathcal{I}
    \end{cases}.
\end{equation}

\begin{dfn}\label{restr}
    Let $A$ be an $n$-leg tensor indexed by $\mathcal{I}$. The \textbf{restriction} of $A$ with the first leg restricted to $\mathcal{J}$ is defined as the $n$-leg tensor $\bar{A}$ indexed by $\mathcal{I}$ with tensor elements given by the following formula:
    \begin{equation}\label{restrd}
        \bar{A}_{i_1, \ldots, i_n} = \chi_{\mathcal{J}} (i_1) A_{i_1,\ldots,i_n}.
    \end{equation}
\end{dfn}

\begin{rmk}
    This definition admits natural generalisation for an arbitrary set of restricted legs.
\end{rmk}

In graphical notation, the restriction is indicated by either placing a $\mathcal{J}$ label near the relevant leg, as shown for $\bar{A}$ in \eqref{restrd_example} below, or by altering the line style, such as using dashed lines or ticks.
\begin{equation}\label{restrd_example}
    \bar{A}_{i_1 i_2 \ldots i_n} = \includegraphics[scale=0.4,valign=c]{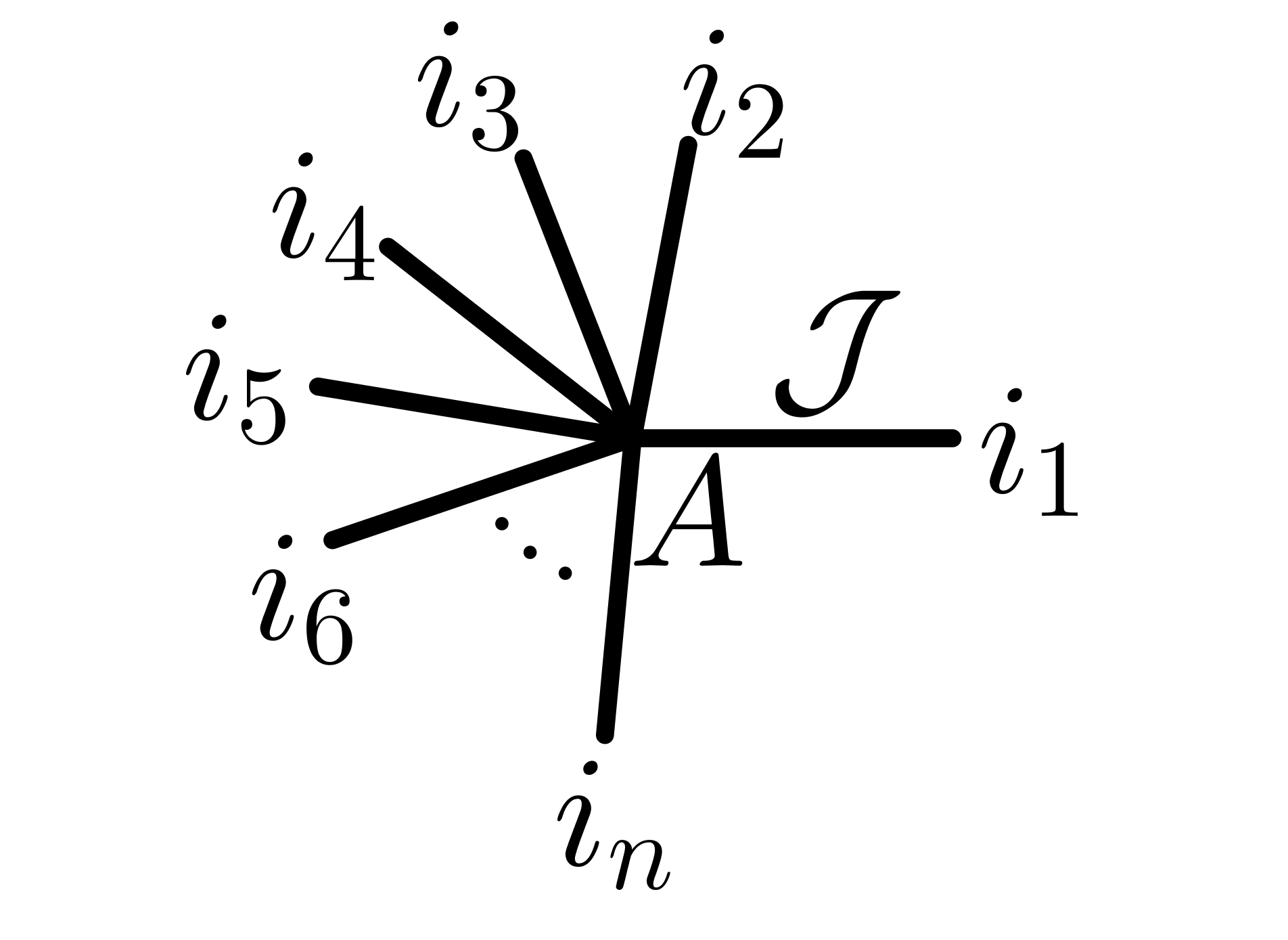};
\end{equation}
Restrictions to $\{0\}$ are represented by a \textbf{dashed line}. For example, equation \eqref{devin0} in this notation looks as following:
\begin{equation}
    \includegraphics[valign=c]{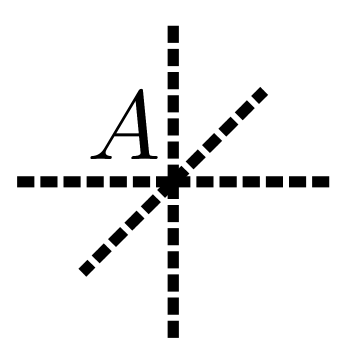} =0.
\end{equation}

Note that $A_*$ satisfies the equation:
\begin{equation}\label{Astrestr}
    \includegraphics[valign=c]{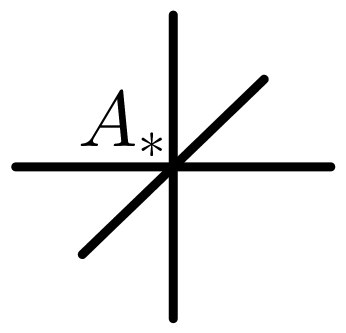}=\includegraphics[valign=c]{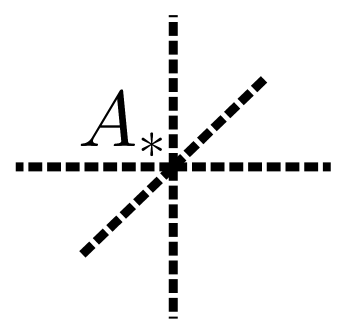},
\end{equation}
because all other components but $(A_*)_{000000}$ vanish.\footnote{Equation \eqref{Astrestr} can be written in the standard notation as $
        (A_*)_{i_1 i_2 \ldots i_6 }=\prod_{k=1}^6 \chi_{\{0\}}(i_k) \,
        (A_*)_{i_1 i_2 \ldots i_6 }. $}

For this reason, we draw $A_*$ with all legs dashed. Moreover, the label $A_*$ is dropped from the graphical notation due to its special role and frequent appearance in equations. Thus, $A_*$ will be depicted as:
\begin{equation}\label{FixedP}
    A_*=\includegraphics[valign=c]{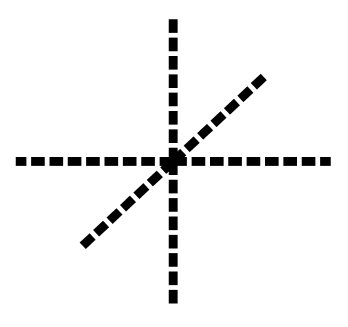}.
\end{equation}

\begin{dfn}\label{RGmap}
    Let $\s$ be a positive integer. An \textbf{RG map} with a lattice rescaling factor $\s$ is defined to be a map of $6$-leg tensors $A \mapsto A'$ with the property that for any $N$ divisible by $\s$:
    \begin{equation}\label{RG}
        Z(A,N)=Z(A', N/\s).
    \end{equation}
\end{dfn}

The main result of this paper is the following theorem.

\begin{trm}\label{maintheor}
    For any sufficiently small $\epsilon >0$ there exists an RG map $A \mapsto A'$ defined on $U_\epsilon (A_*)$ with the lattice rescaling factor $\s=4$, such that:
    \begin{subequations}\label{maincond}
        \begin{align}
             & \deltaA' \text{ and } \z' \text{(normal deviation and factor of $A'$) depend analytically on }\deltaA;\label{analit} \\
             & \|\deltaA'\| < t \epsilon^{h}.\label{contr}
        \end{align}
    \end{subequations}
    Here $t$ and $h>1$ are universal constants independent of $A$ and $\epsilon$.
\end{trm}
\begin{rmk}
    In Section~\ref{theexponent}, we will show that our construction yields $h=16/15$.
\end{rmk}
\begin{rmk}
    In addition, it will be shown that $\z' = 1+O(\epsilon^{h})$.
\end{rmk}

The proof of this theorem is divided into three following sections.

\section{Corner structure}\label{cornstr}

Let us provide two crucial examples of RG maps. Maps from both examples are defined on $\mathcal{H}$ as it is the case of interest in the context of Theorem \ref{maintheor}.

The first example is a gauge transformation of a tensor. It is an RG map with lattice rescaling factor~$1$.

Let $G$ be an invertible bounded operator on $l_2(\N_0)$ whose inverse $G^{-1}$ is also bounded. In graphical notation, operators $G$ and $G^{-1}$ are denoted as follows:
\begin{equation}
    G_{ij}=j\includegraphics[scale=1,valign=c]{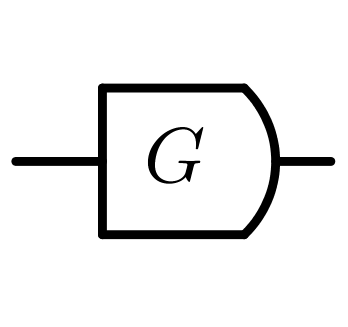}i, \ (G^{-1})_{ij}=j\includegraphics[scale=1,valign=c]{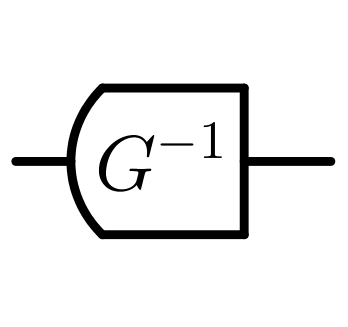}i.
\end{equation}
Let $A$ be an element of $\mathcal{H}$ and $A'$ be the tensor defined as follows:
\begin{equation}\label{gaugex}
    A'_{i_1 j_1 i_2 j_2 i_3 j_3} = \sum_{i'_1, j'_1 \in \N_0} G_{i_1 i'_1} A_{i'_1 j'_1 i_2 j_2 i_3 j_3 } (G^{-1})_{j'_1 j_1},
\end{equation}
or in the graphical notation:
\begin{equation}
    A'= \includegraphics[scale=1,valign=c]{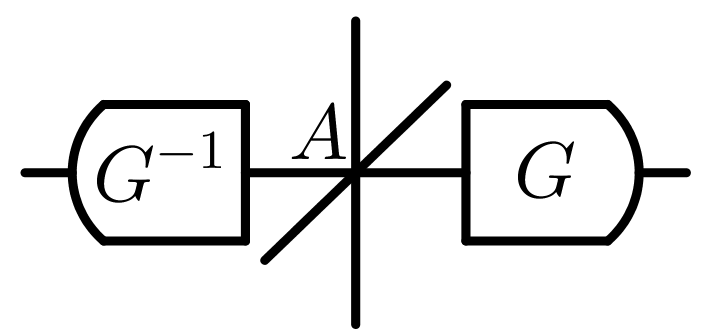}.
\end{equation}
The map $A \mapsto A'$ is an example of a gauge transformation.

Note that operators $G$ and $G^{-1}$ cancel each other in a network built of $A'$ as it is shown in this example:
\begin{equation}
    \includegraphics[scale=1,valign=c]{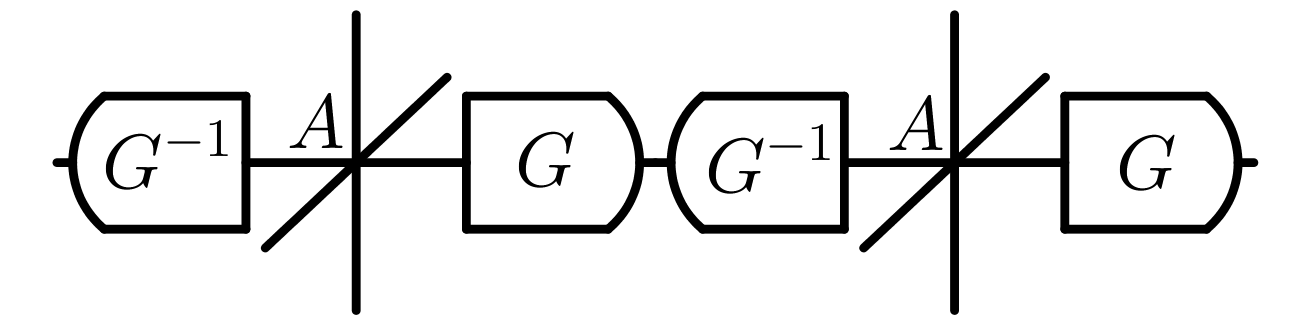}  =  \includegraphics[scale=1,valign=c]{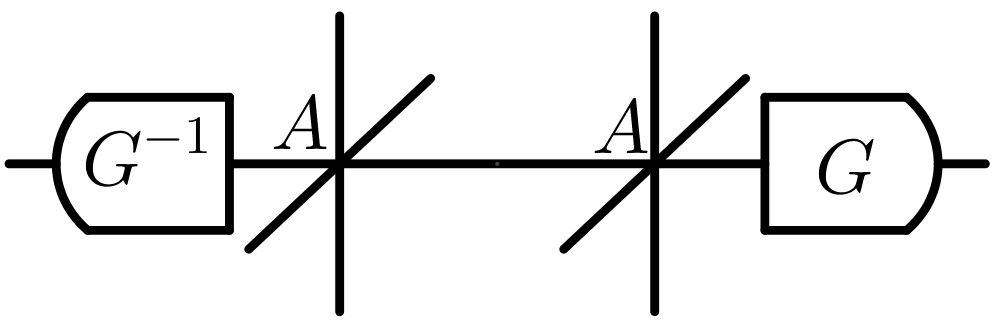} .
\end{equation}
Hence, the partition function remains unchanged after this transformation.

Note that the norm of $A'$ is finite and bounded by $\|A\| \|G\|_{op} \|G^{-1}\|_{op}$, where $\| \ \|_{op}$ is the operator norm.

Transformation analogous to \eqref{gaugex} can be performed along other directions, i.e., using indices $i_2 j_2$ or $i_3 j_3$ to contract $A$ with $G$ and $G^{-1}$. A general \textbf{gauge transformation} is defined as a composition of such transformations along all three directions (not necessarily with the same $G$).

The second example of an RG map is a simple RG step.

Let $A$ be an element of $\mathcal{H}$ and $T$ be the tensor defined by contracting $8$ copies of $A$ as follows:
\begin{equation}\label{TTens}
    T=\includegraphics[valign=c]{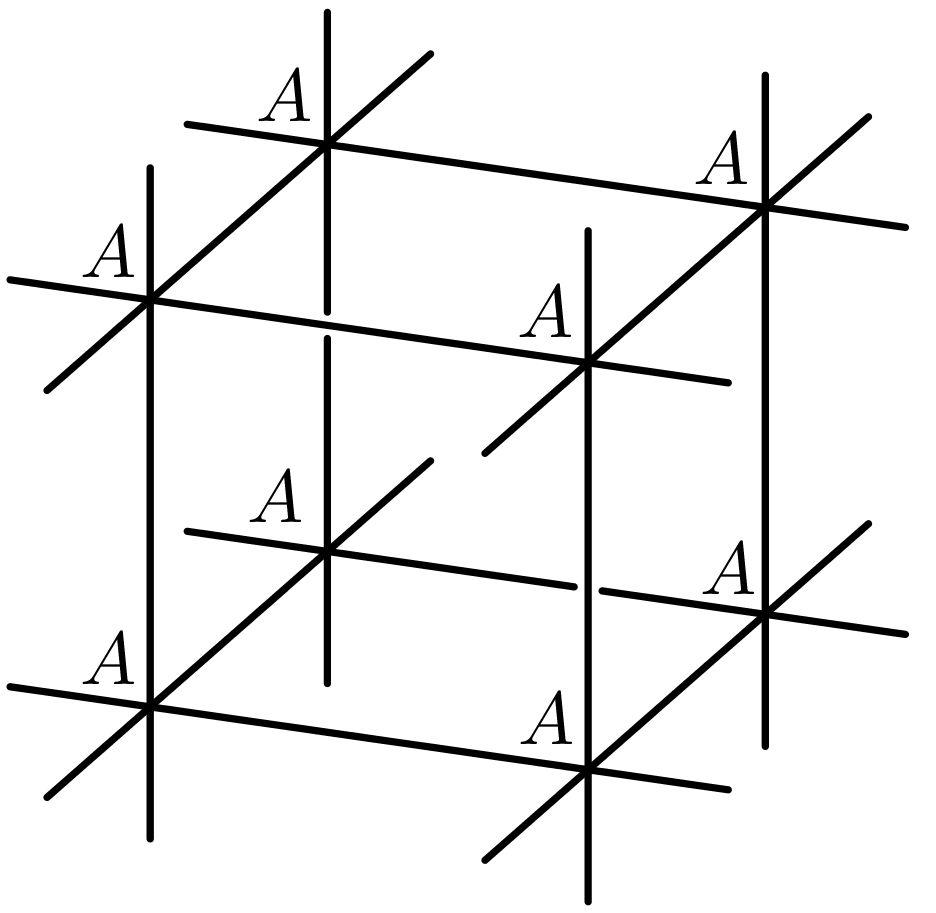}.
\end{equation}
Note that $T$ has finite norm (see Proposition 2.4 in \cite{Kennedy2022a}).

The tensor $T$ has $24$ legs, which can be naturally divided into $6$ groups of $4$, one group per face of the cube. One such group is highlighted here:
\begin{equation}\label{Thigh}
    \includegraphics[valign=c]{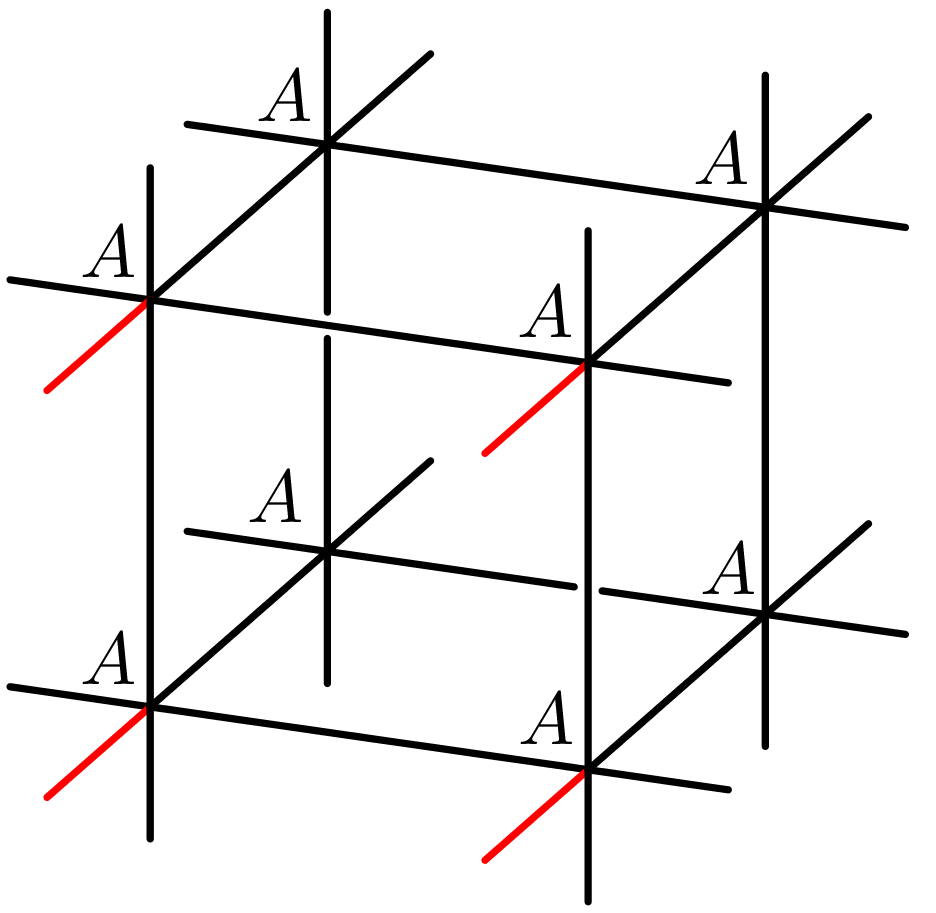}.
\end{equation}
Therefore, $T$ may be viewed as a $6$-leg tensor indexed by $\mathcal{I}=\N_0^4$. This is an example of a \textbf{leg grouping} operation, using terminology from \cite{Kennedy2022a}.

Clearly, $Z(T,N/2)=Z(A,N)$. Hence, the map $A \mapsto T$ is an RG map with lattice rescaling factor $2$. This map is called the \textbf{simple RG step}.

In what follows, it will be convenient to represent $T$ as a tensor indexed by $\N_0$ and not $\N_0^4$. This is achieved by using a \textbf{reindexing} operation (terminology from \cite{Kennedy2022a}). Let $\mathcal{I}_1$ and $\mathcal{I}_2$ be two index sets of the same cardinality, and let $\rho$ be a bijection from $\mathcal{I}_1$ to $\mathcal{I}_2$. Let $A$ be an $n$-leg tensor indexed by $\mathcal{I}_2$. Then, the $n$-leg tensor $A'$ indexed by $\mathcal{I}_1$, defined by:
\begin{equation}\label{reindex}
    A'_{i_1, \ldots, i_n} = A_{\rho(i_1) \ldots \rho(i_n)},
\end{equation}
is said to be obtained from $A$ by reindexing. Note that $A'$ has the same norm as $A$ and that the partition function remains unchanged after reindexing.

Consider again the tensor $T$ defined by \eqref{TTens}. As was discussed, leg grouping allows viewing $T$ as a 6-leg tensor indexed by $\N_0^4$. Reindexing allows a transformation of $T$ into a $6$-leg tensor $A'$ indexed by $\N_0$. Note that the map $\rho$ acting from $\N_0$ to $\N_0^4$ is not unique. We leave $\rho$ arbitrary up to the following condition:
\begin{equation}\label{0t0}
    \rho(0)=(0,0,0,0).
\end{equation}
The condition \eqref{0t0} plays the following role. It implies that
\begin{equation}\label{fpexpl}
    A'_{000000} = T_{(0,0,0,0) \ldots (0,0,0,0)}.
\end{equation}
When $A=A_*$, the r.h.s. of \eqref{fpexpl} is the only nonzero tensor component of $T$ and is equal to $1$. Thus, \eqref{0t0} guarantees that if $A=A_*$, then $A' = A_*$. In other words, \eqref{0t0} guarantees that $A_*$ is a fixed point of the simple RG step followed by reindexing.

Applying a combination of a gauge transformation and the simple RG step followed by reindexing to a tensor from $U_\epsilon (A_*)$, one can produce a tensor whose normal deviation has a corner structure, which is defined in Definition~\ref{corndef}. This property is essential for the proof of Theorem \ref{maintheor}.

\begin{rmk}
    In what follows, when a tensor is said to be $O(\epsilon^p)$, this means that its norm is bounded by $\kappa \epsilon^p$, where $\kappa$ is some universal constant independent of parameters defining the tensor.  For example, any tensor $C \in U_\epsilon$ is $O(\epsilon)$ with $\kappa =1$.
\end{rmk}

\begin{dfn}\label{corndef}
    A $6$-tensor $C$ belonging to $\tilde{\mathcal{H}}$ is said to have the \textbf{corner structure} if it is $O(\epsilon)$ and satisfies the following equation:
    \begin{equation}\label{corn}
        \includegraphics[valign=c]{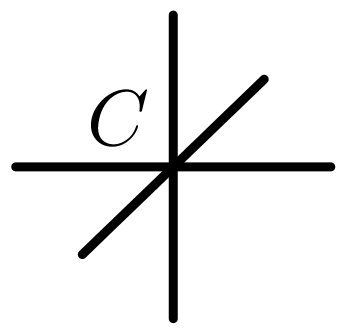}= \includegraphics[valign=c]{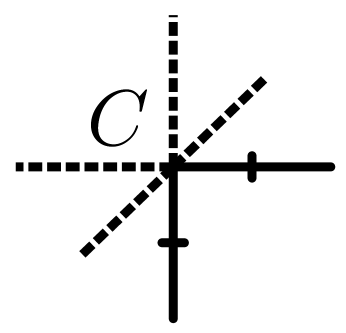}\!+\!11 \  +\includegraphics[valign=c]{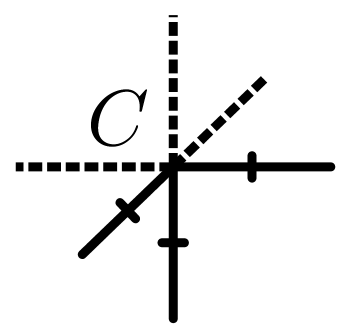}\!+\!7 \ + O(\epsilon^2).
    \end{equation}
    Here and below \textbf{a ticked line} represents the restriction to $\N$.
\end{dfn}

The notation "$+n$" refers to other restrictions of $C$ described by the patterns of ticked and dashed lines related to the shown pattern by rotations (similar notation is used in \cite{Kennedy2022a}). Eq.~\eqref{corn} means that up to the second order in $\epsilon$, the tensor $C$ is equal to its restrictions with only $2$ or $3$ legs restricted to $\N$ and all other legs restricted to $\{0\}$, with the additional condition that none of the legs restricted to $\N$ are opposite to each other.

The two following lemmas provide further details about the transformations which produce a tensor whose normal deviation has the corner structure.

\begin{lem}\label{gaugelem}
    There exists $\epsilon>0$ and a gauge transformation $A \mapsto A'$  defined on $U_{\epsilon}(A_*)$ and having the following properties:
    \begin{subequations}
        \begin{align}
             & \z' = 1+ O(\deltaA^2),  \text{ where $\z'$ is the normal factor of $A'$;} \label{nuanalit}                                                                  \\
             & \deltaA' = O(\deltaA),  \text{where $\deltaA'$ is the normal deviation of $A'$};                                                                            \\
             & \includegraphics[valign=c]{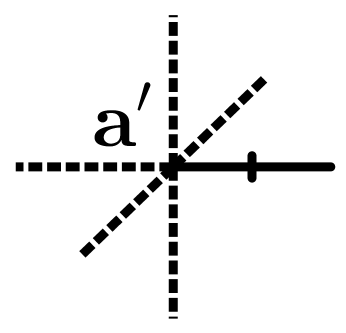}, \includegraphics[valign=c]{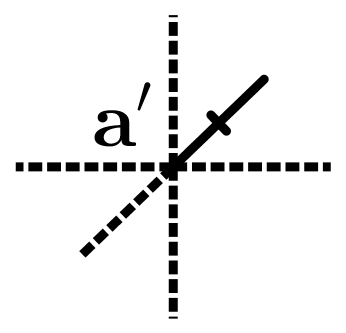}, \includegraphics[valign=c]{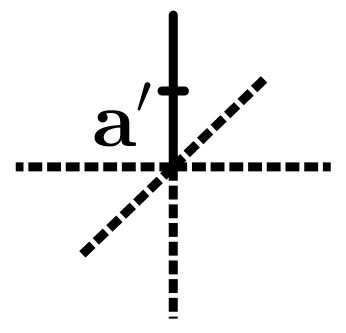}, \ldots  =  O(\deltaA^2) \label{oneleg}.
        \end{align}
    \end{subequations}
\end{lem}
The notation $O(\deltaA^n)$ will be used for analytic functions ($\mathbb{C}$ valued or $\mathcal{H}$ valued) of $\deltaA$ in domain $U_{\epsilon}$ which are bounded by $\kappa \|\deltaA\|^n$, where $\kappa$ is a constant independent of $\deltaA$ and $\epsilon$. The property \eqref{oneleg} means that any restriction of $\deltaA'$ with one leg restricted to $\N$ and other legs restricted to $\{0\}$ is $O(\deltaA^2)$.

Let $A'$ be the result of the gauge transformation from Lemma \ref{gaugelem} applied to a tensor $A \in U_\epsilon (A_*)$. Let $A''$ be the result of the simple RG map followed by reindexing applied to $A'$.

\begin{lem}\label{simplerglem}
    Tensor $\deltaA''$ (the normal deviation of $A''$) has the corner structure \eqref{corn}. In addition, $\z''$ (the normal factor of $A''$) satisfies the bound \eqref{nuanalit} and, for any sufficiently small $\epsilon$, $\deltaA''$ is $O(\deltaA)$.
\end{lem}

We omit the proof of these lemmas. In \cite{Kennedy2022}, analogous statements were proven for the $2D$ case; see Propositions 2.1 and 2.2 in \cite{Kennedy2022}. The generalisation to the $3D$ case is straightforward. The only noticeable difference is that in \cite{Kennedy2022}, all tensors are real, and the bounds, analogous to \eqref{nuanalit}, \eqref{oneleg}, etc., are given in terms of $O(\epsilon^p)$ and not in terms of $O(\deltaA^p)$. However, gauge operators $G$ from \cite{Kennedy2022} are given as convergent power series in $\deltaA$, which makes them suitable for the complex case and implies analytic dependence on $\deltaA$. Moreover, it is clear that $O(\epsilon^p)$ can be replaced by $O(\deltaA^p)$ in the proofs of Propositions 2.1 and 2.2 in \cite{Kennedy2022}, which leads to the bounds from Lemmas \ref{gaugelem} and \ref{simplerglem}.

The gauge transformation from the Lemma \ref{gaugelem}, followed by a simple RG map, followed by reindexing, is an RG map $A \mapsto  A''$ with lattice rescaling factor $2$ satisfying analyticity requirements \eqref{analit}. As discussed in Remark \ref{factorrem}, normal factor $\z''$ can be factored out from the partition function. Therefore, the proof of Theorem \ref{maintheor} is reduced to constructing an $\s=2$ RG map $A \mapsto A'$ such that conditions \eqref{maincond} are satisfied under the additional assumption that $\deltaA$ has corner structure \eqref{corn}. This will be done below.

\section{Rearrangement RG step for \texorpdfstring{$3D$}{3D}}\label{rRG}

In the previous $2D$ study \cite{Kennedy2022}, an RG map satisfying a property similar to \eqref{contr} was constructed.\footnotemark The construction from \cite{Kennedy2022} relied on special operators called disentanglers. These operators, acting on two bonds, allow the reduction of "entanglement" in a network, which was the key step in the whole procedure.
\footnotetext{The analogue of \eqref{contr} in \cite{Kennedy2022} was $\|\deltaA'\| \leq C \|\deltaA\|^{3/2}$, where $C$ is a universal constant. This property is stronger than \eqref{contr}, but the meaning is similar. Both properties mean that the RG map contracts any region around $A_*$ that is small enough. The map from \cite{Kennedy2022} was not analytic. In \cite{Kennedy2022a}, it was promoted to the analytic map with the following property: $\|\deltaA'\| \leq C \epsilon^{1/2} \|\deltaA\|$, which is still stronger than \eqref{contr}. See more details in Appendix~\ref{2Drd}}

\begin{rmk}
    We omit a rigorous definition of entanglement as it will not be used in what follows. Intuitively, tensor entanglement measures the effective dimensions of bonds in tensor contractions (see Remark 2.6 in \cite{Kennedy2022a}). For example, suppose $C$ is a contraction of two tensors $A$ and $B$. If the largest contribution to $C$ is given by only one index value of the bond connecting $A$ and $B$, then the entanglement of this contraction is low. Conversely, if contributions with many different bond index values are equally large, then entanglement is high.
\end{rmk}

In \cite{Kennedy2022a}, a method very similar to the one from \cite{Kennedy2022} was used to construct an RG map whose domain includes tensors representing the $2D$ Ising model at low temperatures. It was proven that the two ground states of the Ising model correspond to the two stable fixed points of the constructed RG map and that there is an unstable fixed point at which the first-order phase transition occurs.

Unfortunately, it appears to be impossible to use disentanglers acting on two bonds to generalise results of \cite{Kennedy2022} to $3D$. The problem has to do with entanglement in groups of tensors not lying in any plane, see examples in \figref{fig:entShapes}.\footnote{\label{refereequestion} The examples in \figref{fig:entShapes} also demonstrate that there are significantly more ways to entangle tensors around the cube in $3D$ than around a $2D$ plaquette. This could explain the slower decay of the environment spectra in the Gilt-TNR algorithm in the $3D$ case compared to the $2D$ one (see Figure 7 in \cite{Gilt}). This spectra behaviour is the main challenge for the Gilt-TNR algorithm in $3D$, as it implies that a higher bond dimension is required for obtaining the same precision as in $2D$.} We omit a detailed discussion of why the method of \cite{Kennedy2022} does not work in $3D$ (at least our attempts have not been successful) but will instead present a method which does work.

\begin{figure}
    \centering
    \includegraphics[scale=1]{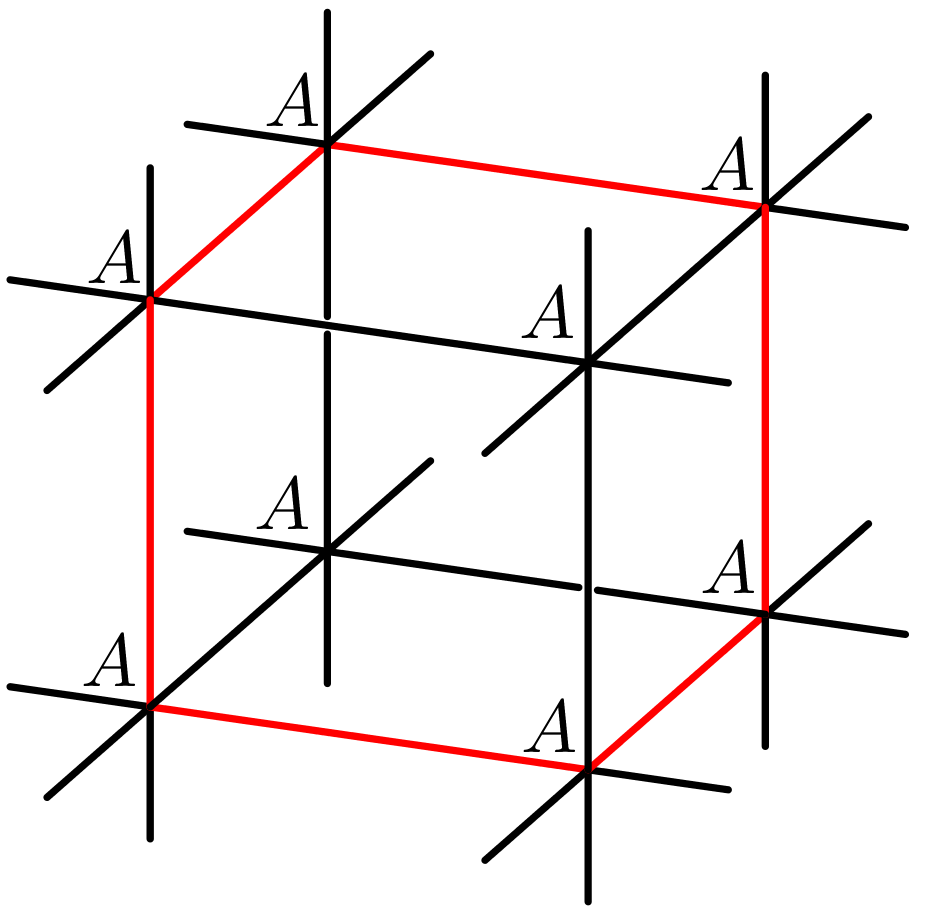}
    \includegraphics[scale=1]{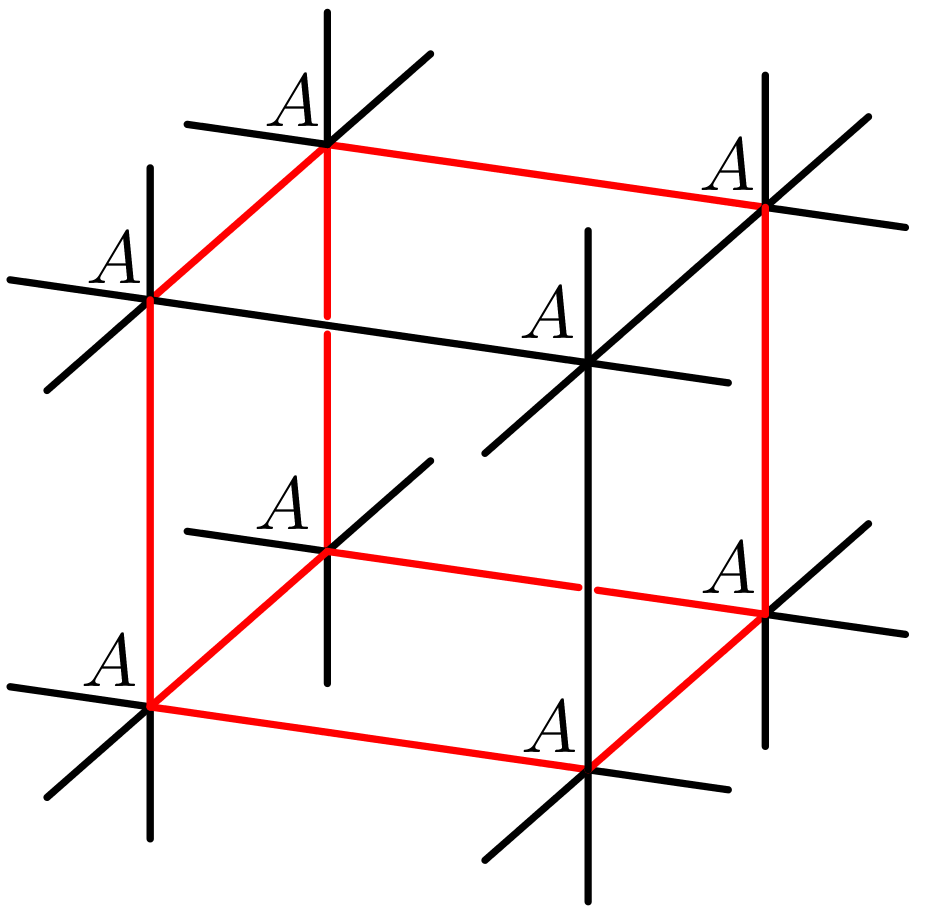}
    \caption{Examples of the groups of tensors (contracted by highlighted bonds), which cannot be disentangled by methods from \cite{Kennedy2022} generalised to three dimensions}
    \label{fig:entShapes}
\end{figure}

Tom Kennedy proposed an alternative construction of the RG map for the $2D$ case in his unpublished notes \cite{Kennedy2022b}.\footnote{In Appendix~\ref{2Drd}, we give a brief review of the ideas from \cite{Kennedy2022b} and explain what our contribution is to this story (apart from applying methods from \cite{Kennedy2022b} to $3D$).} This method, called "Rearrange disentanglers," can be generalised to three dimensions, which is done in this and the following sections. The idea is to define an RG map using a set of tensors $B_1, \ldots, B_8$ such that:
\begin{equation}\label{maineq}
    \includegraphics[scale=1,valign=c]{CubeA}=\includegraphics[scale=1,valign=c]{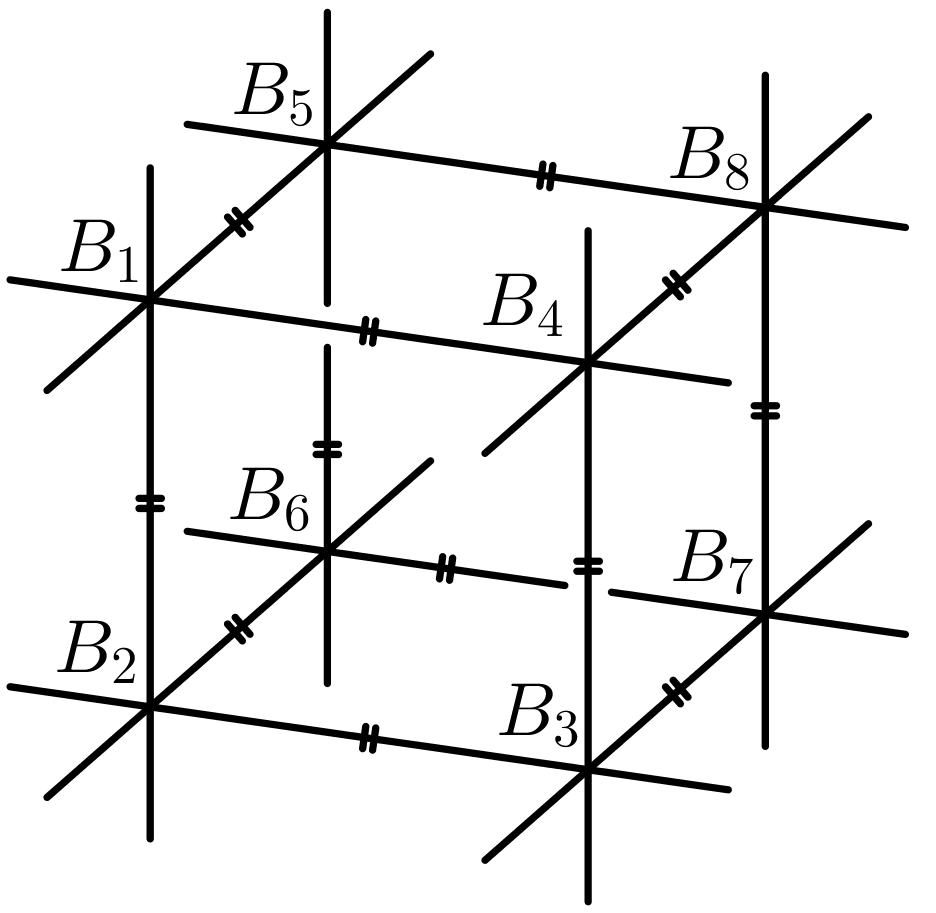}.
\end{equation}
Here and below, a \textbf{line with two ticks} represents the new index set $\mathcal{D}_0$, which will be defined precisely in Section~\ref{D0} below (briefly, the union of $\N_0$ with a finite number of disjoint copies of $\N$ and $\{0\}$). The property \eqref{maineq} allows the replacement of a tensor network built of $A$ tensors by the network built of $B_v$ ($v=1,\ldots,8$) tensors if the size of the network is even. This will be explained in more detail later in the proof of Theorem \ref{maintheor}.

Now, let us explain some terminology. The contraction of $B_v$ tensors defined by the r.h.s. of \eqref{maineq} is called $T'$:
\begin{equation}\label{Tprime}
    T'=\includegraphics[valign=c]{CubeB}.
\end{equation}
Equation \eqref{maineq} thus says $T=T'$.

Contractions of $8$ tensors, as in \eqref{Tprime} and \eqref{TTens}, will often appear in this paper. These will be referred to as \textbf{cubic contractions} with \textbf{insertions} of given tensors in given vertices. The vertices of such a contraction are labelled by numbers in the order shown in \figref{fig:labv}. Thus $T'$ is the cubic contraction with insertions of $B_1, \ldots, B_8$ in vertices $1, \ldots,8$, respectively. By convention, if insertions are specified for less than $8$ vertices, for example, if there are $n<8$ insertions, then the $A_*$ tensor is placed at the rest of the vertices.

\begin{figure}
    \centering
    \includegraphics[scale=1]{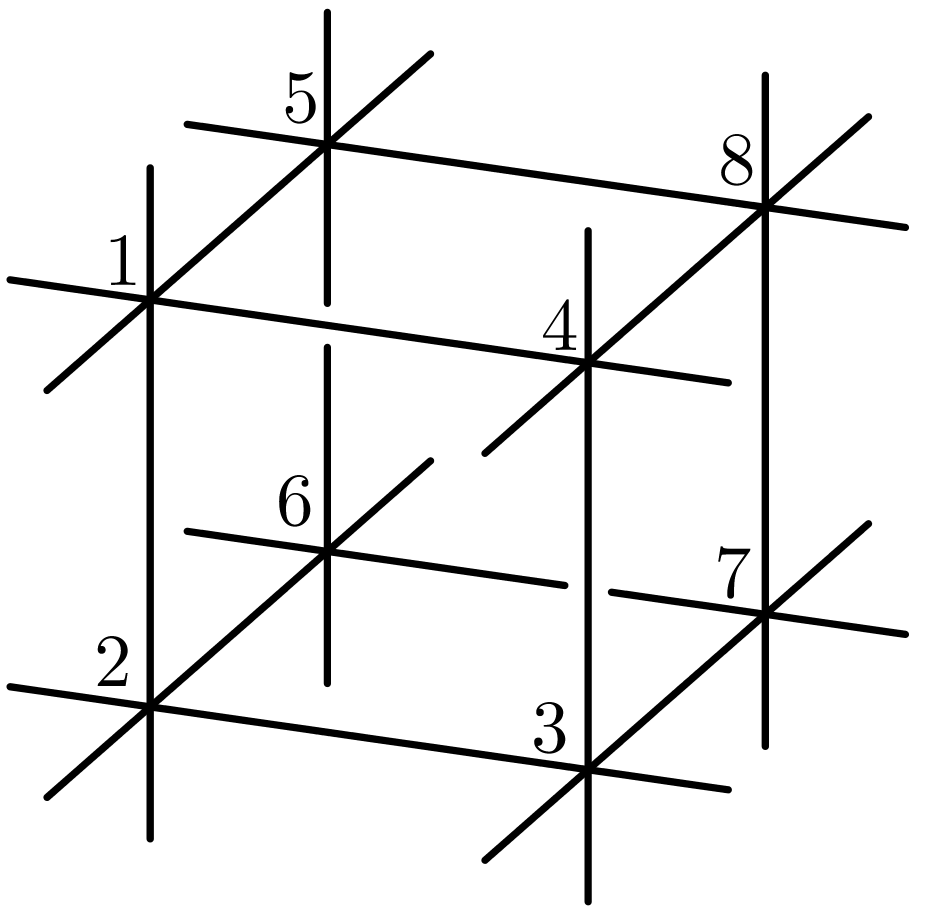}
    \caption{Labels of vertices of a cubic contraction}
    \label{fig:labv}
\end{figure}

It will also be important to divide the legs of $B_v$ tensors into internal and external leg groups with respect to the contraction $T'$. Let $L$ be a cubic contraction with insertions of tensors $C_1, \ldots, C_8$ in vertices $1, \ldots, 8$, respectively. The legs of a $C_v$ tensor involved in contractions within $L$ are called \textbf{$L$-internal}, while those not involved are called \textbf{$L$-external}. This division relies solely on the tensor's position in the contraction. For instance, consider tensor $A$ at the vertex $3$ in the contraction $T$, Eq.~\eqref{TTens}. Its legs $i_1 j_1 i_2 j_2 i_3 j_3$ (see the convention \eqref{solidel}) are divided as follows: $j_1,i_2, i_3$ are $T-$internal and  $i_1, j_2, j_3$ are $T-$external.

The proof of Theorem~\ref{maintheor} requires tensors $B_v$ of a special form. The following lemma explains what this form is and asserts the existence of the required tensors.

\begin{lem}\label{mainlem}
    There exist $a >1/2$ and $b>1$, such that for any $A=A_*+\deltaA$, where $\deltaA$ has corner structure \eqref{corn}, tensors $B_v$ satisfying \eqref{maineq} can be found in the form $B_v=A_*+\deltaB_v$, where $\deltaB_v $ tensors satisfy the following properties:
    \begin{subequations}\label{Pconds}
        \begin{align}
             & \ \text{tensors $\deltaB_v$ depend analytically on $\deltaA$};\label{P0} \\
             & \ \deltaB_v = O(\epsilon^a); \label{P1}                                  \\
             & \ \bar{\deltaB_v} = O(\epsilon^b); \label{P2}                            \\
             & \ (\deltaB_v)_{000000}=0. \label{P3}
        \end{align}
    \end{subequations}
\end{lem}

Here $\bar{\deltaB_v}$ denotes the restriction of $\deltaB_v$ with the $T'-$external legs of $\deltaB_v$ restricted to $\{0\}$.

It is useful to look at the properties \eqref{P1} and \eqref{P2} in the graphical notation. Here is an example for $\deltaB_1$:
\begin{subequations}
    \begin{align}
        \eqref{P1}: & \ \includegraphics[valign=c]{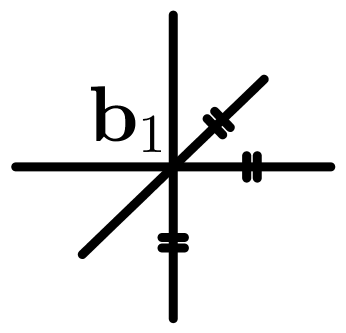} = O(\epsilon^a);   \\
        \eqref{P2}: & \  \includegraphics[valign=c]{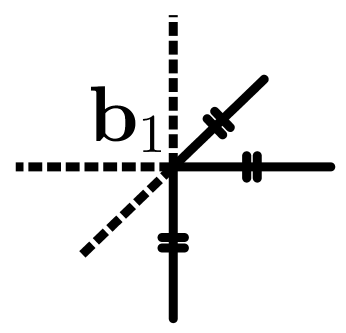} = O(\epsilon^b).
    \end{align}
\end{subequations}

Lemma \ref{mainlem} will be proven in the next section. Here, assuming Lemma \ref{mainlem} is true, the proof of Theorem \ref{maintheor} is given.

\begin{proof}[Proof of Theorem \ref{maintheor}]

    \begin{figure}
        \centering

        \begin{subfigure}[t]{0.3\textwidth}
            \includegraphics[scale=0.45]{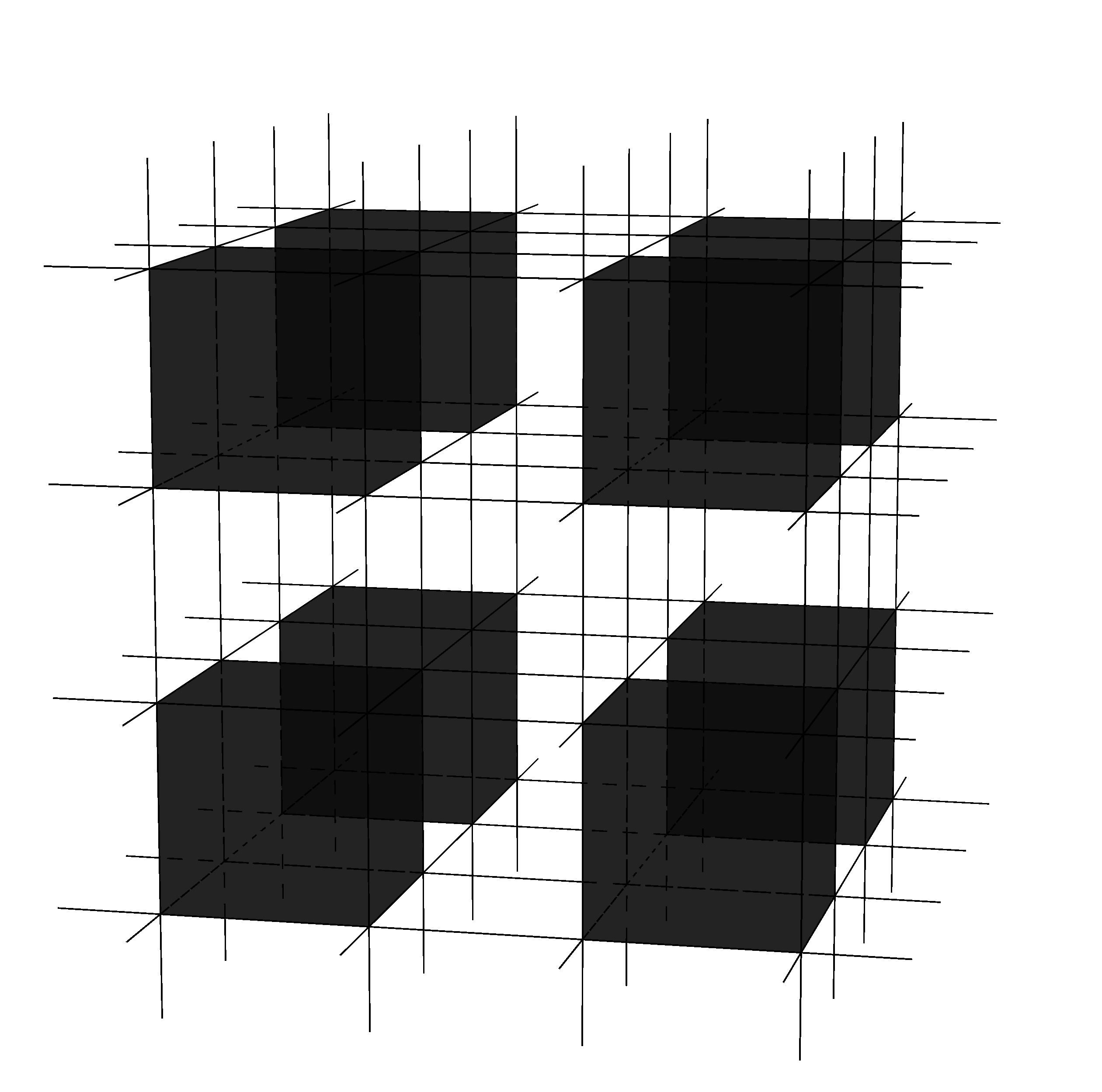}
            \caption{Groups of $A$ tensors. Each black cube represents one group.}
        \end{subfigure}
        \hspace{0.01\textwidth}
        \begin{subfigure}[t]{0.3\textwidth}
            \includegraphics[scale=0.45]{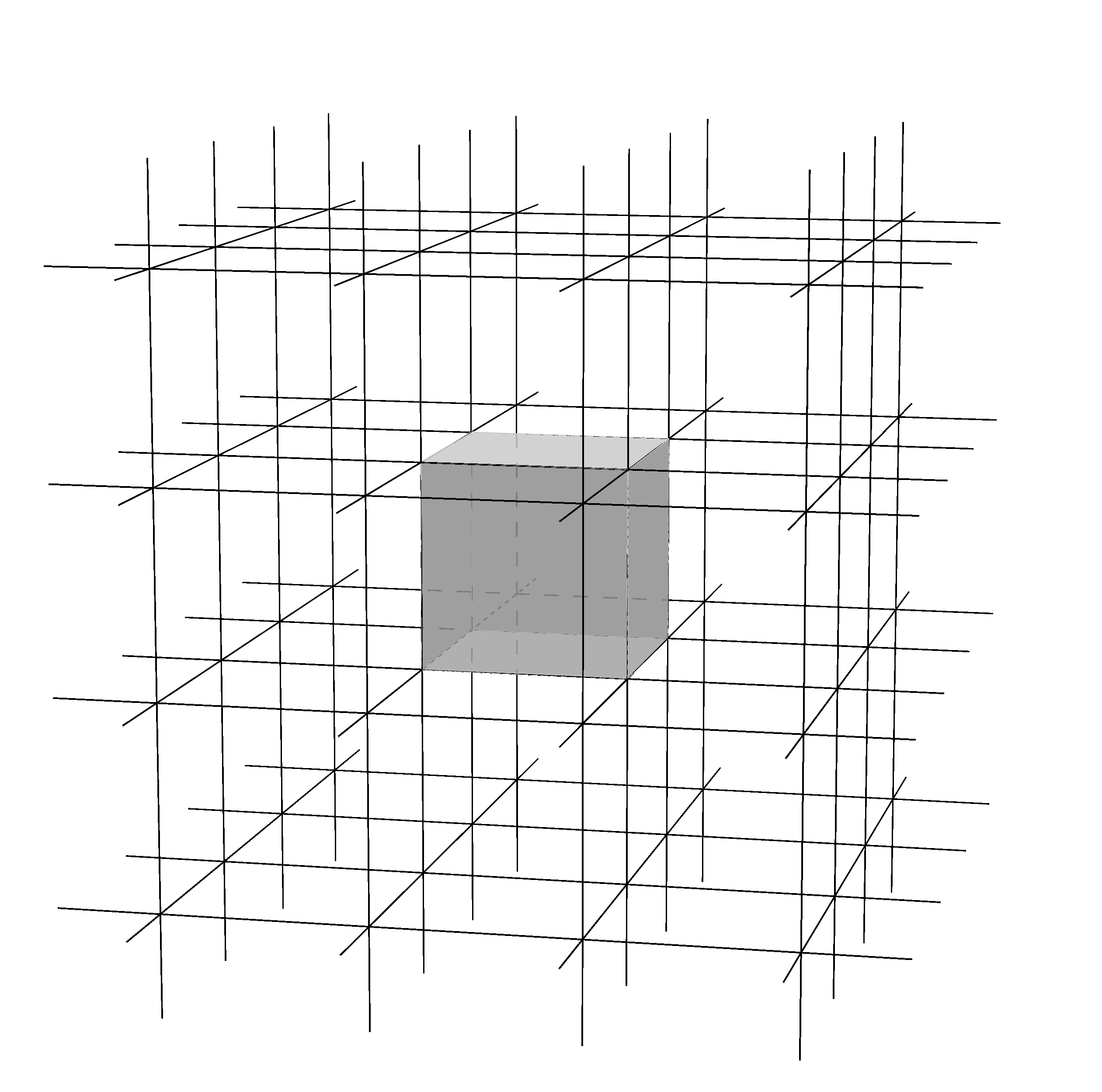}
            \caption{The grey cube represents a group of $B_v$ tensors. Other grey cubes are not depicted because they are partially outside of the given part of the network}
        \end{subfigure}
        \hspace{0.01\textwidth}
        \begin{subfigure}[t]{0.3\textwidth}
            \includegraphics[scale=0.45]{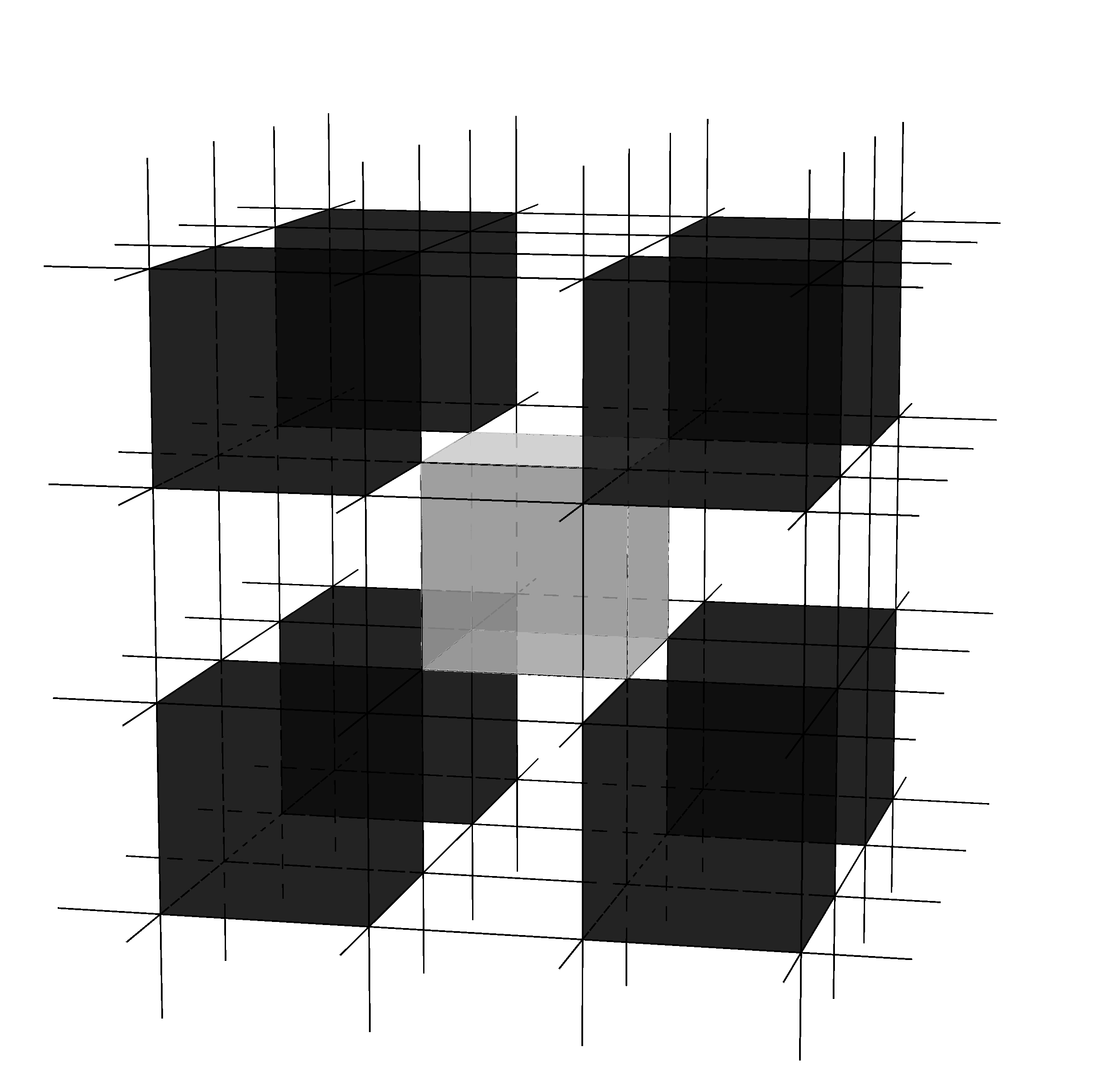}
            \caption{Black cubes that touch the grey one.}
        \end{subfigure}

        \caption{A part of the network with groupings of tensors for the proof of Theorem \ref{maintheor}. Tensor labels $A, B_v$ are suppressed to make the picture more readable}
        \label{fig:grouping}
    \end{figure}

    As noticed at the end of Section \ref{cornstr}, to prove the theorem, it is enough to construct a $b=2$ RG map $A \mapsto A'$ such that conditions \eqref{maincond} are satisfied under the additional assumption that $\deltaA$ has corner structure \eqref{corn}.

    Consider a tensor network of size $N$ built of $A=A_*+\deltaA$, where $\deltaA$ has corner structure \eqref{corn}. Assume that $N$ is even. Gather tensors $A$ in groups of $8$ tensors, such that each group form a cube on the lattice, and these cubes (referred to as black cubes) are arranged in a lattice of the size $N/2$ as in \figref{fig:grouping} (a). Each such block cube may be viewed as the tensor $T$. Replace $A$ tensors within the block cubes by $B_v$ from Lemma \ref{mainlem} using \eqref{maineq}. The result is a new tensor network of the same size and with the same partition function.

    Next, group $B_v$ tensors into other cubes (referred to as grey cubes) arranged in the lattice of size $N/2$, like the black ones, but this lattice is shifted such that each grey cube touches $8$ black cubes as in \figref{fig:grouping} (c). Each such group may be viewed as a cubic contraction $U$ given by:
    \begin{equation}\label{Utens}
        U=\includegraphics[scale=1, valign=c]{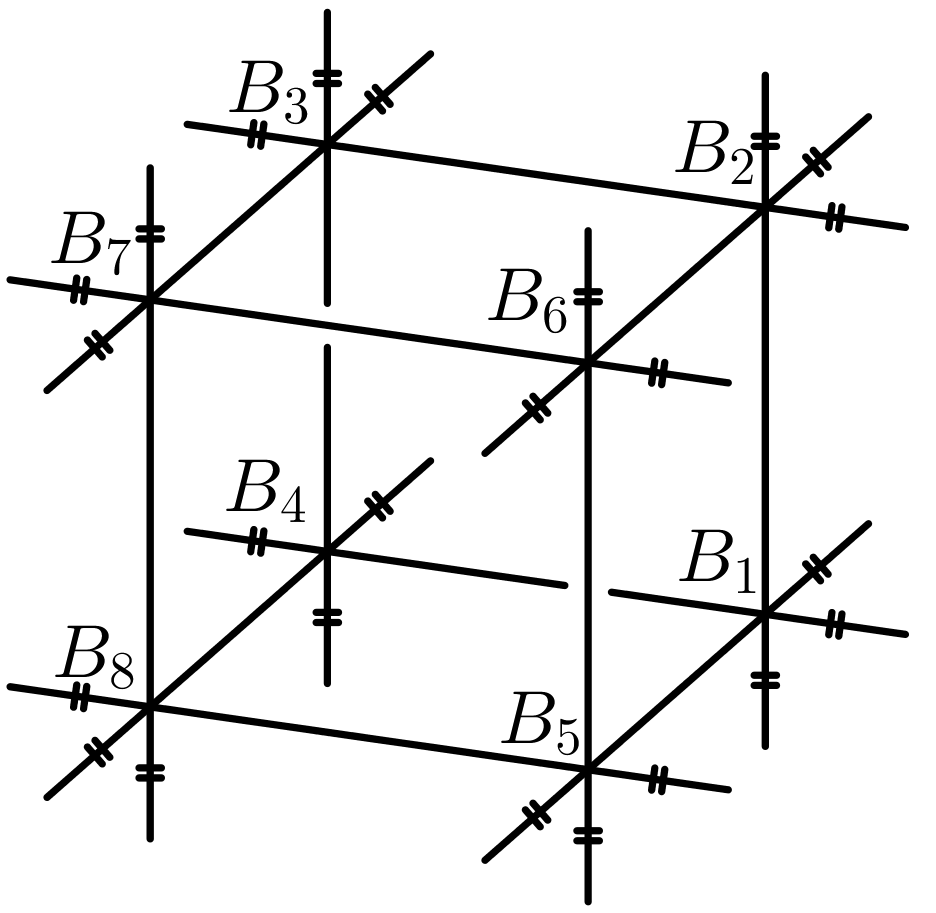}.
    \end{equation}
    Note that $U$ is made of the same tensors $B_v$ as $T'$, but with a different assignment to the vertices of the cubic contraction. The most important feature is that $T'-$external legs of $B_v$ tensors are  $U-$internal.

    Next, group the legs of $U$ in the same way as it was done with legs of $T$ in the simple RG step example in Section \ref{cornstr}. Thus, $U$ may be viewed as a $6$-leg tensor indexed by $\mathcal{D}_0^4$. As mentioned earlier, we will define $\mathcal{D}_0$ in Section~\ref{D0}. For now, let us assure the reader that $\mathcal{D}_0$ is countable. Then, let $\rho$ be a bijection from $\N_0$ to $\mathcal{D}_0^4$ satisfying \eqref{0t0}. Define $A'$ by reindexing of $U$ as follows:
    \begin{equation}\label{reindA}
        A'_{i_1 j_1 i_2 j_2 i_3 j_3 } = U_{\rho(i_1) \rho(j_1) \rho(i_2 )\rho(j_2) \rho(i_3) \rho(j_3) }.
    \end{equation}
    The map $A \mapsto A'$ is called the \textbf{rearrangement RG step}. By construction, it is an RG map with lattice rescaling factor $2$. It will now be shown that $A'$ satisfies \eqref{maincond} (provided that $\epsilon$ is sufficiently small), thus proving the theorem.

    Substitute $B_v = A_* + \deltaB_v$ into $U$ and expand the result in the number of insertions of $\deltaB_v$ tensors. Let $T_*$ be the cubic contraction with $0$ insertions of $\deltaB_v$, i.e., the cubic contraction with $A_*$ at every vertex. Let $\deltaU_n, n\geq 1$ be the sum of all cubic contractions, appearing in the expansion of $U$, with $n$ insertions of $\deltaB_v$ tensors. For reference, there are $8$ terms in $\deltaU_1$, $48$ terms in $\deltaU_2$, and so on, with only one term in $\deltaU_8$. Let $\deltaU$ be the sum of all $\deltaU_n$ tensors:
    \begin{equation}\label{dUdec}
        \deltaU = \deltaU_1 +\deltaU_2 + \ldots + \deltaU_8.
    \end{equation}
    Thus,
    \begin{equation}\label{Udec}
        U=T_*+\deltaU.
    \end{equation}
    Note that $\deltaU_{(0,0,0,0)\ldots (0,0,0,0)}$ is not necessarily zero.

    Let us show that $\deltaU$ is $O(\epsilon^\gamma)$ with $\gamma >1$, namely:
    \begin{equation}\label{Uord}
        \deltaU = O(\epsilon^{\min(2a,b)}).
    \end{equation}
    Consider $\deltaU_1$:
    \begin{equation}\label{du1}
        \deltaU_1 =  \includegraphics[scale=1,valign=c]{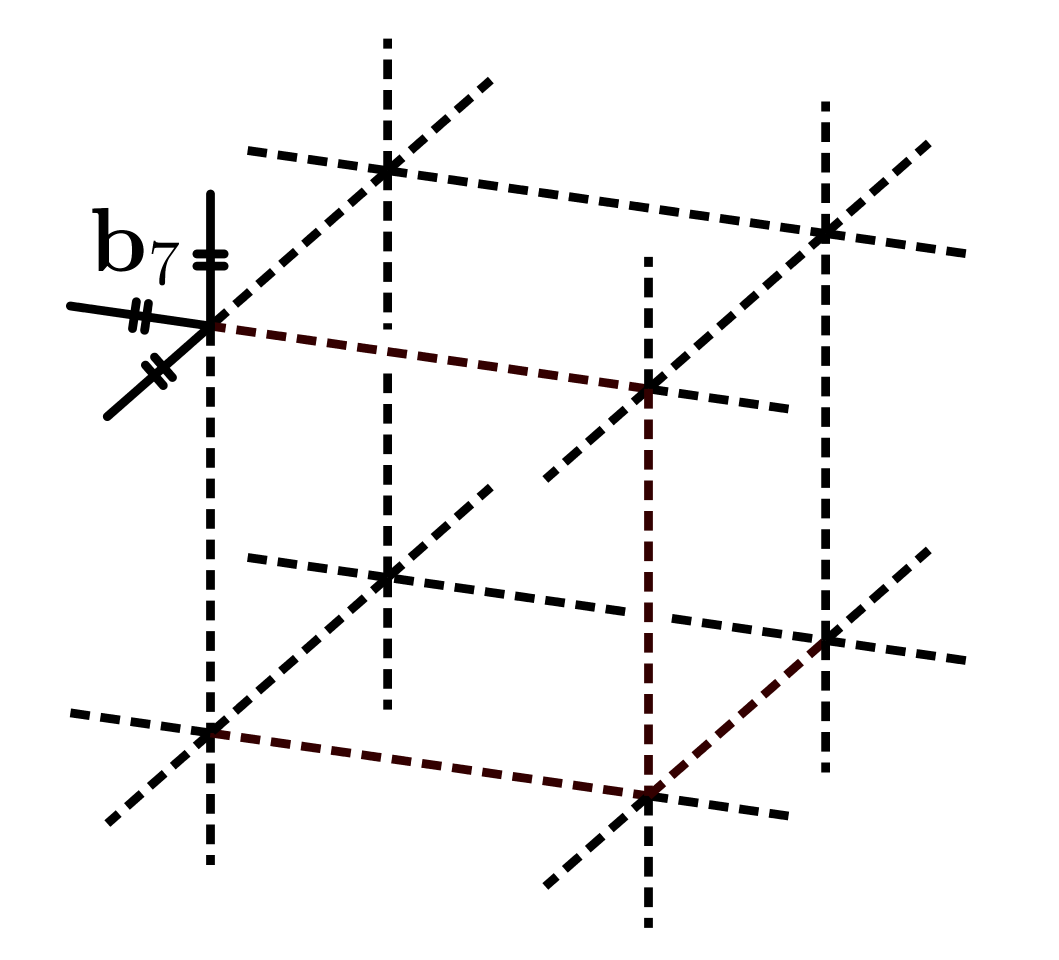} + \ldots,
    \end{equation}
    where "$\ldots$" stands for the other cubic contractions with a single $\deltaB_v$ insertion. In the contractions in \eqref{du1}, tensors $\deltaB_v$ are contracted with copies of $A_*$ by all their legs which are $T'\text{-external}$ ($U$-internal). Thus, only the restrictions of the $\deltaB_v$ tensors with $T'\text{-external}$ legs restricted to $\{0\}$ appear in these contractions. Then, Property \eqref{P2} from Lemma \ref{mainlem} implies that each of these contractions is $O(\epsilon^b)$. Therefore,
    \begin{equation}\label{U1est}
        \deltaU_1 = O(\epsilon^b).
    \end{equation}
    As for $\deltaU_n, n\geq 2$, Property \eqref{P1} from Lemma \ref{mainlem} implies that
    \begin{equation}\label{U28est}
        \deltaU_n = O(\epsilon^{na}), n\geq 2.
    \end{equation}
    Equations \eqref{U1est} and \eqref{U28est} imply \eqref{Uord}.

    Using \eqref{Uord}, one can prove that $A'$ defined by \eqref{reindA} satisfies \eqref{contr}. Consider the normal decomposition of $A'$:
    \begin{eqnarray}\label{willBeusefull}
        A' = \z' (A_*+\deltaA').
    \end{eqnarray}
    Equations \eqref{reindA} and \eqref{0t0} imply that $\z'=U_{(0,0,0,0) \ldots (0,0,0,0)}$. Note that \eqref{Uord} implies:
    \begin{equation}\label{nuord}
        \z' = 1+O(\epsilon^{\min(2a,b)}).
    \end{equation}

    As pointed out in Section \ref{cornstr}, reindexing does not change the norm of a tensor. Hence, \eqref{Uord}, \eqref{Udec} and \eqref{0t0} imply that
    \begin{equation}\label{Apord}
        A' = A_*+O(\epsilon^{\min(2a,b)}).
    \end{equation}
    Bounds \eqref{Apord} and \eqref{nuord} imply that, for sufficiently small $\epsilon$,
    \begin{equation}\label{dApbound}
        \deltaA' = O(\epsilon^{\min(2a,b)}).
    \end{equation}
    This is precisely \eqref{contr} with $h=\min (2a,b)$.

    It is left to show that \eqref{analit} is satisfied by $A'$. Note that $A'$ depends analytically on $\deltaA$  as it is a composition of analytic maps: tensors $B_v$ depend analytically on $\deltaA$ by Property \eqref{P0} from Lemma \ref{mainlem}; $U$ is analytic as it is a contraction of $B_v$ tensors; reindexing is analytic as it is a bounded linear map. Therefore $\z'$ is analytic, and $\deltaA'$ is analytic (if $\epsilon$ is sufficiently small, so that $\z'$ is nonzero, as it is necessary to divide $A'$ by $\z'$ to define $\deltaA'$). Thus, \eqref{analit} is satisfied.

    Therefore, $A'$ satisfies \eqref{maincond}. Hence, Theorem \ref{maintheor} is proven.

\end{proof}

\section{Proof of Lemma~\ref{mainlem}}\label{prf}

Let us begin by discussing the notion of a template, which is extensively used throughout the proof.

\subsection{Templates}\label{tmplandD}

For this discussion, it is convenient to introduce the notion of the \textbf{restriction of a bond}. Let us provide an example. Consider $C$, a contraction of two $6$-tensors $A$ and $B$, given by \eqref{contraction}. Let $\mathcal{J}$ be another index set. Then \eqref{contraction} can be written as
\begin{equation}\label{restrb}
    C_{i_1\ldots i_{11}} = \sum_{j\in  \mathcal{I} \cap \mathcal{J}} A_{ji_1\ldots i_5} B_{i_6 j i_7 \ldots i_{11}}+\sum_{j\in \mathcal{I}\setminus \mathcal{J}} A_{ji_1\ldots i_5} B_{i_6 j i_7 \ldots i_{11}}.
\end{equation}
Denote the first sum in \eqref{restrb} by $\bar{C}$. We say that $\bar{C}$ is obtained from $C$ by restricting the bond between $A$ and $B$ to $\mathcal{J}$. In the graphical notation, a restriction of a bond will be depicted analogously to the restriction of a leg, i.e., either by label $\mathcal{J}$ near the line as here:
\begin{equation}
    \bar{C}=\includegraphics[scale=0.7,valign=c]{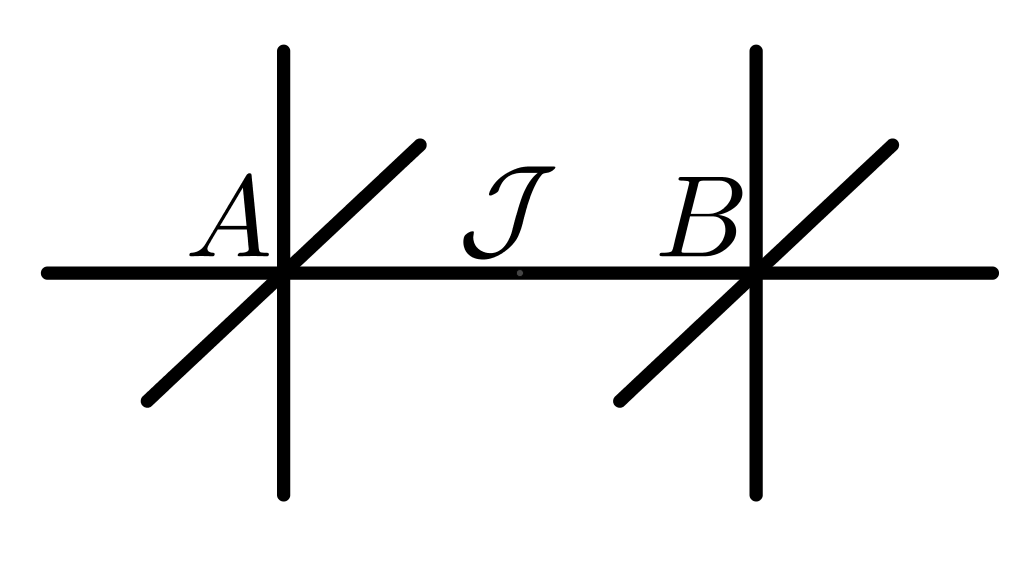},
\end{equation}
or by using different types of lines, such as a dashed line for restriction to $\{0\}$, and a ticked line for restriction to $\N$.

Now, consider the cubic contraction $T$ given by \eqref{TTens} and appearing in the l.h.s. of \eqref{maineq}. Label all legs and bonds of this contraction by numbers $1,2,\ldots,36$ in an arbitrary but fixed order. For definiteness, we choose the order as follows:
\begin{equation}
    \includegraphics[scale=0.33,valign=c]{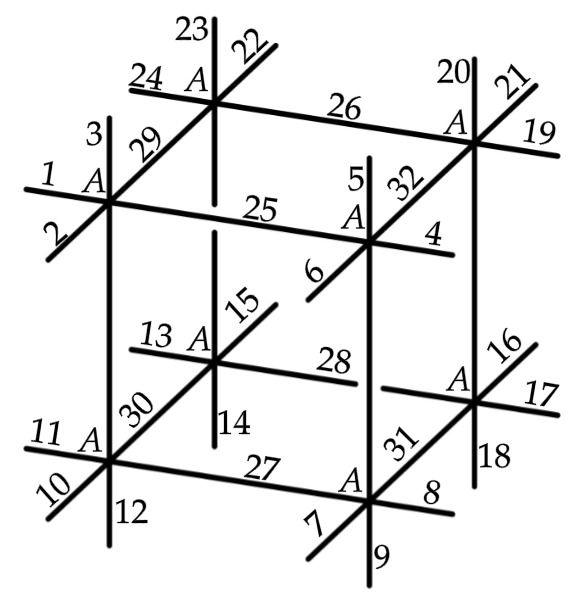}.
\end{equation}
Let $\gamma$ be a map from $\{1,\ldots,36\}$ to $\{\{0\}, \N\}$. We refer to such a map as a \textbf{template}. The set of all templates is denoted $\mathrm{Tmpl}$.

Using templates, we expand $T$ as:
\begin{equation}\label{Tdec}
    T=\sum_{\gamma \in \mathrm{Tmpl}} \includegraphics[scale=0.33,valign=c]{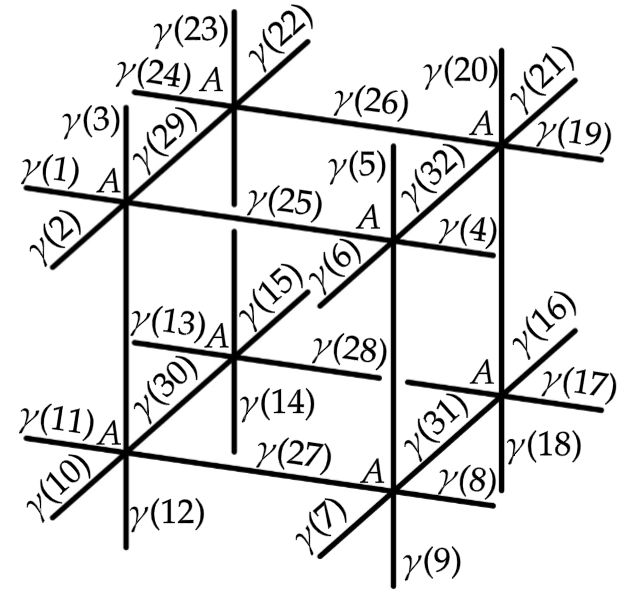}.
\end{equation}
This expansion may look frightening, but it simply states that $T$ is equal to the sum of all possible restrictions of itself with legs and bonds restricted to $\{0\}$ or $\N$.

We will prove Lemma~\ref{mainlem} by carefully analysing expansion \eqref{Tdec}. For this, we will now divide $\mathrm{Tmpl}$ into disjoint sets, which will require different treatments in our analysis.

As the first step of the division of $\mathrm{Tmpl}$, we distinguish the trivial template. The trivial template $\tau$ is the template that assigns $\{0\}$ to each bond and leg of $T$. We call "non-trivial" all other templates and denote the set of non-trivial templates by $\mathrm{Tmpl}'$.

\hypertarget{targ:diags}{Further division relies on the graphical representation of templates $\gamma$ by diagrams $\diag(\gamma)$ in which we make a bond or leg number $n$ dashed or ticked depending on whether $\gamma(n)=\{0\}$ or $\N$. For example, the following diagram corresponds to the template $\gamma$ with $\gamma(n)=\N$ for $n=3,29,6,20,35$,  and $\gamma(i)=\{0\}$ for all other $n$.
    \begin{equation}\label{templateEX}
        \diag(\gamma)=\includegraphics[scale=1,valign=c]{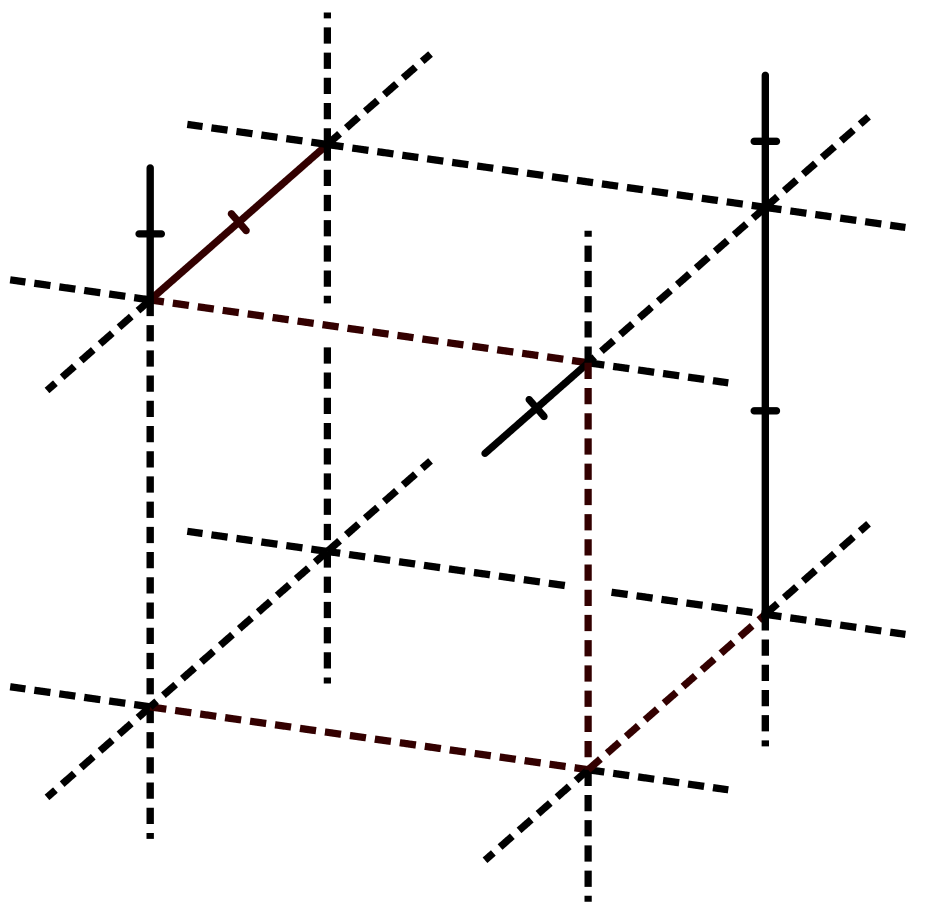}
    \end{equation}
    Note that the correspondence between templates $\gamma$ and diagrams $\diag(\gamma)$ is one-to-one.}

\hypertarget{targ:graphs}{We also need the notion of the internal graph of a non-trivial template. It is defined as follows. Consider $\diag(\gamma)$, where $\gamma \in \mathrm{Tmpl}'$. Let $V$ be the set of vertices of the cube where at least one ticked line of $\diag(\gamma)$ ends. Let $E$ be the set of ticked bonds in $\diag(\gamma)$.
    \begin{dfn}\label{intgraphdef}
        The \textbf{internal graph} of a non-trivial template $\gamma$, denoted by $\graph (\gamma)$, is defined as the graph with vertices $V$ and edges $E$. Note that $\graph (\gamma) \neq \void$.
    \end{dfn}
    For instance, the internal graph of the template $\gamma$ from \eqref{templateEX} is highlighted here in light blue:
    \begin{equation}
        \includegraphics[scale=1,valign=c]{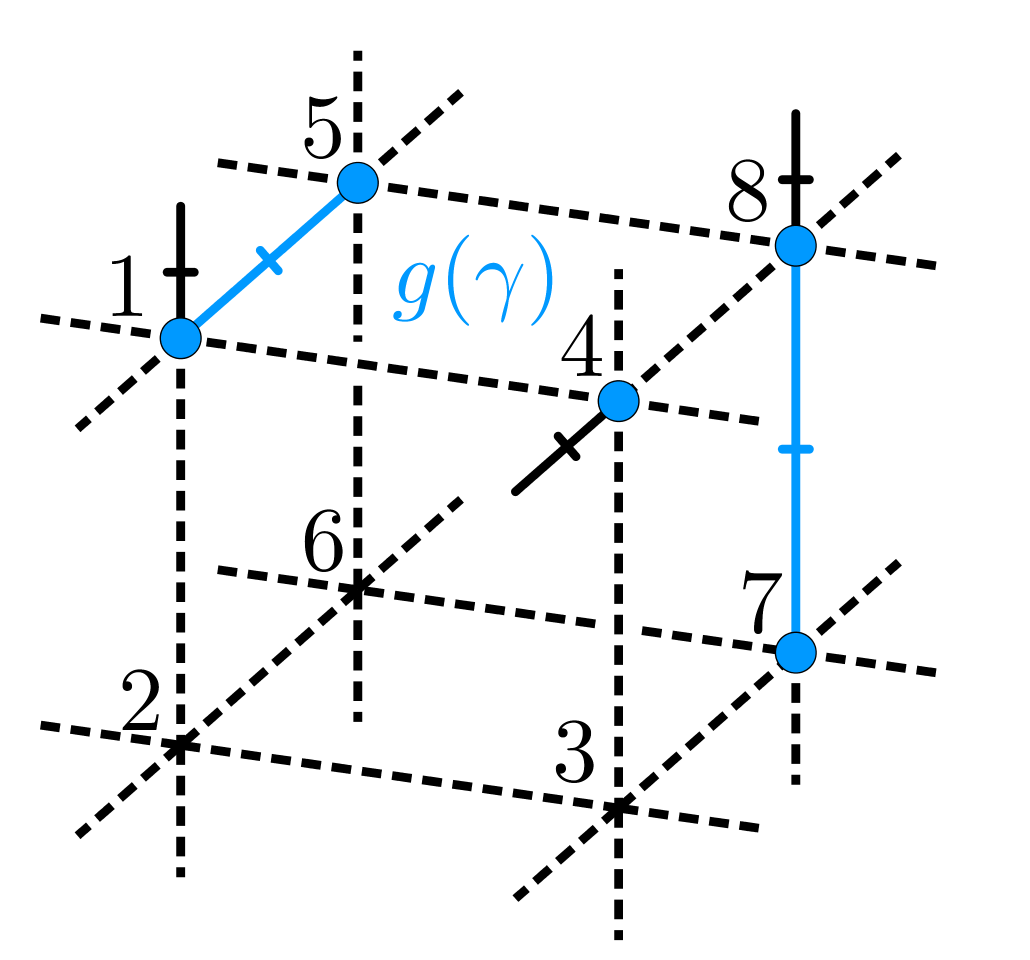}.
    \end{equation}
    It has $5$ vertices, $V=\{1,4,5,7,8\}$, and two edges $E=\{(1,5), (7,8)\}$.}

A non-trivial template is called \textbf{connected} (resp. \textbf{disconnected}) if its internal graph is connected (resp. disconnected). We denote the set of all connected templates by $\mathrm{Con}$ and the set of all disconnected templates by $\mathrm{Disc}$.

We will express each tensor $B_v$ as a sum over connected templates. There are two types of connected templates, which give rise to different terms in these sums. The difference between these types of templates is in the presence of sources, which are defined as follows.

Let $\gamma \in \mathrm{Con}$. We say that a vertex $v \in\graph (\gamma)$ is a \textbf{source} of $\gamma$ if at least one of the following two conditions holds:
\begin{equation} \label{source1}
    \text{a) there is a ticked leg\footnotemark coming out of } v,
\end{equation}
\footnotetext{Recall that we distinguish between legs and bonds. A bond is a pair of contracted legs.}
or
\begin{equation} \label{source2}
    \text{b) } v \text{ is a vertex of degree } 1 \text{ in }\graph (\gamma).
\end{equation}
If $v$ is not a source of $\gamma$, we call it a \textbf{sink} of $\gamma$. The sets of connected templates with and without sources are denoted by $\mathrm{Src}$ and $\mathrm{NSrc}$, respectively.

Here is an example. Consider the template $\gamma$ given by the following diagram:
\begin{equation}\label{diag_gamma}
    \diag(\gamma)=\includegraphics[valign=c,scale=1]{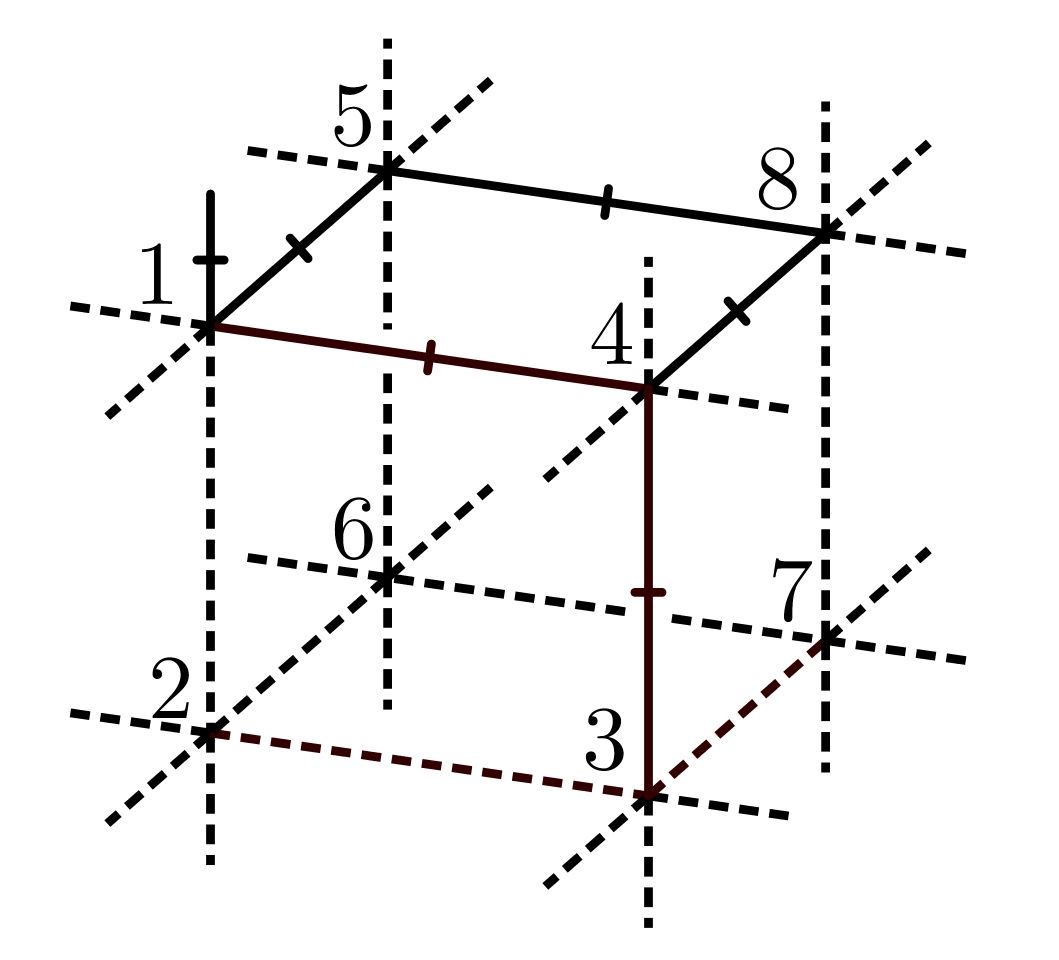}.
\end{equation}
Vertices $1$ and $3$ are sources of this template due to \eqref{source1} and \eqref{source2} respectively, and vertices $4,5,8$ are sinks.

\tikzset{every picture/.style={line width=0.75pt}} %set default line width to 0.75pt        

In this way, we obtain the division of templates into subsets summarised in the following diagram:
\begin{equation}\label{classif}
    \includegraphics[scale=0.3,valign=c]{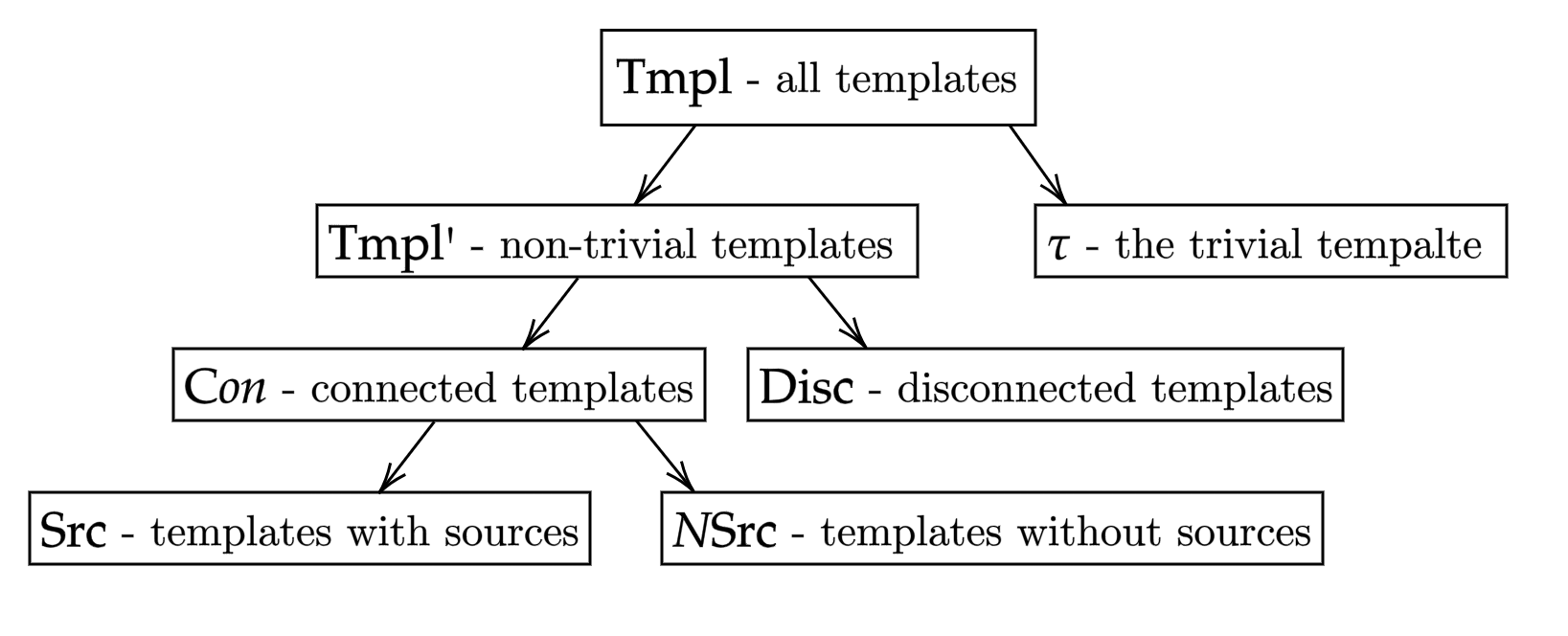}
\end{equation}
Note in particular that $\mathrm{Tmpl}$ is the following disjoint union:
\begin{equation}
    \mathrm{Tmpl}=\mathrm{Src} \sqcup \mathrm{NSrc}  \sqcup \mathrm{Disc} \sqcup \{\tau\}.
\end{equation}

To handle disconnected templates, we also introduce the notion of the union of templates. Consider templates $\gamma$ and $\mu$ such that:
\begin{equation}\label{uncond}
    \graph (\gamma) \cap\graph (\mu) = \void.
\end{equation}

\begin{dfn}\label{def:union}
    If \eqref{uncond} holds, the \textbf{union} of $\gamma$ and $\mu$, denoted $\gamma \star \mu$, is defined as the template given by the following formula:
    \begin{equation}\label{uni_def}
        \gamma \star \mu (n)=\begin{cases}
            \N,    & \text{if } \gamma(n)=\N \text{ or } \mu(n)=\N \\
            \{0\}, & \text{if } \gamma(n)=\mu(n)=\{0\}
        \end{cases}.
    \end{equation}
    If \eqref{uncond} does not hold, the union is not defined.
\end{dfn}
Here is an example, which by slight abuse of notation, we provide in terms of diagrams:
\begin{equation}\label{unionExample}
    \includegraphics[valign=c,scale=0.8]{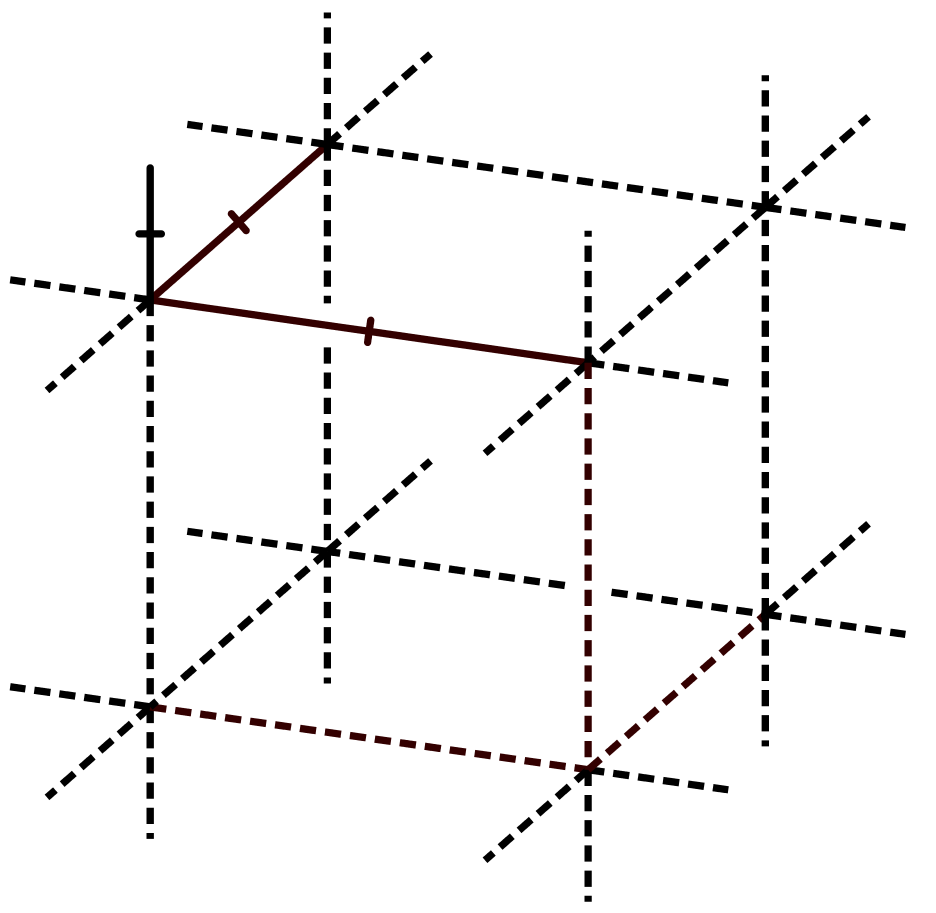} \star \includegraphics[valign=c,scale=0.8]{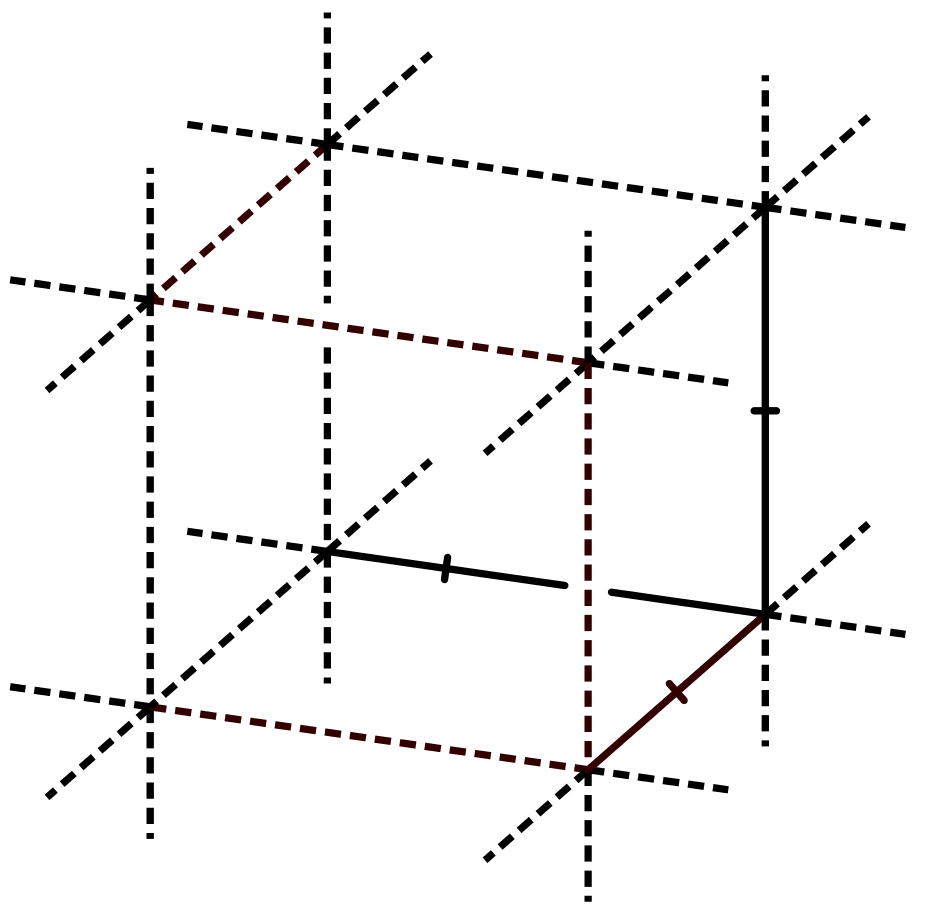}= \includegraphics[valign=c,scale=0.8]{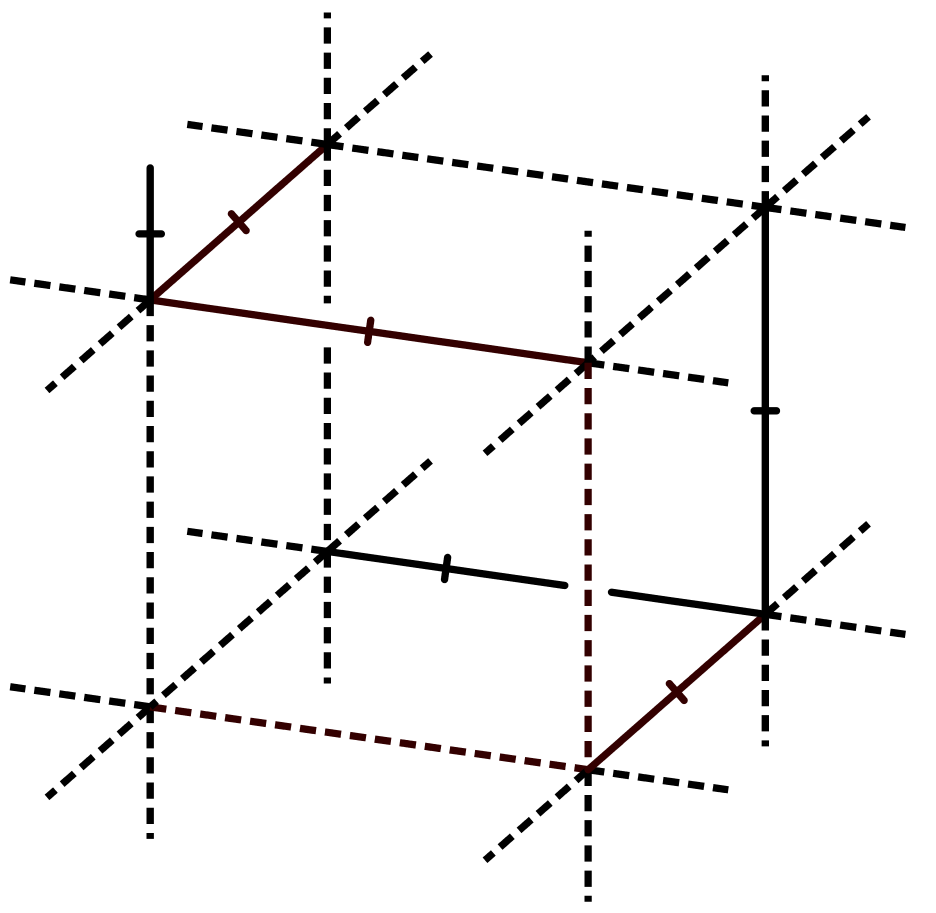}.
\end{equation}
Note that any disconnected template $\gamma$ can be uniquely written as the union of connected templates $\gamma_1,\ldots,\gamma_m$. We will refer to templates $\gamma_k$ as \textbf{connected components} of $\gamma$.

\subsection{Index set \texorpdfstring{$\mathcal{D}_0$}{D0}}\label{D0}

Recall that our goal is to construct tensors $B_v$ which solve Eq.~\eqref{maineq} and satisfy additional properties given in Lemma \ref{mainlem}. As indicated in Eq.~\eqref{maineq} by lines with two ticks, some legs of these tensors will take values in the new index set.\footnote{These are $T'$-internal legs of tensors $B_v$, using the notation from Section \ref{rRG}. See \eqref{Tprime} and the discussion below.} We denote this set by $\mathcal{D}_0$ and define it as follows.

For each $\gamma \in \mathrm{Src}$, we define a copy of $\N$ denoted by $\N_\gamma$. For each $\gamma \in \mathrm{NSrc}$, we define two distinct single-element index sets: one with the element denoted by $0_\gamma$ and another with the element denoted by $1_\gamma$. \footnote{If we wanted to be pedantic, we could define $\N_\gamma=\{(\gamma ,i), i\in \N\}, \ 0_\gamma =(\gamma,0), \ 1_\gamma =(\gamma, 1)$. \label{footNg}} All $\{0_\gamma\}, \{1_\gamma\}, \N_\gamma$ are meant to be disjoint from each other and from $\{0\}, \N$. We then define:
\begin{subequations}
    \begin{align}
         & \mathcal{O}=\{0\} \sqcup \bigsqcup_{\gamma \in \mathrm{NSrc}} \{0_\gamma\};                                                    \\
         & \mathcal{D}= \N \sqcup \bigsqcup_{\gamma \in \mathrm{Src}} \N_\gamma \sqcup \bigsqcup_{\gamma \in \mathrm{NSrc}} \{1_\gamma\}; \\
         & \mathcal{D}_0 = \mathcal{O} \sqcup \mathcal{D}.
    \end{align}
\end{subequations}

We update the graphical notation to work with tensors indexed by $\mathcal{D}_0$. Until now, a dashed line meant restriction to $\{0\}$ and a ticked line to $\N$. From now on, a dashed line will mean restriction to $\mathcal{O}$, and a ticked line will mean restriction to $\mathcal{D}$. Note that this new convention is not in conflict with the old one as restricting a tensor indexed by $\N_0$ to $\mathcal{O}$ (resp. $\mathcal{D}$) is equivalent to restricting it to $\{0\}$ (resp. $\N$).

Note that in the discussion about templates, we explicitly used sets $\{0\}$ and $\N$. However, it is clear that the definition of templates we gave in Section~\ref{tmplandD} can be generalised for use with any pair of disjoint sets in place of $\{0\}$ and $\N$. From this point onwards, we will work with "extended" templates $\gamma$, which have values $\gamma(i) \in \{\mathcal{O}, \mathcal{D}\}$. We think about set $\mathcal{O}$ (resp. $\mathcal{D}$) as an extension of $\{0\}$ (resp. $\N$).

Let us go through the main points of Section~\ref{tmplandD} to show that everything works smoothly after this extension:
\begin{itemize}
    \item Equation \eqref{Tdec} is still valid due to the equivalence of restrictions to $\{0\}, \N$ and $\mathcal{O}, \mathcal{D}$ for tensors indexed by $\N_0$.
    \item Diagrams $\diag(\gamma)$ consist of dashed and ticked lines, in agreement with the new graphical conventions introduced earlier (dashed line for $\mathcal{O}$, ticked for $\mathcal{D}$). They are indistinguishable from the diagrams of templates before extension.
    \item The trivial template $\tau$ now assigns $\mathcal{O}$ to each bond.
    \item Definitions that relied on the diagrams (internal graph $\graph (\gamma)$, sources and sinks) remain the same since diagrams $\diag(\gamma)$ are the same.
    \item Division of templates \eqref{classif} is the same since it relied on the diagrams $\diag(\gamma)$.
    \item In the definition of union, we replace \eqref{uni_def} by
          \begin{equation}\label{uni_def_new}
              \gamma \star \mu (n)=\begin{cases}
                  \mathcal{D}, & \text{if } \gamma(n)=\mathcal{D} \text{ or } \mu(n)=\mathcal{D} \\
                  \mathcal{O}, & \text{if } \gamma(n)=\mu(n)=\mathcal{O}
              \end{cases}.
          \end{equation}
\end{itemize}

We now have all the tools to proceed to the proof of Lemma~\ref{mainlem}.

\subsection{Preliminary remarks and the plan of the proof}\label{preps}

We will use the notation $T(\{\deltaB_v\})$ to represent the cubic contractions of tensors $B_v=A_*+\deltaB_v$, positioned as in \eqref{Tprime}. When all tensors $\deltaB_v$ are equal to $\deltaA$, we write $T(\deltaA)$. In this notation, Eq.~\eqref{maineq} can be expressed as:
\begin{equation}\label{maineq1}
    T(\{\deltaB_v\})=T(\deltaA).
\end{equation}
We recall that $T(\deltaA)$ has been expressed by \eqref{Tdec}. The contraction $T(\{\deltaB_v\})$ is given by the analogous equation with $A$ replaced by $B_v$.\footnote{Note that after the extension of templates in Section~\ref{D0}, the expansion \eqref{Tdec} is valid when $\deltaB_v$ is indexed by any $\mathcal{I} \subseteq \mathcal{D}_0$.} We will denote the individual terms in this expansion by $T_\gamma (\{\deltaB_v\})$ ($T_\gamma(\deltaA)$ when all $\deltaB_v=\deltaA$). In this notation, \eqref{maineq1} becomes:
\begin{equation}\label{maineq2}
    \sum_{\gamma \in \mathrm{Tmpl}}T_\gamma (\{\deltaB_v\})=\sum_{\gamma \in \mathrm{Tmpl}} T_\gamma (\deltaA).
\end{equation}

\begin{rmk}\label{orderconvention}
    Note that the arguments of $T$ and $T_\gamma$ are \emph{labelled} sets of tensors: each $\deltaB_v$ is labelled by the vertex $v$, which tells us where it should be inserted in the cubic contraction. Because of this labelling, the order in which $\deltaB_v$’s are listed does not matter.
\end{rmk}

Let us remark on the structure of individual terms in \eqref{maineq2}. Consider terms $T_\gamma(\deltaA)$ from the r.h.s. of \eqref{maineq2}. For illustrative purposes, we provide a particular example where $\gamma$ is given by the following diagram:
\begin{equation}\label{gamma_generic}
    \diag(\gamma) = \includegraphics[valign=c,scale=1]{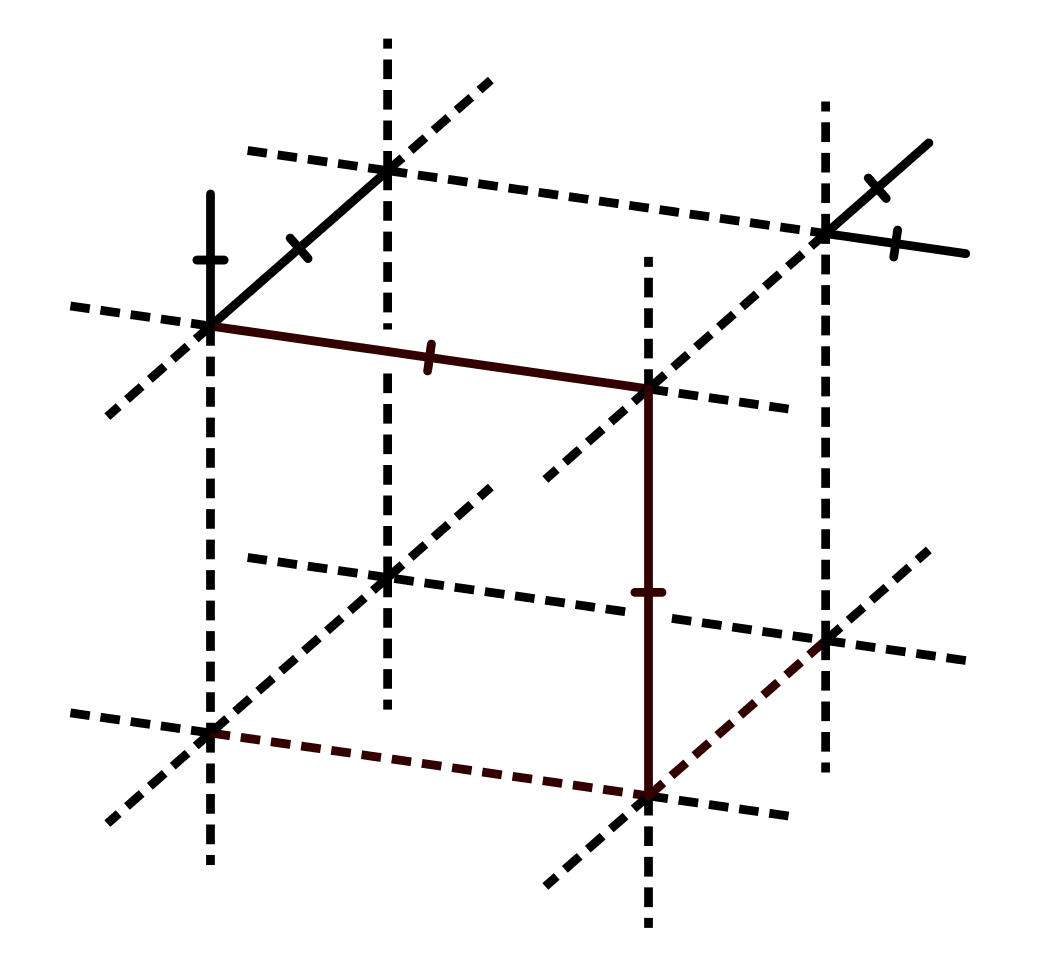}.
\end{equation}
The corresponding term is
\begin{equation}\label{Texp_generic}
    T_\gamma(\deltaA) = \includegraphics[valign=c]{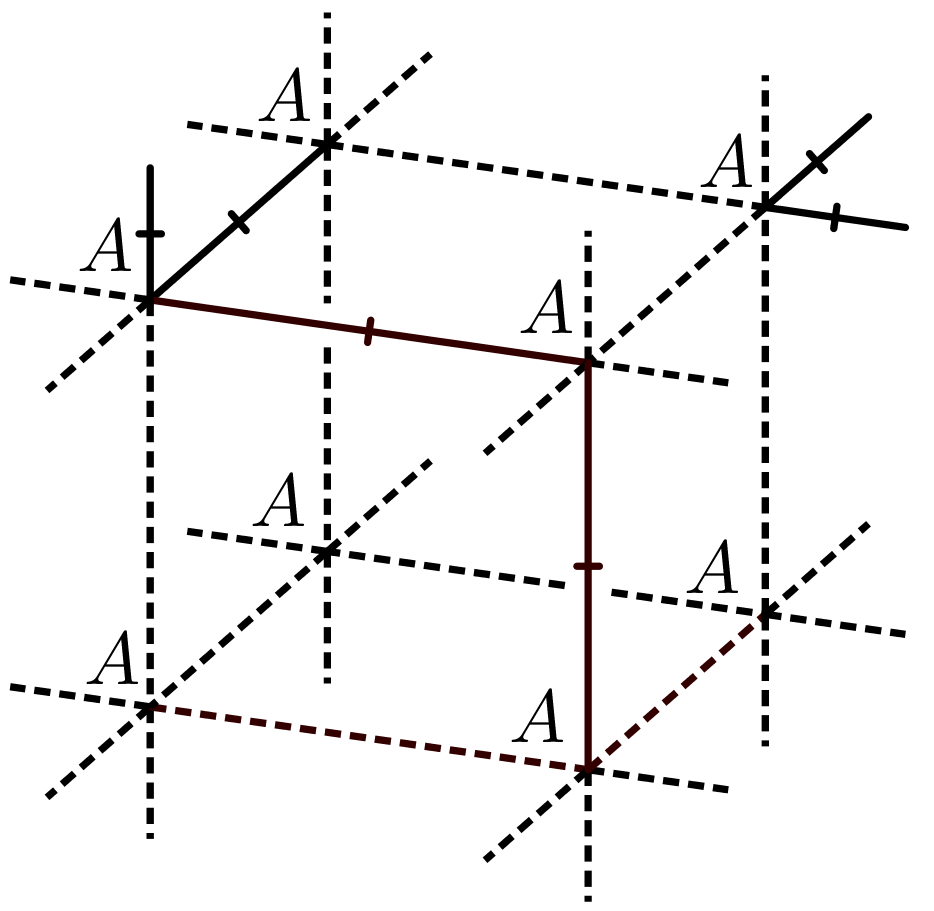}.
\end{equation}
Note that the restriction of $A$ with all legs restricted to $\mathcal{O}$ is $A_*$. At the same time, a restriction of $A$ with at least one leg restricted to $\mathcal{D}$ is equal to the restriction of $\deltaA$ with legs restricted in the same way.\footnote{Note that $A$ is indexed by $\N_0$. Thus, restrictions to $\mathcal{O}$ and $\mathcal{D}$ are equivalent to restrictions to $\{0\}$ and $\N$, respectively.} Consequently, at each vertex of $T_\gamma (\deltaA)$, we can replace $A$ with $A_*$ if the vertex is surrounded by dashed lines or replace it with $\deltaA$ otherwise. Below, we often make such replacements without an explicit mention. Here is \eqref{Texp_generic} after this replacement ($\deltaA$ is represented by the red disk, as shown in \eqref{solid}):
\begin{equation}\label{rep_restr}
    T_\gamma(\deltaA) = \includegraphics[valign=c]{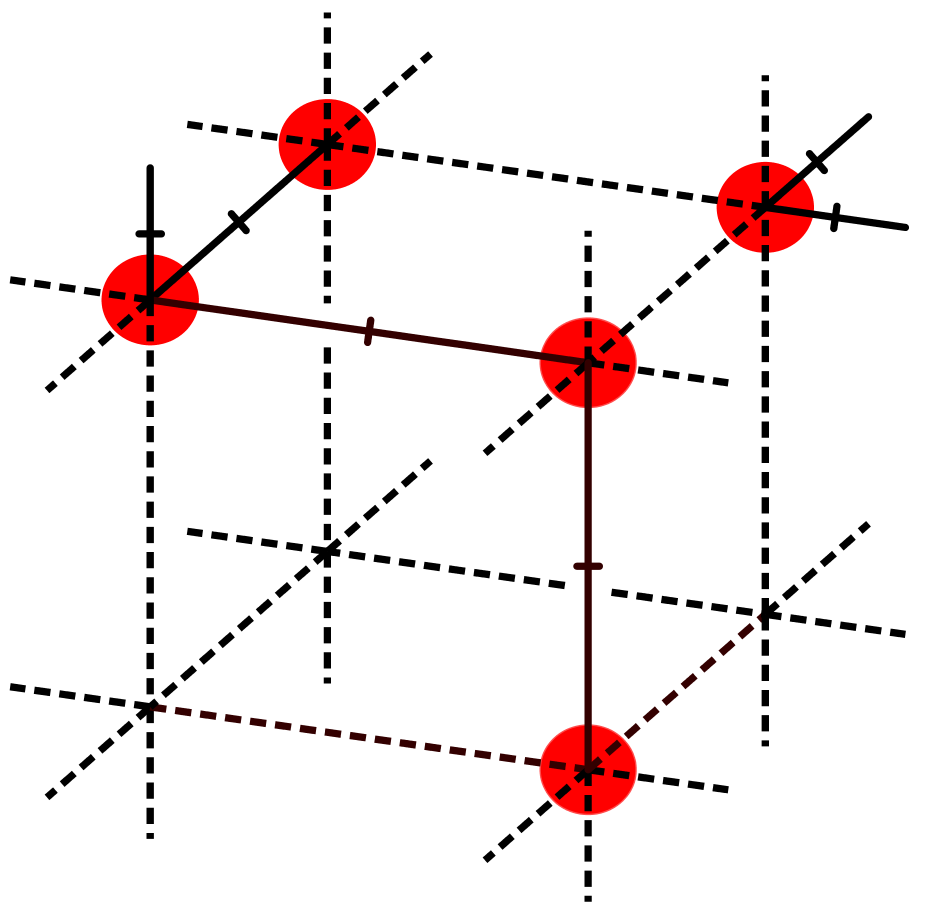}.
\end{equation}
In this way, we observe that each $T_\gamma(\deltaA)$ is the cubic contraction in which:
\begin{subequations}\label{trm_dscrA}
    \begin{align}
         & \text{Tensors $\deltaA$ are inserted in vertices $v \in\graph (\gamma)$.}\label{trm_dscrA1} \\
         & \text{Tensors $A_*$ are inserted in vertices $v \notin\graph (\gamma)$.}\label{trm_dscrA2}  \\
         & \text{Bonds and legs are restricted according to $\gamma$.}\label{trm_dscrA3}
    \end{align}
\end{subequations}
Here and below, "restricted according to $\gamma$" refers to the restriction as in \eqref{Tdec}. We will use this expression for cubic contractions, as in \eqref{trm_dscrA3}, and for individual $6$-tensors. For the latter, we say that a $6$-tensor is restricted according to $\gamma$ around a vertex $v$ if its legs are restricted in agreement with dashed ($\mathcal{O}$) and ticked ($\mathcal{D}$) lines touching the vertex $v$ in $\diag(\gamma)$. If a tensor is labelled by a vertex $v$ (see Remark~\ref{orderconvention}), we simply say that the tensor is restricted according to $\gamma$, meaning that it is restricted according to $\gamma$ around the vertex $v$.

Let us make an analogous discussion for $T_\gamma(\{\deltaB_v\})$. In this paper, we only consider tensors $\deltaB_v$ that satisfy the following condition:
\begin{equation}\label{dbvOr}
    \text{the restriction of $\deltaB_v$ with all legs restricted to $\mathcal{O}$ is $0$.\footnotemark}
\end{equation}
\footnotetext{Such a condition is automatically satisfied by $\deltaA$ since it is indexed by $\N_0$ and $\deltaA_{000000}=0$, see \eqref{ATens}.} Assuming this condition, $T_\gamma(\{\deltaB_v\})$ is (analogous to $T_\gamma(\deltaA)$) the cubic contraction in which:
\begin{subequations}\label{trm_dscr}
    \begin{align}
         & \text{Tensors $\deltaB_v$ are inserted in vertices $v \in\graph (\gamma)$.}\label{trm_dscr1} \\
         & \text{Tensors $A_*$ are inserted in vertices $v \notin\graph (\gamma)$.}\label{trm_dscr2}    \\
         & \text{Bonds and legs are restricted according to $\gamma$.}\label{trm_dscr3}
    \end{align}
\end{subequations}
Here is an example of $T_\gamma(\{\deltaB_v\})$ for the template $\gamma$ given by \eqref{gamma_generic}:
\begin{equation}
    T_\gamma(\{\deltaB_v\})=\includegraphics[valign=c]{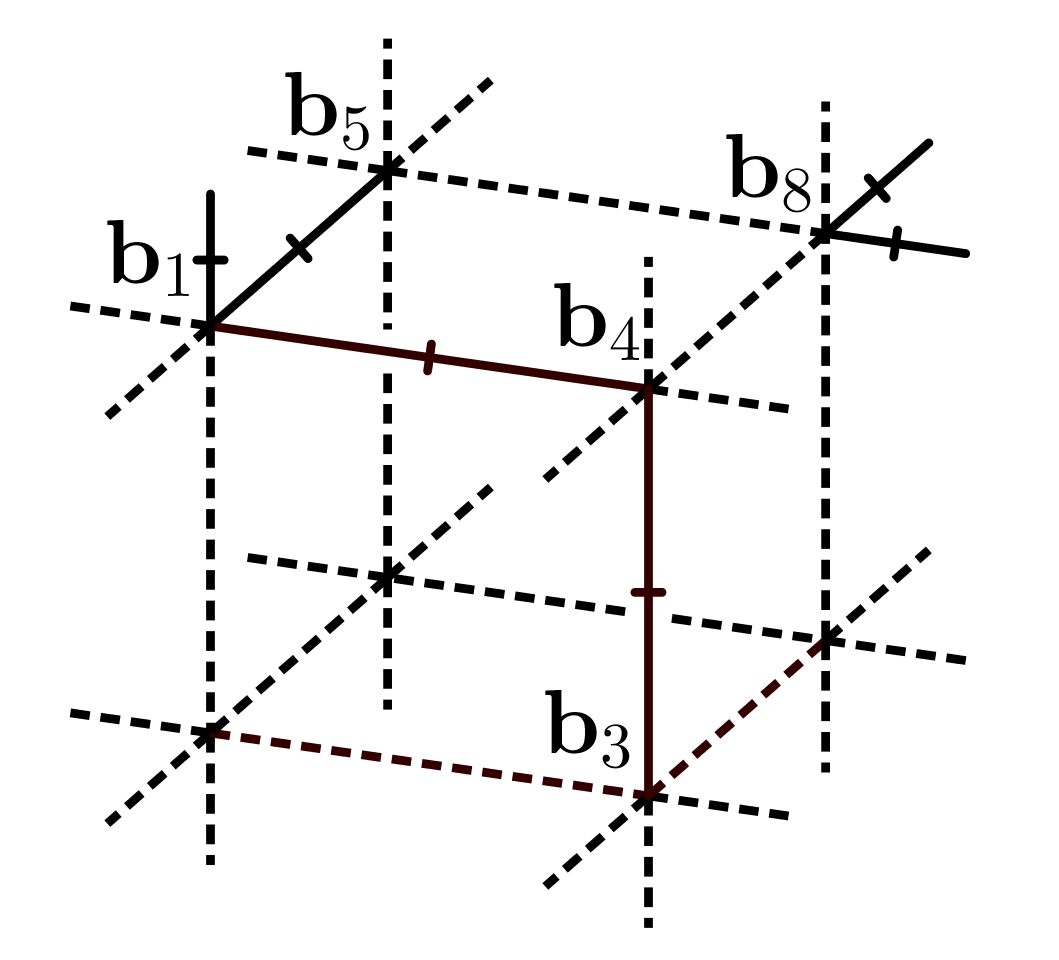}.
\end{equation}

In agreement with the above discussion, we have:
\begin{equation}\label{triv}
    T_\tau (\{\deltaB_v\})= T_*, \qquad T_\tau(\deltaA)= T_*,
\end{equation}
where $\tau$ is the trivial template and $T_*$ denotes the cubic construction of $8$ $A_*$ tensors.

After these preliminary remarks, we delve into the analysis of the main equation \eqref{maineq2}. Given the tensor $\deltaA$, our goal is to find tensors $\deltaB_v$ satisfying \eqref{maineq2} and \eqref{Pconds}. We will look for these in the form
\begin{equation}\label{dbv12dec}
    \deltaB_v = \deltaB_v^1 +\deltaB_v^2,
\end{equation}
where:
\begin{itemize}
    \item The tensors $\deltaB_v^1$ solve the following equation:
          \begin{equation}\label{Src_sect}
              \sum_{\gamma \in \mathrm{Tmpl}} T_\gamma(\{\deltaB_v^1\}) = T_*+\sum_{\gamma \in \mathrm{Src}} T_\gamma(\deltaA) +\sum_{\gamma \in \mathrm{DiscS}} T_\gamma (\deltaA).
          \end{equation}
          Here and below, $\mathrm{DiscS} \subset \mathrm{Disc}$ denotes the set of disconnected templates which are unions (see Def.~\ref{def:union} and Eq.~\eqref{uni_def_new}) of connected templates with sources, i.e.,
          \begin{equation}\label{Discs_def}
              \mathrm{DiscS}=\{\gamma_1 \star \ldots \star \gamma_m \mid \gamma_k \in \mathrm{Src}, m \geq 2 \}.
          \end{equation}
          Note the difference between \eqref{Src_sect} and \eqref{maineq2}. In \eqref{Src_sect}, the non-trivial contributions on the r.h.s. are limited to those coming from templates $\gamma \in \mathrm{Src}$ (see \eqref{classif}) and their unions. In \eqref{maineq2}, all templates contribute to the r.h.s..
    \item The tensors $\deltaB_v^2$ are correction terms that ensure $\deltaB_v$ resolves \eqref{maineq2}.
\end{itemize}

We will ensure that both $\deltaB_v^1$ and $\deltaB_v^2$ satisfy \eqref{Pconds}, and hence their sum $\deltaB_v$ will also satisfy \eqref{Pconds}. In particular, we will achieve \eqref{P1} and \eqref{P2} using two different methods for $\deltaB_v^1$ and $\deltaB_v^2$. Let us provide a coarse-grained picture of these methods.

In both cases, we will first find tensors $\deltaB_v^\gamma$ that fulfill \eqref{P1}, \eqref{P2}, and analogues of the following equation:
\begin{equation}\label{randlabl}
    T(\{\deltaB_v^\gamma\}) = T_\gamma(\deltaA),
\end{equation}
where $\gamma \in \mathrm{Src}$ for the $\deltaB_v^1$ case and $\gamma \in \mathrm{NSrc}$ for the $\deltaB_v^2$ case. Then, we will define $\deltaB_v^1$ and $\deltaB_v^2$ as sums of the corresponding $\deltaB_v^\gamma$ tensors.\footnote{For brevity, certain technical details are omitted here.} The difference between the $\deltaB_v^1$ and $\deltaB_v^2$ cases lies in the approach to satisfying \eqref{P1} and \eqref{P2} for $\deltaB_v^\gamma$.

In the $\deltaB_v^1$ case, we will use that $\deltaA$ restricted according to a template $\gamma \in \mathrm{Src}$ around a source vertex satisfies \eqref{P1} and \eqref{P2} even after being slightly scaled up.\footnote{We will be demonstrate this in Section~\ref{PcondsDBV1}} We will scale up $\deltaA$'s in the sources and scale down $\deltaA$'s in the sinks of $T_\gamma(\deltaA)$ in the r.h.s. of \eqref{randlabl}, ensuring that all tensors therein satisfy \eqref{P1} and \eqref{P2}. The tensors $\deltaB_v^\gamma$ will be defined as these scaled versions of $\deltaA$.

In the $\deltaB_v^2$ case, we will seek a solution of \eqref{randlabl} where only $N_\gamma-1$ of  $\deltaB_v^\gamma$ tensors are nonzero, where $N_\gamma$ is the number of vertices in $\graph (\gamma)$. Note that $T_\gamma(\deltaA)=O(\epsilon^{N_\gamma})$, and so it is possible to find $\deltaB_v^\gamma$ bounded as $O(\epsilon^{N_\gamma/(N_\gamma-1)})$, and therefore satisfying \eqref{P1} and \eqref{P2}. Our search for appropriate $\deltaB_v^\gamma$'s will be simplified dramatically by the fact that $T_\gamma(\deltaA) \propto T_*$ when $\gamma \in \mathrm{NSrc}$.

\subsection{Construction of tensors \texorpdfstring{$\deltaB_v^1$}{bv1}} \label{dbv1}

In this subsection, we find tensors $\deltaB_v^1$ satisfying \eqref{Src_sect} and \eqref{Pconds}. Our idea is to solve \eqref{Src_sect} term-by-term. We will construct $\deltaB_v^1$ such that for a non-trivial template $\gamma$:
\begin{equation}\label{dbv1_goal}
    T_\gamma(\{\deltaB_v^1\})=\begin{cases}
        T_\gamma(\deltaA), & \text{ if } \gamma \in \mathrm{Src} \sqcup \mathrm{DiscS}    \\
        0,                 & \text{ if } \gamma \notin \mathrm{Src} \sqcup \mathrm{DiscS}
    \end{cases}.
\end{equation}
Clearly, this implies \eqref{Src_sect}.

We will search for $\deltaB_v^1$ satisfying \eqref{dbv1_goal} and \eqref{Pconds} in the form:
\begin{equation}\label{dbv1exp}
    \deltaB_v^1 = \sum_{\mu \in \mathrm{Src}} \deltaB_v^\mu.
\end{equation}
We will ensure that $\deltaB_v^\mu$ tensors satisfy \eqref{Pconds}, and hence $\deltaB_v^1$ will also satisfy \eqref{Pconds}. We will achieve \eqref{dbv1_goal} by constructing $\deltaB_v^\mu$ such that:
\begin{subequations}\label{compatrel1}
    \begin{align}
        \forall \gamma \in \mathrm{Tmpl}',\mu \in \mathrm{Src}: \qquad & T_\gamma (\{\deltaB_v^{\mu}\})=\delta_{\gamma,\mu} T_\gamma (\deltaA);\footnotemark \label{compat1}                 \\
        \forall \gamma \in  \mathrm{Con}:  \qquad                      & T_\gamma \left( \{\deltaB_v^1\}\right)  =\sum_{\mu \in \mathrm{Src}} T_\gamma (\{\deltaB_v^{\mu}\});\label{compat2} \\
        \forall \gamma \in \mathrm{Disc}: \qquad                       & T_\gamma(\{\deltaB_v^1\})=\begin{cases}
                                                                                                       T_\gamma (\deltaA), & \text{if } \gamma \in \mathrm{DiscS}    \\
                                                                                                       0,                  & \text{if } \gamma \notin \mathrm{DiscS}
                                                                                                   \end{cases}.\label{restB}
    \end{align}
\end{subequations}
\footnotetext{Here, $\delta_{\gamma,\mu}$ is the Kronecker delta. Note that in this paper, we do not use Einstein's convention of summation over repeated indices. E.g. there is no summation over $\gamma$ in \eqref{compat1}}

Let us show that \eqref{dbv1_goal} follows from \eqref{compatrel1}. Let $\gamma$ be a connected template. Then, using \eqref{compat2},\eqref{compat1} we obtain:
\begin{equation}\label{water2}
    T_\gamma(\{\deltaB_v^1\}) = \sum_{\mu \in \mathrm{Src}} T_\gamma (\{\deltaB_v^\mu\})=\sum_{\mu \in \mathrm{Src}} \delta_{\gamma,\mu} T_\gamma(\deltaA)=\begin{cases}
        T_\gamma (\deltaA), & \text{ if } \gamma \in \mathrm{Src}    \\
        0,                  & \text{ if } \gamma \notin \mathrm{Src}
    \end{cases},
\end{equation}
i.e., \eqref{dbv1_goal} for connected templates $\gamma$. For disconnected $\gamma$, \eqref{dbv1_goal} follows from \eqref{restB}.

In this way, we reduced the initial problem of finding tensors $\deltaB_v^1$ satisfying \eqref{Pconds} and \eqref{Src_sect} to the problem of finding tensors $\deltaB_v^\mu, \mu \in \mathrm{Src}$ satisfying \eqref{Pconds} and \eqref{compatrel1}.

Before providing $\deltaB_v^\mu$, let us remark on the notion of internal and external legs. In Section 3, we divided the legs of $B_v$ tensors into two groups: $T'$-internal and $T'$-external, where $T'=T(\{\deltaB_v\})$. For simplicity, we will omit the prefix "$T'$-" and refer to these groups as internal and external legs. This convention will apply to all tensors whose position in a cubic contraction is specified by the index $v=1,\ldots,8$. Here, we provide the diagram with the external legs of each tensor in the contraction highlighted in light blue:
\begin{equation}\label{int_legs_tab}
    \includegraphics[valign=c]{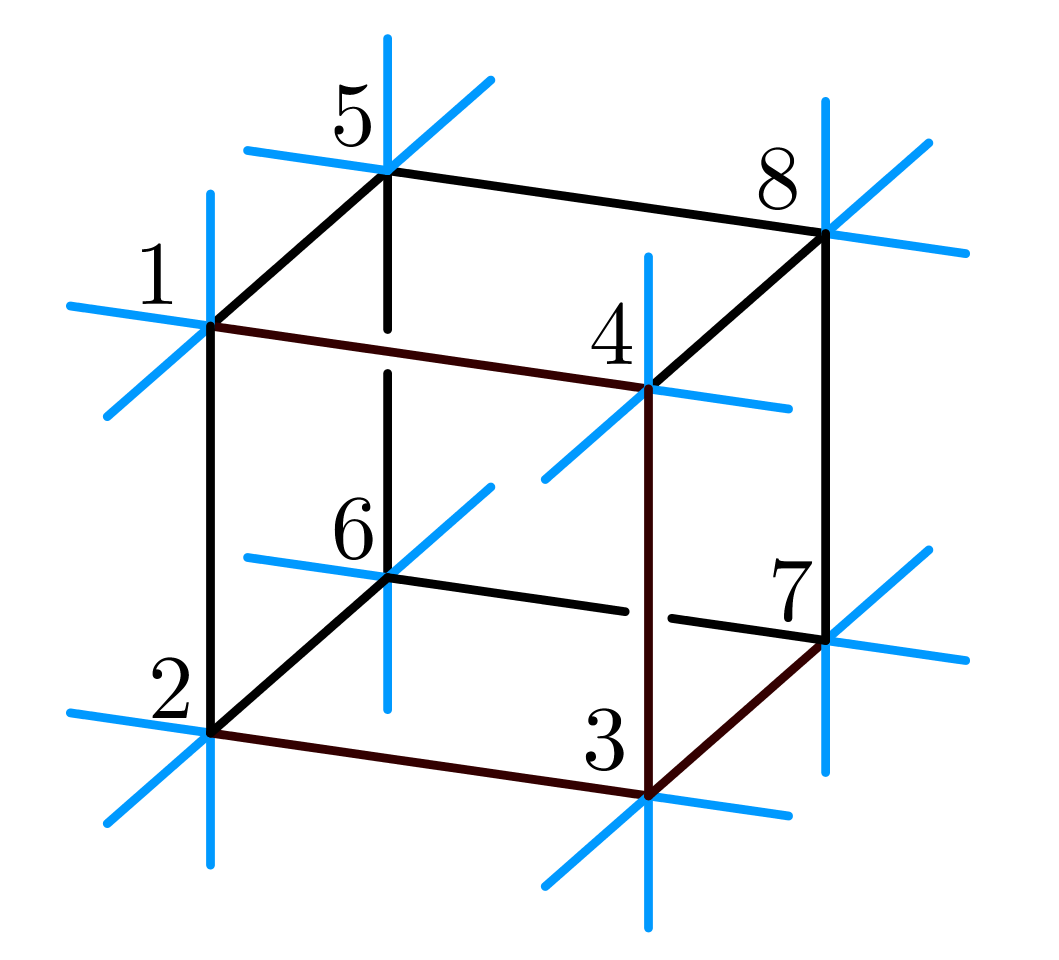}
\end{equation}

After this clarification, we proceed to the construction of $\deltaB_v^\mu$ tensors satisfying \eqref{Pconds} and \eqref{compatrel1}. We will present the construction as follows. First, using an example with a particular $\mu \in \mathrm{Src}$, we will illustrate all steps of the construction. Next, we will generalise our example and construct $\deltaB_v^\mu$ for all $\mu \in \mathrm{Src}$. Finally, we will check that all $\deltaB_v^\mu$ satisfy \eqref{Pconds} and \eqref{compatrel1}.

\subsubsection{Example}\label{example_section_dbv1}
Consider $\mu \in \mathrm{Src}$ given by the following diagram:
\begin{equation}\label{muEX}
    \diag(\mu) = \includegraphics[valign=c,scale=0.75]{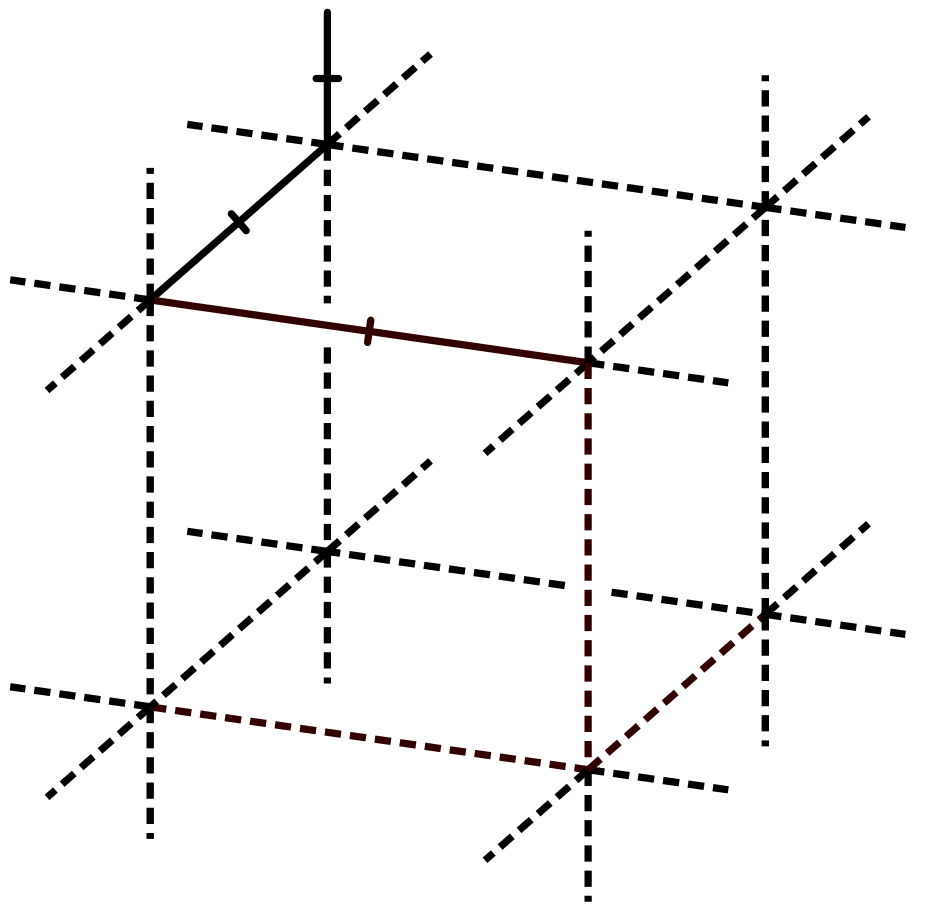}.
\end{equation}
We focus on the corresponding $T_\mu(\deltaA)$:
\begin{equation}\label{stp11}
    T_\mu(\deltaA)=\includegraphics[valign=c,scale=0.75]{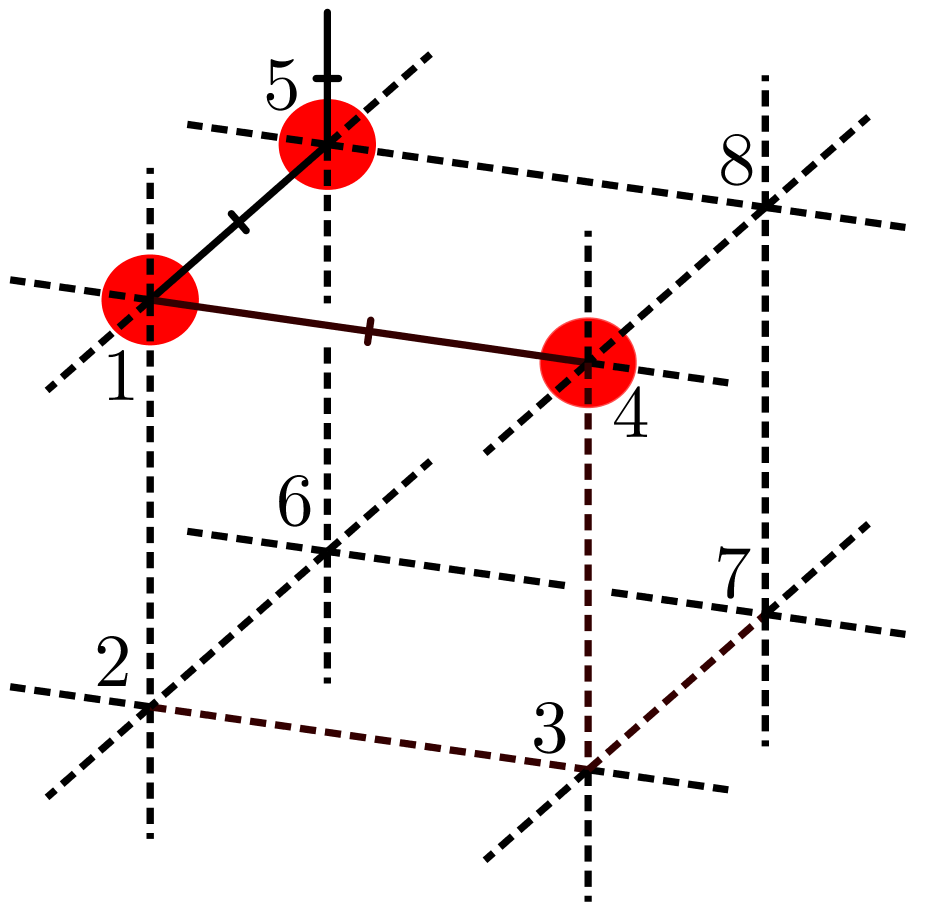}.
\end{equation}
Note that at each vertex $v \in\graph (\mu)$ in \eqref{stp11}, i.e., $v=1,4,5$, tensor $\deltaA$ has legs restricted to $\mathcal{O}$ and $\mathcal{D}$ (equivalently, to $\{0\}$ and $\N$ since $\deltaA$ is indexed by $\N_0$). Let us denote these restrictions as $\deltaA_v$:
\begin{equation}\label{stp12}
    \deltaA_1 =\includegraphics[valign=c]{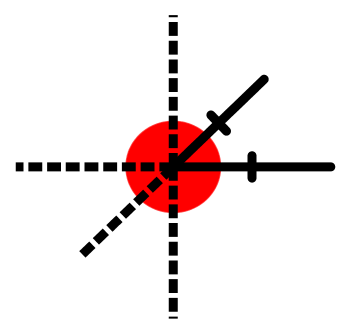} , \deltaA_4 =\includegraphics[valign=c]{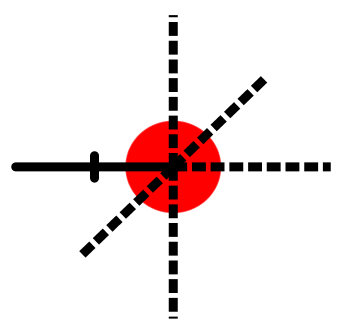} , \deltaA_5 = \includegraphics[valign=c]{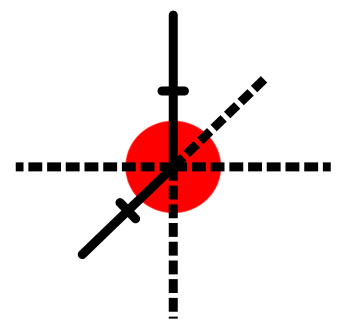}
\end{equation}
Additionally, for $v \notin\graph (\mu)$ we define:
\begin{equation}\label{zerosdav}
    \deltaA_v=0, \ v \in \{2,3,6,7,8\}.
\end{equation}

To simplify the verification of \eqref{compat1} later on, we will track the values of $T_\mu$ evaluated on the tensors from each step of the construction. For $\deltaA_v$ defined in \eqref{stp12}, \eqref{zerosdav}, we have:
\begin{equation}\label{semistep1}
    T_\mu(\{\deltaA_v\})=T_\mu (\deltaA).
\end{equation}
In particular, \eqref{zerosdav} does not affect \eqref{semistep1}, as $T_\mu$ inserts $A_*$ in vertices $v \notin\graph (\mu)$ (see \eqref{trm_dscr}).

Next, let us apply reindexing (see \eqref{reindex}) to obtain tensors $\deltaA^\mu_v$ from $\deltaA_v$. This reindexing will be crucial for the proof of \eqref{compat2}, as it will make the contraction of two $\deltaB_v^\mu$ tensors with different $\mu$ vanish when restricted to $\mathcal{D}$. We perform reindexing only on internal legs of $\deltaA_v$, replacing values $i\in \N$ with their copies from $\N_\mu$.\footnote{Recall that sets $\N_\mu$ are disjoint copies of $\N$. See Section~\ref{D0} and Footnote~\ref{footNg}.} Note that we do not replace $0$, leaving it as it is. In graphical notation, we represent tensors $\deltaA^\mu_v$ by writing $\uparrow \!\! \mu$ next to the internal ticked legs, as illustrated here:
\begin{equation}\label{shift_graph}
    \deltaA^\mu_1 = \includegraphics[valign=c, scale=0.6]{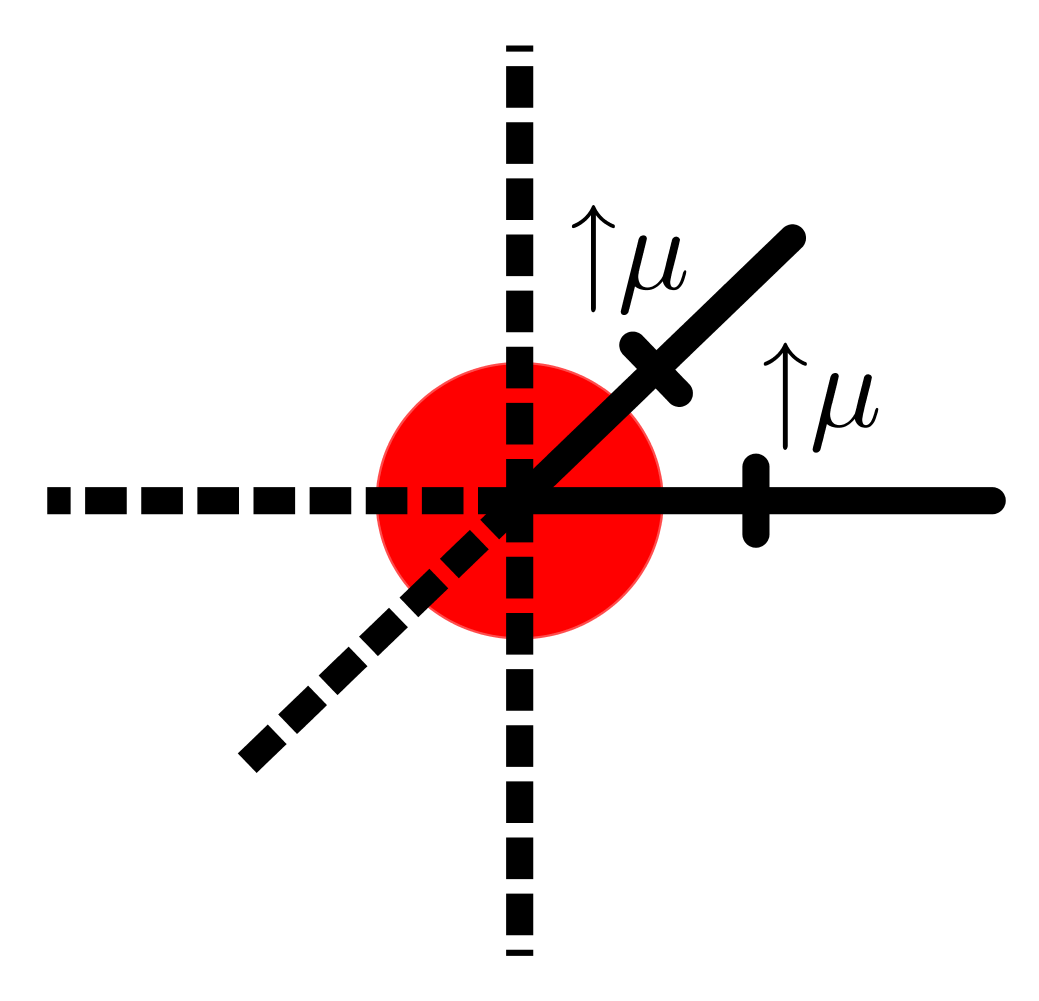} , \deltaA^\mu_4= \includegraphics[valign=c, scale=0.6]{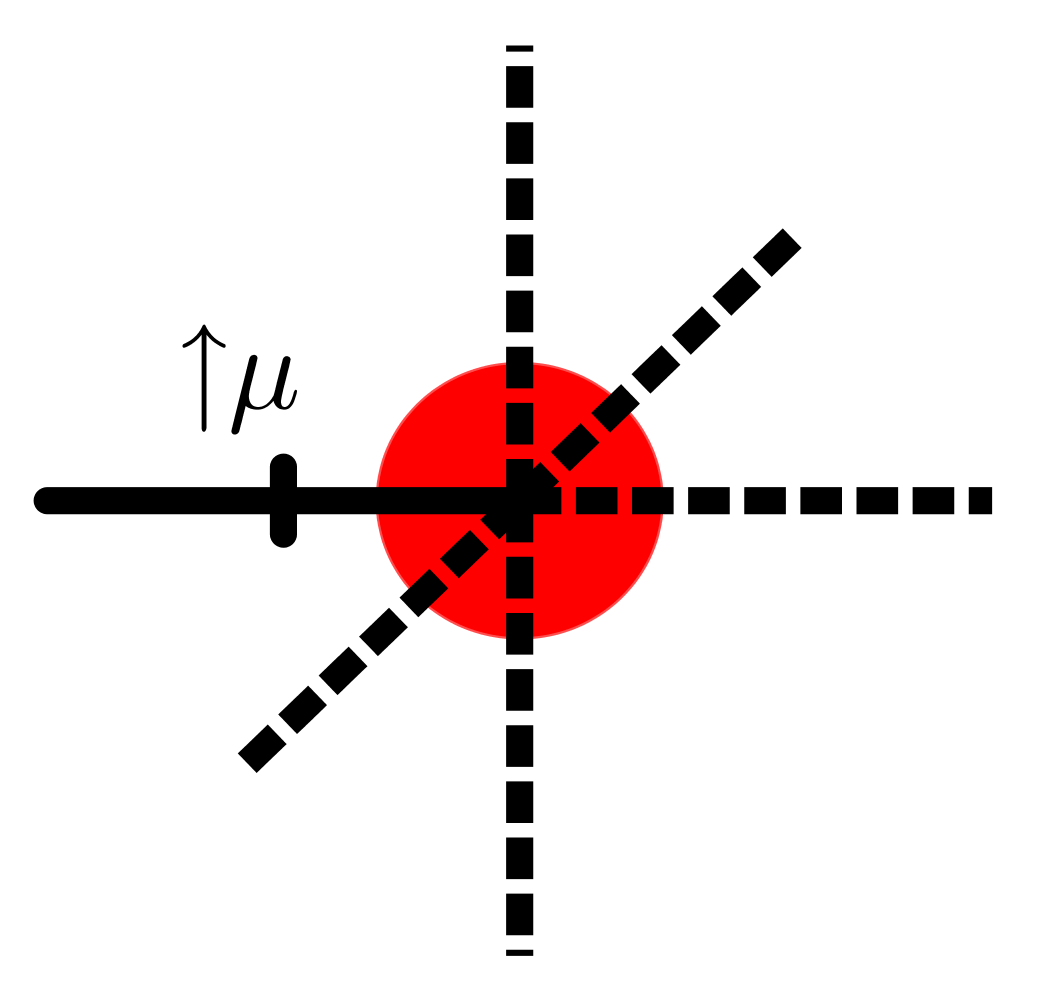}, \ \deltaA^\mu_5 =\includegraphics[valign=c, scale=0.6]{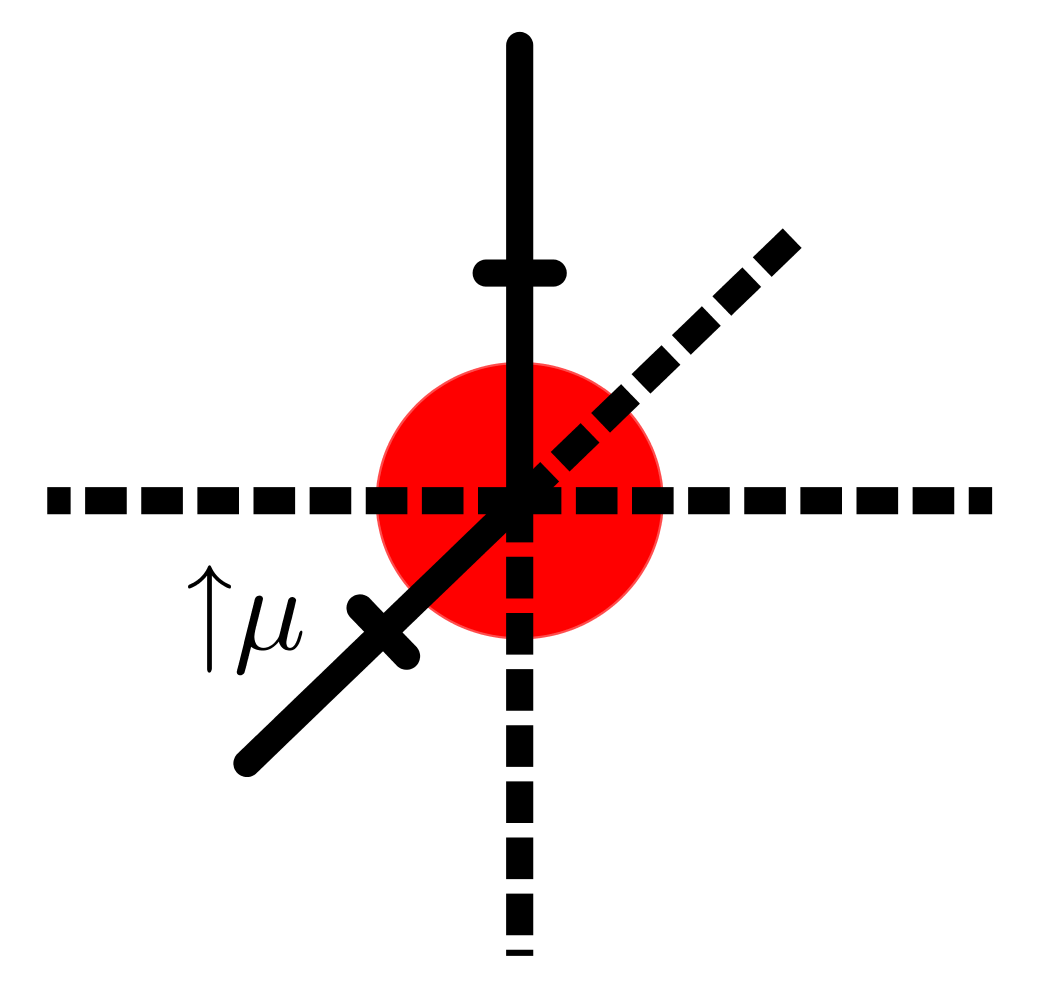}.
\end{equation}
We do not write anything next to dashed internal legs as they remain restricted to $\{0\}$.

Note that we have not changed the tensor elements of $\deltaA_v$ but have merely replaced positive integer values of the internal leg indices with elements of $\N_\mu$. Therefore,
\begin{equation}\label{semistep2}
    T_\mu (\{\deltaA^\mu_v\})=T_\mu(\{\deltaA_v\}).
\end{equation}

Finally, we define $\deltaB_v^\mu$. First, we fix $c >0$. We require $c$ to be sufficiently small, namely $c<1/2$. The significance of $1/2$ will become apparent below when we check \eqref{Pconds} for $\deltaB_v^\mu$.

Next, for sources $v=4,5$, we define:
\begin{equation}\label{dbv_src}
    \deltaB_4^{\mu} = \epsilon^{-c}  \deltaA^\mu_4, \qquad \deltaB_5^{\mu} = \epsilon^{-c}  \deltaA^\mu_5 .
\end{equation}
For the sink $v=1$, we define:
\begin{equation}\label{dbv_snk}
    \deltaB_1^{\mu} = \epsilon^{2 c} \deltaA^\mu_1.
\end{equation}
For vertices $v \notin \graph (\mu)$, we define:
\begin{equation}\label{dbv_out}
    \deltaB_v^\mu =0.
\end{equation}

In Eqs.~\eqref{dbv_src}-\eqref{dbv_out}, we assume that $\deltaB_v^\mu$ tensors have internal legs indexed by $\{0\} \sqcup \N_\mu$ and external legs indexed by $\N_0$ (otherwise \eqref{dbv_src} and \eqref{dbv_snk} does not make sense). For simplicity, we extend by zeros all tensors $\deltaB^\mu_v$ such that their internal legs can take any values in $\mathcal{D}_0$ (i.e. we add the missing tensor elements by fixing their values to be zero).

We can see that the weights $\epsilon^{-c}$ and $\epsilon^{2 c}$ cancel each other in $T_\mu(\{\deltaB_v^\mu\})$, yielding:
\begin{equation}\label{semisteps}
    T_\mu(\{\deltaB^\mu_v\}) = T_\mu(\{\deltaA^\mu_v\}).
\end{equation}
This procedure of rescaling tensors so that their contraction remains the same is called \textbf{reweighting}.

Furthermore, by combining \eqref{semistep1}, \eqref{semistep2} and \eqref{semisteps}, we obtain:
\begin{equation}\label{fullstep}
    T_\mu (\{\deltaB_v^\mu\}) = T_\mu(\deltaA).
\end{equation}

We summarise the construction of $\deltaB_v^\mu$ tensors for $\mu$ given by \eqref{muEX} in the following diagram:
\begin{equation}
    \deltaA \xrightarrow[\eqref{stp12}]{restriction} \deltaA_v \xrightarrow[\eqref{shift_graph}]{reindexing} \deltaA^\mu_v \xrightarrow[\eqref{dbv_src},\eqref{dbv_snk}]{reweighting} \deltaB_v^\mu.
\end{equation}

Let us conclude our example by showing that the tensors $\deltaB_v^\mu$ defined in \eqref{dbv_src}, \eqref{dbv_snk}, and \eqref{dbv_out} satisfy \eqref{Pconds}. For $v \notin\graph (\mu)$, conditions \eqref{Pconds} are trivially satisfied. Also, it is evident that all $\deltaB_v^\mu$ satisfy \eqref{P0} and \eqref{P3}. Consequently, we will focus on conditions \eqref{P1} and \eqref{P2} for tensors $\deltaB_v^\mu$ with $v \in\graph (\mu)$, i.e., $v=1,4,5$.

We will show that \eqref{P1} and \eqref{P2} hold for any parameters $a,b$ belonging to intervals:
\begin{subequations}\label{domains}
    \begin{align}
        1/2 <a \leq 1-c, \label{inta} \\
        1<b \leq \min (1+2c,2-c). \label{intb}
    \end{align}
\end{subequations}
The lower bounds are the conditions on $a,b$ in Lemma~\ref{mainlem}. Since $c<1/2$, both these intervals are non-empty.

We start with $\deltaB_1^\mu$ defined in \eqref{dbv_snk}. Since $\deltaA=O(\epsilon)$, we have $\deltaA_v, \deltaA^\mu_v=O(\epsilon)$, and
\begin{equation}\label{bound_b1}
    \deltaB_1^\mu=O(\epsilon^{1+2c}),
\end{equation}
which clearly satisfies \eqref{P1} and \eqref{P2} with $a,b$ as in \eqref{domains}.

Next, consider $\deltaB_5^\mu$ defined in \eqref{dbv_src}. Note that
\begin{equation}\label{bound_b5}
    \deltaB_5^\mu = O(\epsilon^{1-c}),
\end{equation}
which satisfies \eqref{P1} with $a$ as in \eqref{inta}. Condition \eqref{P2} is trivially satisfied by $\deltaB_5^\mu$ for any $b$. Indeed, $\bar{\deltaB_5^\mu}=0$ since one of the external legs of $\deltaB_5^\mu$ was already restricted to $\mathcal{D}$ (see \eqref{shift_graph}), and $\mathcal{D} \cap \{0\}= \void$.

Finally, consider $\deltaB_4^\mu$ defined in \eqref{dbv_src}. Note that $\deltaA_4$ has only one leg that takes values in $\N$ (see \eqref{stp12}).  Hence, since by assumption, $\deltaA$ has corner structure \eqref{corn}, we have $\deltaA_4, \deltaA^\mu_4=O(\epsilon^2)$, and
\begin{equation}\label{bound_b4}
    \deltaB_4^\mu=O(\epsilon^{2-c}),
\end{equation}
which satisfies \eqref{P1} and \eqref{P2} with $a,b$ as in \eqref{domains}.

In this way, we see that the constructed tensors $\deltaB_v^\mu$ satisfy \eqref{P1} and \eqref{P2} with any $a$ and $b$ within the intervals \eqref{domains}, as claimed.

\subsubsection{General construction}\label{gendbv1}

The general construction will mimic our example, except that the constant $c$ will have to be reduced.

We fix a positive $c<1/14$. The precise value $1/14$ will be important for the check of \eqref{Pconds}. Then, for every $\mu \in \mathrm{Src}$, we will obtain $\deltaB_v^\mu$ by the following step-by-step construction:
\begin{enumerate}
    \setcounter{enumi}{-1}
    \item Let $N_{source}$ and $N_{sink}$ denote the numbers of sources and sinks of $\mu$.
    \item \label{itm:dAv} We denote by $\deltaA_v$ the restriction of $\deltaA$ which is restricted according to $\mu$ around $v$.\footnote{The expression "restricted according to $\mu$ around $v$" was introduced in the discussion after \eqref{trm_dscrA}.} For $v \notin\graph (\mu)$, this implies $\deltaA_v=0$ as $\deltaA_{000000}=0$. Note that \eqref{semistep1} holds.
    \item \label{itm:dAvtld} We define tensors $\deltaA^\mu_v$ by the reindexing of the internal legs of $\deltaA_v$ as in our example. Namely, the reindexing operation replaces positive integer index values with elements of $\N_\mu$. Note that \eqref{semistep2} holds and that tensors $\deltaA^\mu_v$ are restricted according to $\mu$.
    \item \label{itm:dbv} We define $\deltaB^\mu_v$ by reweighting as follows:
          \begin{subequations}\label{main_def}
              \begin{align}
                  \forall v \in\graph (\mu) \text{ that is a source}: & \ \deltaB_v^{\mu} = \epsilon^{-c N_{sink} }  \deltaA^\mu_v ;\label{sources} \\
                  \forall v \in\graph (\mu) \text{ that is a sink}:   & \ \deltaB_v^{\mu} = \epsilon^{c N_{source} } \deltaA^\mu_v; \label{sinks}   \\
                  \forall v \notin\graph (\mu):                       & \ \deltaB_v^{\mu}= 0.\label{nothing}
              \end{align}
          \end{subequations}
          We also extend all $\deltaB_v^\mu$ by zeros such that their internal legs take values in $\mathcal{D}_0$. Note that $\deltaB_v^\mu$ tensors are restricted according to $\mu$.

          This construction generalises \eqref{dbv_src} and \eqref{dbv_snk} of our example where the template $\mu$ had $N_{source}=2$ and $N_{sink}=1$. As in the example, the product of all reweighting factors gives $1$:
          \begin{equation}\label{factors_cancel}
              \left(\epsilon^{-cN_{sink} }\right)^{N_{source}}\left(\epsilon^{cN_{source} }\right)^{N_{sink}}=1.
          \end{equation}
          Hence, we have \eqref{semisteps} as well as \eqref{fullstep}.
\end{enumerate}

\subsubsection{Verification of \texorpdfstring{\eqref{Pconds}}{(\getrefnumber{Pconds})}} \label{PcondsDBV1}
It is evident that \eqref{P0} and \eqref{P3} are satisfied. Thus, we focus on \eqref{P1} and \eqref{P2}. We will show that \eqref{P1} and \eqref{P2} hold for any $a,b$ within the intervals:
\begin{equation}\label{gen_intervals}
    1/2<a\leq 1-7c, \qquad 1<b\leq \min (1+c,2-7c)= 1+c.
\end{equation}
As in the example, the lower bounds are the conditions on $a,b$ in Lemma~\ref{mainlem}. The condition $c<1/14$ ensures that the interval for $a$ is non-empty.

For $v \notin\graph (\mu)$, \eqref{P1} and \eqref{P2} are trivially satisfied due to \eqref{nothing}. Consequently, we focus on $\deltaB_v^\mu$ with $v \in\graph (\mu)$. Our reasoning is analogous to the one under Eq.~\eqref{domains} in Section~\ref{example_section_dbv1}. In particular, one can see that the bounds \eqref{bound_b1}---\eqref{bound_b4} in the example are the special cases of \eqref{bound_sink}---\eqref{v_is2} with $N_{source}=2$ and $N_{sink}=1$.

We need to consider three cases for $v \in\graph (\mu)$: 1) a sink, 2) a source satisfying \eqref{source1}, and 3) a source satisfying \eqref{source2}. If $v $ is a sink, by \eqref{sinks}, we have:
\begin{equation}\label{bound_sink}
    \deltaB_v^\mu = O(\epsilon^{1+c N_{source}}).
\end{equation}
Since $\mu \in \mathrm{Src}$, we have $N_{source} \geq 1$, and so \eqref{P1} and \eqref{P2} hold with $a,b$ as in \eqref{gen_intervals}.

If $v$ is a source satisfying \eqref{source1}, then, by the reasoning as under \eqref{bound_b5} in the example, we can see that $\bar{\deltaB_v^\mu}=0$. Thus, \eqref{P2} holds with any $b$. As for \eqref{P1}, by \eqref{sources}, we have:
\begin{equation}\label{v_is1}
    \deltaB_v^\mu = O(\epsilon^{1-c N_{sink}}).
\end{equation}
Note that $N_{sink}\leq 7$, as we should have at least one source, and so \eqref{P1} holds with $a$ as in \eqref{gen_intervals}.

Finally, let $v$ be a source satisfying \eqref{source2}. We will assume that $v$ violates \eqref{source1}, since otherwise, we are in the previous case. Then, $\deltaA_v$ has exactly one leg that takes values in $\N$. Since $\deltaA$ has the corner structure \eqref{corn}, we have $\deltaA_v, \deltaA^\mu_v=O(\epsilon^2)$, and, by \eqref{sources},
\begin{equation}\label{v_is2}
    \deltaB_v^\mu = O(\epsilon^{2-c N_{sink}}),
\end{equation}
which, by $N_{sink} \leq 7$, satisfies \eqref{P1} and \eqref{P2} with $a,b$ as in \eqref{gen_intervals}.

\subsubsection{Verification of \texorpdfstring{\eqref{compatrel1}}{(\getrefnumber{compatrel1})}} \label{compatdbv1rels}

Recall that a template $\mu$ is a function from $\{1,\ldots,36\}$ to $\{\mathcal{O}, \mathcal{D}\}$. Numbers $1,\ldots,36$ label legs and bonds of a contraction as in \eqref{Tdec}. Let $n^v_1,\ldots,n^v_6 \in \{1,\ldots,36\}$ be the numbers labelling in Eq.~\eqref{Tdec} legs and bonds around vertex $v$ in some arbitrary fixed order, for definiteness as in the following diagram:
\begin{equation}\label{positions_of_ns}
    \includegraphics[scale=0.2, valign=c]{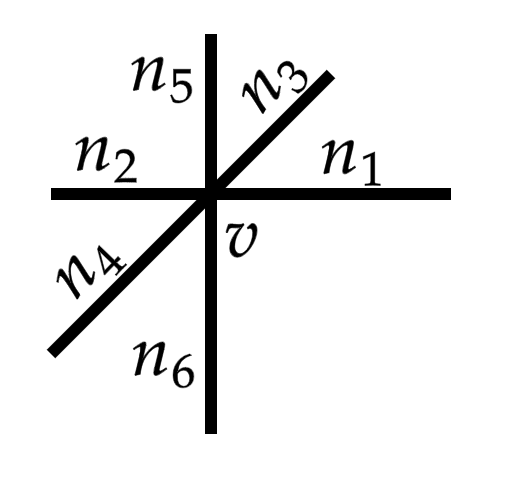}
\end{equation}
These labels $n_i^v$ of legs are not to be confused with indices of the corresponding leg, which take values in $\mathcal{D}$ or $\mathcal{O}$.

A crucial property of the constructed $\deltaB_v^\mu$ tensors is that they are restricted according to $\mu$ (see the discussion after \eqref{trm_dscrA} and the end of Section~\ref{gendbv1}). Using the labels $n^v_i$ of legs, we can express this as follows:
\begin{equation}\label{support}
    \text{each leg $n^v_i$ of $\deltaB_v^\mu$ is restricted to $\mu(n^v_i)$.}
\end{equation}

We start by checking \eqref{compat1}. If $\gamma=\mu$, then \eqref{compat1} coincides with \eqref{fullstep}, which holds as mentioned at the end of Section~\ref{gendbv1}. Assume next that $\gamma \neq \mu$. We need to show that $T_\gamma(\{\deltaB_v^\mu\})=0$.

We will prove that $T_\gamma(\{\deltaB_v^\mu\})=0$ by examining $\deltaB_w^\mu$ with $w\in \graph(\gamma)$ as in \eqref{vertexW} below.
\begin{lem}\label{lem:vertexW}
    Let $\mu \in \mathrm{Con}$ and $\gamma \in \mathrm{Tmpl}'$. If $\gamma \neq \mu$, then there is
    \begin{equation}\label{vertexW}
        \text{a vertex $w \in\graph (\gamma)$ and a leg $n_i^w$ such that $\gamma(n_i^w) \neq \mu(n_i^w)$.}
    \end{equation}
\end{lem}
Having the vertex $w$ as in \eqref{vertexW}, we proceed as follows. By \eqref{trm_dscr1}, $T_\gamma(\{\deltaB_v^\mu\})$ inserts $\deltaB_w^\mu$ in $w$. By \eqref{trm_dscr3} and \eqref{support}, the leg $n^w_i$ of $\deltaB_w^\mu$ is restricted in this contraction to $\gamma(n^w_i) \cap \mu(n^w_i)=\mathcal{D} \cap \mathcal{O} = \void$. Therefore, $T_\gamma(\{\deltaB_v^\mu\})=0$.

\begin{proof}[Proof of the Lemma~\ref{lem:vertexW}]
    To find such a $w$, we consider two possibilities for $\gamma \neq \mu$. The first possibility is that there is a vertex $w$ and a leg or bond $n^w_i$ such that $\gamma(n^w_i) =\mathcal{D}$ and $\mu (n^w_i) = \mathcal{O}$. As $\gamma(n^w_i)=\mathcal{D}$, we have $w \in\graph (\gamma)$. So $w$ as in \eqref{vertexW} is found.

    The second possibility is that $\gamma \neq \mu$, but there is no vertex $w$ as in the previous paragraph. Then, whenever $\gamma(n) \neq \mu (n)$, we have $\gamma (n) = \mathcal{O}$ and $\mu(n)=\mathcal{D}$. This means that $\diag(\gamma)$ is obtained from $\diag(\mu)$ by replacing some ticked ($\mathcal{D}$) legs and/or bonds with dashed ones ($\mathcal{O}$). Let $S_{replace}$ be the set of so-replaced ticked legs and bonds, and let $S_{remain}$ be the set of remaining ticked legs and bonds (nonempty because $\gamma$ is nontrivial). Consider the set $V'$ of all vertices touched by lines from $S_{replace}$. We choose a vertex $w \in V'$ which is touched by a line from $S_{remain}$. [Such a vertex exists because if not, the graph with vertices=$V'$ and edges ="bonds from $S_{replace}$" would form a nontrivial connected component of $\graph (\mu)$, a contradiction since $\mu\in \mathrm{Src} \subset \mathrm{Con}$.] Since $w$ is touched by a line from $S_{replace}$, there is a $n_j^w$ such that $\gamma(n^w_j)=\mathcal{O}$, $\mu(n^w_j)=\mathcal{D}$. Since $w$ is touched by a line from $S_{remain}$, we have $w \in\graph (\gamma)$. So $w$ as in \eqref{vertexW} is again found.
\end{proof}

We continue by checking \eqref{compat2}. Let $\gamma$ be a connected template. We denote by $v_1,\ldots,v_N$ the vertices of $\graph (\gamma)$. By \eqref{trm_dscr}, the contraction $T_\gamma(\{\deltaB_v^1\})$ (see \eqref{dbv1exp}) depends only on $\deltaB_{v_1}^1,\ldots, \deltaB_{v_N}^1$. To emphasise this, we will use the following notation:
\begin{equation}\label{dep_on_r}
    T_\gamma(\{\deltaB_v^1\})=T_\gamma (\deltaB_{v_1}^1,\ldots, \deltaB_{v_N}^1).
\end{equation}
We recall our basic notational convention that the tensor with label $v_i$ is inserted into the vertex $v_i$, see Remark~\ref{orderconvention}. The order in which $\deltaB_{v_1}^1,\ldots,\deltaB_{v_N}^1$ are listed in the argument of $T_\gamma$ is irrelevant.

We substitute \eqref{dbv1exp} into the r.h.s. of \eqref{dep_on_r} and expand the result by multilinearity of $T_\gamma(\deltaB_{v_1}^1,\ldots,\deltaB_{v_N}^1)$:
\begin{equation}\label{lin}
    T_\gamma (\deltaB_{v_1}^1, \ldots, \deltaB_{v_N}^1)=T_\gamma \left( \sum_{\mu_1 \in \mathrm{Src}} \deltaB_{v_1}^{\mu_1},\ldots, \sum_{\mu_N \in \mathrm{Src}} \deltaB_{v_N}^{\mu_N}\right)  =\sum_{\mu_1, \ldots, \mu_N \in \mathrm{Src}} T_\gamma \left( \deltaB_{v_1}^{\mu_1},\ldots,\deltaB_{v_N}^{\mu_N} \right).
\end{equation}
By definition of $\mathrm{Con}$, the graph $\graph (\gamma)$ is connected. The legs of each $\deltaB_{v}^\mu$ contracted along the edges of $\graph (\gamma)$ are restricted to index sets $\N_\mu$, which are disjoint for different $\mu$'s. It follows that $T_\gamma \left( \deltaB_{v_1}^{\mu_1},\ldots,\deltaB_{v_N}^{\mu_N} \right)$ is nonzero only for $\mu_1=\ldots=\mu_N$, in which case it equals $T(\{\deltaB_{v}^\mu\})$. This implies \eqref{compat2}.

We finish by checking \eqref{restB}. Let $\gamma \in \mathrm{Disc}$. Recall that we may always uniquely express it as a union (see Def.~\ref{def:union} and  Eq.~\eqref{uni_def_new}) of connected templates $\gamma_1,\ldots, \gamma_m$:
\begin{equation}\label{gamma_is_union}
    \gamma=\gamma_1\star \ldots \star \gamma_m.
\end{equation}
Then, the graph $\graph (\gamma)$ is the disjoint union of connected graphs $\graph (\gamma_k)$, ($k=1,\ldots,m$):
\begin{equation}
    \graph (\gamma) =\graph (\gamma_1) \sqcup \ldots \sqcup\graph (\gamma_m).
\end{equation}

We now consider expansion \eqref{lin} for $T_\gamma(\{\deltaB_v^1\})$.\footnote{Note that \eqref{lin} holds for any template $\gamma$, whether it is connected or not.} In the proof of \eqref{compat2}, the terms on the r.h.s. were vanishing unless all $\mu_i$ were the same. In the current situation, since $\gamma \in \mathrm{Disc}$, the terms on the r.h.s. vanish unless $\mu_i=\mu_j=:\mu^{(k)}$ for each $v_i,v_j$ belonging to the same connected component $\graph (\gamma_k)$. We denote by $v_{k,1},\ldots,v_{k,N_k}$ the vertices of $\graph (\gamma_k)$. Thus, we have:
\begin{equation}\label{lin2}
    T_\gamma(\{\deltaB^1_v\})=\sum_{\mu^{(1)},\ldots,\mu^{(m)} \in \mathrm{Src}} T_{\gamma} (\deltaB_{v_{1,1}}^{\mu^{(1)}},\ldots,\deltaB_{v_{1,N_1}}^{\mu^{(1)}}, \ldots,\deltaB_{v_{m, 1}}^{\mu^{(m)}},\ldots, \deltaB_{v_{m, N_m}}^{\mu^{(m)}} ).
\end{equation}

Now, we examine individual terms in \eqref{lin2}. We claim that the following is true:
\begin{subequations}\label{disc_ab}
    \begin{align}
         & \text{if $\mu^{(k)}=\gamma_k$ for all $k$, the term in the r.h.s of \eqref{lin2} equals $T_\gamma(\deltaA)$;}\label{disc_a} \\
         & \text{if $\mu^{(k)}\neq \gamma_k$ for some $k$, the term in the r.h.s. of \eqref{lin2} vanishes.}\label{disc_b}
    \end{align}
\end{subequations}
Together, these two facts clearly imply \eqref{restB}.

Let us discuss \eqref{disc_a}. It is an easy generalisation of \eqref{fullstep}. Reindexing works out fine (there is no interference between different connected components $\gamma_k$). Reweighting factors, which were introduced in \eqref{main_def} when passing from $\deltaA$ to $\deltaB_{v}^{\mu^{(k)}}$, multiply to $1$ separately for each connected component $\gamma_k$. We omit the details.

Now, let us discuss \eqref{disc_b}. We assume that for some $k$, $\mu^{(k)}\neq \gamma_k$. We need to show that:
\begin{equation}\label{disc_b_w}
    T_{\gamma} (\deltaB_{v_{1,1}}^{\mu^{(1)}},\ldots,\deltaB_{v_{m, N_m}}^{\mu^{(m)}} ) =0
\end{equation}
By Lemma~\ref{lem:vertexW}, there is a vertex $w\in\graph (\gamma_k)$ and $n^w_i$ such that:
\begin{equation}\label{neqgammakmu}
    \gamma_k(n^w_{i}) \neq \mu^{(k)}(n^w_i)
\end{equation}
We note that $\gamma$ coincides with $\gamma_k$ around vertices $v_{k,i} \in\graph (\gamma_k)$ (see, e.g.~\eqref{unionExample}):
\begin{equation}\label{def_uni_new2}
    \gamma(n^{v_{k,i}}_j)=\gamma_k(n^{v_{k,i}}_j).
\end{equation}
[If \eqref{def_uni_new2} is violated, by \eqref{uni_def_new}, we have $\gamma_k (n_j^{v_{k,i}}) = \mathcal{O}$ and $\gamma(n_j^{v_{k,i}})=\mathcal{D}$. Then, by \eqref{gamma_is_union}, there is an $l\neq k$ such that $\gamma_l (n_j^{v_{k,i}})=\mathcal{D}$, and so $v_{k,i} \in\graph (\gamma_l)$, a contradiction since $\graph (\gamma_k)$'s are disjoint.] Thus, by \eqref{neqgammakmu} and \eqref{def_uni_new2}, we have $\gamma(n^w_{i}) \neq \mu^{(k)}(n^w_i)$. Then, by \eqref{trm_dscr1}, $T_{\gamma} (\deltaB_{v_{1,1}}^{\mu^{(1)}},\ldots,\deltaB_{v_{m, N_m}}^{\mu^{(m)}} )$ inserts $\deltaB_w^{\mu^{(k)}}$ in $w$. By \eqref{trm_dscr3} and \eqref{support}, the leg $n^w_i$ of $\deltaB_w^{\mu^{(k)}}$ is restricted in this contraction to $\gamma(n_i^w)\cap \mu^{(k)}(n_i^w)=\mathcal{D}\cap \mathcal{O}=\void$. Therefore, \eqref{disc_b_w} holds.

Let us note that the provided proof of \eqref{disc_b} can be easily generalised to the proof of the following lemma, which will be used in Section~\ref{CON2}.
\begin{lem}\label{lem:incompatible_contractions}
    Let $\gamma$ be a template with connected components $\gamma_1,\ldots, \gamma_m$, where $m\geq 1$ (if $m=1$, $\gamma$ is a connected template). Let $v_1,\ldots,v_n$ be vertices of $\graph(\gamma_k)$ ($k\in \{1,\ldots,n\}$) and $\vec{c}_{v_1},\ldots, \vec{c}_{v_n}$ be tensors restricted according to some template $\mu \neq \gamma_k$. Then
    \begin{equation}
        T_\gamma(\ldots, \vec{c}_{v_1},\ldots, \vec{c}_{v_n}, \ldots  )=0,
    \end{equation}
    where "$\ldots$" to the left and to the right from $\vec{c}_{v_1},\ldots, \vec{c}_{v_n}$ stands for tensors inserted in other vertices of $\graph(\gamma)$ (if $m=1$ there are only $c_{v_1},\ldots,c_{v_n}$ in the argument of $T_\gamma$).
\end{lem}

We have shown that constructed tensors $\deltaB_v^{\mu}$, $\mu \in \mathrm{Src}$, satisfy \eqref{Pconds} and \eqref{compatrel1}. Hence, as discussed at the beginning of Section~\ref{dbv1}, tensors $\deltaB_v^1$ defined by \eqref{dbv1exp} satisfy \eqref{Pconds} and \eqref{Src_sect}. So, we accomplished the goal of this section: to construct $\deltaB_v^1$ satisfying \eqref{Pconds} and \eqref{Src_sect}.

\subsection{Construction of tensors \texorpdfstring{$\deltaB_v^2$}{bv2}}\label{CON2}

In this subsection, we will construct tensors $\deltaB_v^2$, which will be the correction terms that ensure $\deltaB_v=\deltaB_v^1 +\deltaB_v^2$ resolves the main equation \eqref{maineq2}. This will finish the proof of Lemma~\ref{mainlem}.

We will search for $\deltaB_v^2$ in the form analogous to \eqref{dbv1exp} but with summation over $\sigma \in \mathrm{NSrc}$:
\begin{equation}\label{dbv2Exp1}
    \deltaB_v^2=\sum_{\sigma \in \mathrm{NSrc}} \deltaB_v^{\sigma},
\end{equation}
We will ensure that  tensors $\deltaB_v^{\sigma}$ ($\sigma \in \mathrm{NSrc}$) satisfy \eqref{Pconds} and so $\deltaB_v^2$ will also satisfy \eqref{Pconds}. We will achieve \eqref{maineq2} by constructing $\deltaB_v^{\sigma}$ such that:
\begin{itemize}
    \item For each $\sigma \in \mathrm{NSrc}$ there exists a corresponding $\widehat{\sigma} \in \mathrm{Con}$ such that:
          \begin{equation}\label{compat3}
              \forall \gamma \in \mathrm{Con}: \qquad T_\gamma(\{\deltaB_v^{\sigma}\}) =\begin{cases}
                  T_{\sigma}(\deltaA), & \text{if } \gamma = \widehat{\sigma}    \\
                  0,                   & \text{if } \gamma \neq \widehat{\sigma}
              \end{cases}.
          \end{equation}
          Note that we do not require the map $\sigma \mapsto \widehat{\sigma}$ to be injective. We will give the explicit definition of $\widehat{\sigma}$ later.
    \item For any $\gamma \in \mathrm{Con}$ we have:
          \begin{equation}\label{compat4}
              T_\gamma(\{\deltaB^1_v+\deltaB^2_v\})= T_\gamma(\{\deltaB_v^1\})+\sum_{\sigma \in \mathrm{NSrc} } T_\gamma(\{\deltaB_v^{\sigma}\}).
          \end{equation}
    \item For disconnected templates, we have:
          \begin{equation}\label{compat5limit}
              \sum_{\sigma \in \mathrm{Disc}} T_\sigma(\{\deltaB_v^1+\deltaB_v^2\})=\sum_{\sigma \in \mathrm{Disc}} T_\sigma (\deltaA)
          \end{equation}
\end{itemize}

Let us show that \eqref{maineq2} (where $\deltaB_v=\deltaB_v^1+\deltaB_v^2$) follows from \eqref{dbv1_goal} and \eqref{dbv2Exp1}-\eqref{compat5limit}. Using the division of templates \eqref{classif} and Eq.~\eqref{triv}, we express the l.h.s. of \eqref{maineq2} as:
\begin{equation}\label{sum0}
    \sum_{\gamma \in \mathrm{Tmpl}} T_\gamma(\{\deltaB_v\})= T_* + \sum_{\gamma \in \mathrm{Con}} T_\gamma(\{\deltaB_v\})+ \sum_{\gamma \in \mathrm{Disc}} T_\gamma(\{\deltaB_v\}).
\end{equation}
Consider the sum over $\mathrm{Con}$ in the r.h.s. of \eqref{sum0}. Applying \eqref{compat4}, \eqref{dbv1_goal} and \eqref{compat3}, we obtain:
\begin{equation}\label{sum1}
    \begin{aligned}
        \sum_{\gamma \in \mathrm{Con}} T_\gamma(\{\deltaB_v\}) & =                                                                                                                                                                                                        \\
                                                               & =[\text{by \eqref{compat4}}]\sum_{\gamma \in \mathrm{Con}} T_\gamma(\{\deltaB_v^1\}) + \sum_{\gamma \in \mathrm{Con}} \sum_{\sigma \in \mathrm{NSrc}} T_\gamma(\{\deltaB^{\sigma}_v\})                   \\
                                                               & =[\text{by \eqref{dbv1_goal},\eqref{compat3}}]\ \sum_{\gamma \in \mathrm{Src}} T_\gamma(\deltaA) + \sum_{\gamma \in \mathrm{NSrc}} T_{\gamma}(\deltaA)=\sum_{\gamma \in \mathrm{Con}} T_\gamma(\deltaA).
    \end{aligned}
\end{equation}
Then, substituting \eqref{sum1} and \eqref{compat5limit} into the r.h.s. of \eqref{sum0} we obtain \eqref{maineq2}.

Let us explain how we will achieve \eqref{compat3}-\eqref{compat5limit}. Eqs.~\eqref{compat3} and \eqref{compat4} are relatively easy to achieve. For \eqref{compat3}, we will ensure that tensors $\deltaB_v^{\sigma}$ solve the first line of \eqref{compat3} and are restricted according to $\widehat{\sigma}$.\footnote{The expression "restricted according to $\widehat{\sigma}$" is clarified in the discussion after \eqref{trm_dscrA}.}\footnote{We slightly simplified the discussion here. We will see that some tensors $\deltaB_v^{\sigma}$ are not restricted according to $\widehat{\sigma}$. However, these tensors will not appear in the contractions $T_\gamma(\{\deltaB_v^{\sigma}\})$ for connected $\gamma$'s and so will not break \eqref{compat3}.}
Then, \eqref{compat3} will be proven by the reasoning analogous to that leading to \eqref{compat1} in Section~\ref{dbv1} (see the discussion below \eqref{support}).

For \eqref{compat4}, we will ensure that the internal ticked legs of tensors $\deltaB_v^\sigma$ for different $\sigma$'s are restricted to sets disjoint from each other and from $\bigsqcup\limits_{\mu \in \mathrm{Src}} \N_\mu$ (the set to which the internal ticked legs of $\deltaB_v^1$ are restricted).\footnote{This is, again, a slight simplification. There will be tensors violating this property. However, they will not appear in the contractions $T_\gamma (\{\deltaB_v^1 + \deltaB_v^2\})$ with connected $\gamma$'s.} Then, \eqref{compat4} will be proven by the reasoning analogous to that leading to \eqref{compat2} in Section~\ref{dbv1} (see the discussion around \eqref{dep_on_r}).

The most challenging property is \eqref{compat5limit}. We will achieve it as follows. We number the elements of $\mathrm{NSrc}$ in some arbitrary fixed order:
\begin{equation}\label{NSrc_ord}
    \mathrm{NSrc}=\{\sigma_1,\ldots,\sigma_{|\mathrm{NSrc}|}\}.
\end{equation}
Then, \eqref{dbv2Exp1} can be written as:
\begin{equation}\label{dbv2Exp}
    \deltaB_v^2=\sum_{q=1}^{|\mathrm{NSrc}|} \deltaB_v^{\sigma_q}.
\end{equation}

We define sets $\mathrm{Disc}_M \subset \mathrm{Disc}$ ($M=1,\ldots, |\mathrm{NSrc}|$) as the sets containing all unions of templates belonging to $\mathrm{Src}\sqcup\{\sigma_1, \ldots, \sigma_M\}$:
\begin{equation}\label{DiscMdef}
    \mathrm{Disc}_M = \{\gamma_1 \star \ldots \star \gamma_l \mid \gamma_k \in \mathrm{Src} \sqcup \{ \sigma_1, \ldots, \sigma_{M}\}, \ l\geq 2  \}.
\end{equation}
Note that in the limiting case $M=|\mathrm{NSrc}|$, in \eqref{DiscMdef}, by \eqref{NSrc_ord}, we have $\gamma_k \in \mathrm{Src}\sqcup \mathrm{NSrc}=\mathrm{Con}$. Thus,
\begin{equation}\label{DiscMLim}
    \mathrm{Disc}_{|\mathrm{NSrc}|}= \mathrm{Disc},
\end{equation}
as any disconnected template can be expressed as a union of connected ones.

We will construct tensors $\deltaB_v^{\sigma_q}$ so that the following equation is satisfied for each $M=1,\ldots, |\mathrm{NSrc}|$:
\begin{equation}\label{compat5}
    \sum_{\gamma \in \mathrm{Disc}} T_\gamma\left(\left\{\deltaB^1_v+\sum_{q=1}^{M}\deltaB^{\sigma_q}_v\right\}\right)=\sum_{\gamma \in \mathrm{Disc}_M} T_\gamma(\deltaA).
\end{equation}
Note that by \eqref{dbv2Exp} and \eqref{DiscMLim} this equation reduces to \eqref{compat5limit} for $M=|\mathrm{NSrc}|$.

In this way, we reduced the initial problem of finding the correction terms $\deltaB_v^2$ satisfying \eqref{Pconds} to the problem of finding $\deltaB_v^{\sigma_q}$, satisfying \eqref{Pconds}, \eqref{compat3},\eqref{compat4}, and \eqref{compat5}. We will present the construction of $\deltaB_v^{\sigma_q}$ tensors as follows. Firstly, we will consider an example: we will discuss how $\deltaB_v^{\sigma_q}$ is constructed for $q=1$ and for a particular $\sigma_1 \in \mathrm{NSrc}$. Secondly, we will consider the general case and construct $\deltaB_v^{\sigma_q}$ for all $q=1,\ldots, |\mathrm{NSrc}|$. Finally, we will check \eqref{Pconds},  \eqref{compat3}, \eqref{compat4}, and \eqref{compat5}.

\subsubsection{Example at \texorpdfstring{$q=1$}{q=1}}\label{Con2Example}

Let $\sigma_1 \in \mathrm{NSrc}$ be given by the following diagram:\footnote{The specific choice of $\sigma_1$ is not important. We could have chosen any ordering of $\sigma$'s, and in particular, any $\sigma \in \mathrm{NSrc}$ could be chosen as $\sigma_1$. In the general construction below, the ordering will be arbitrary. \label{ftntnord}}
\begin{equation}\label{NSrc_ex_diag_mu}
    \diag(\sigma_1)=\includegraphics[scale=0.75, valign=c]{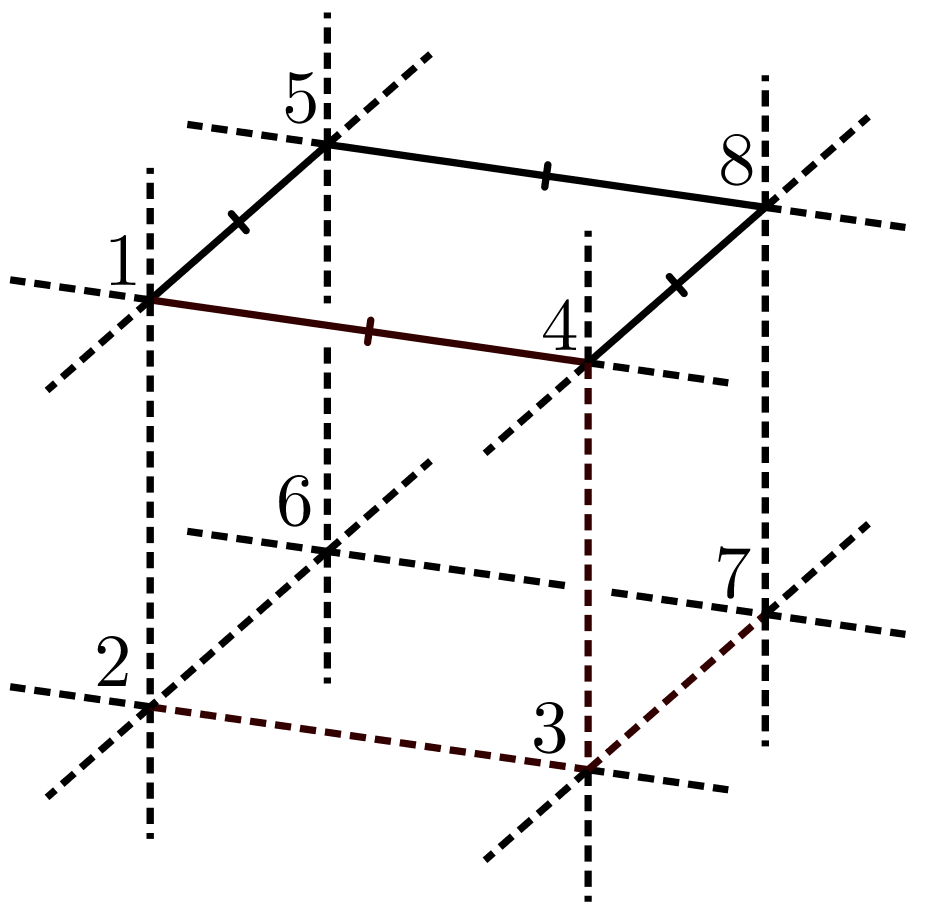}.
\end{equation}
Consider the corresponding $T_{\sigma_1}(\deltaA)$:
\begin{equation}\label{NSrc_C}
    T_{\sigma_1} (\deltaA) = \includegraphics[scale=0.75, valign=c]{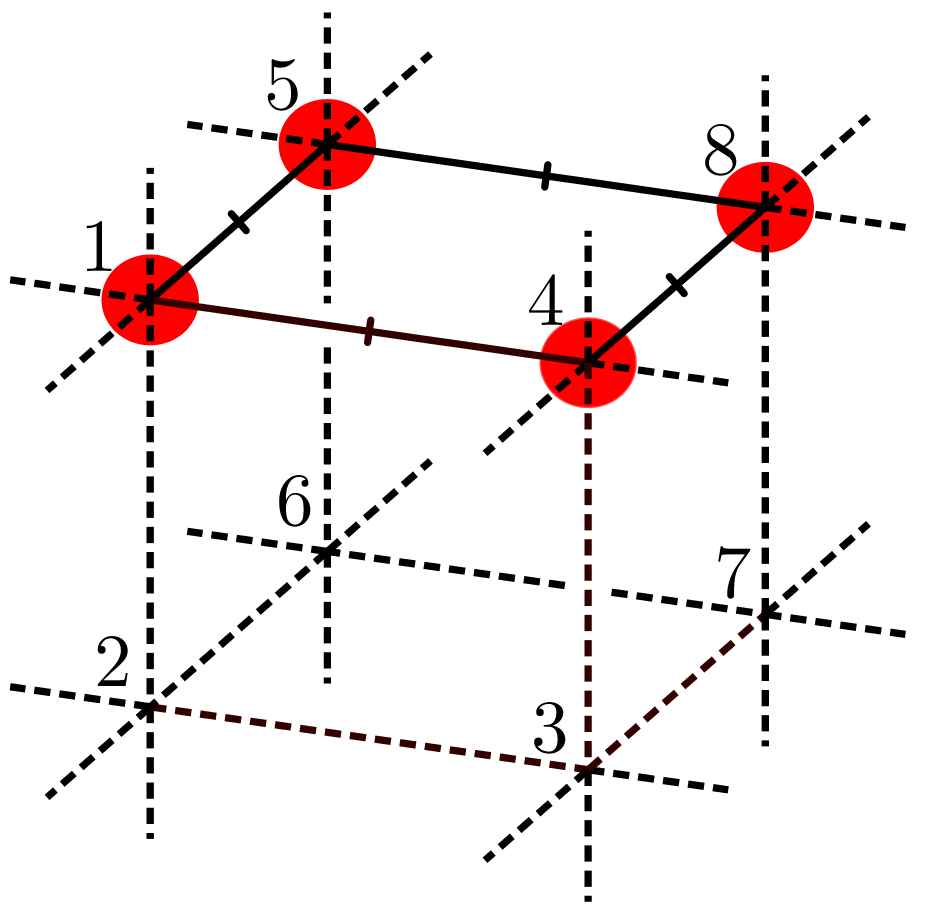}.
\end{equation}
All the legs of $T_{\sigma_1} (\deltaA)$ are restricted to $\{0\}$ (as $\deltaA$ is indexed by $\N_0$), so $T_{\sigma_1}(\deltaA)$ is proportional to $T_*$ (the cubic contraction of $8$ copies of $A_*$):
\begin{equation}\label{NSrc_toget}
    T_{\sigma_1} (\deltaA) = C_{\sigma_1} T_*,
\end{equation}
where we introduced a scalar quantity $C_{\sigma_1} = C_{\sigma_1}(\deltaA)$.

Now, we will construct $\widehat{\sigma_1}$ for $\sigma_1$ given by \eqref{NSrc_ex_diag_mu}. To satisfy the first line of \eqref{compat3}, we should find a template $\widehat{\sigma_1} \in \mathrm{Con}$ and tensors $\deltaB_v^{\sigma_1}$ such that:
\begin{equation}\label{NSrcEx_m}
    T_{\widehat{\sigma_1}}(\{\deltaB_v^{\sigma_1}\}) = T_{\sigma_1}(\deltaA)=[\text{by \eqref{NSrc_toget}}] \ C_{\sigma_1} T_*.
\end{equation}
As $T_{\sigma_1} (\deltaA)$ has $4$ insertions of $\deltaA$, we have:
\begin{equation}\label{Const_ord}
    C_{\sigma_1}=O(\epsilon^4).
\end{equation}
Consequently, as we want $\deltaB_v^{\sigma_1}$ to satisfy \eqref{Pconds}, $\graph (\widehat{\sigma_1})$ should have less than $4$ vertices. Otherwise, due to \eqref{NSrcEx_m} and \eqref{Const_ord}, at least one $\deltaB_v^{\sigma_1}$ tensor would violate \eqref{P2}. Then, to construct $\widehat{\sigma_1}$, we arbitrarily select a vertex $w \in\graph (\sigma_1)$ and obtain $\diag(\widehat{\sigma_1})$ by replacing all ticked bonds which touch the vertex $w$ in $\diag(\sigma_1)$ with dashed ones.\footnote{The vertex $w$ will have to be chosen slightly more carefully in the general construction.} Choosing $w=4$, we get:
\begin{equation}\label{NSrc_hatsigma_diag}
    \diag(\widehat{\sigma_1})=\includegraphics[scale=0.75, valign=c]{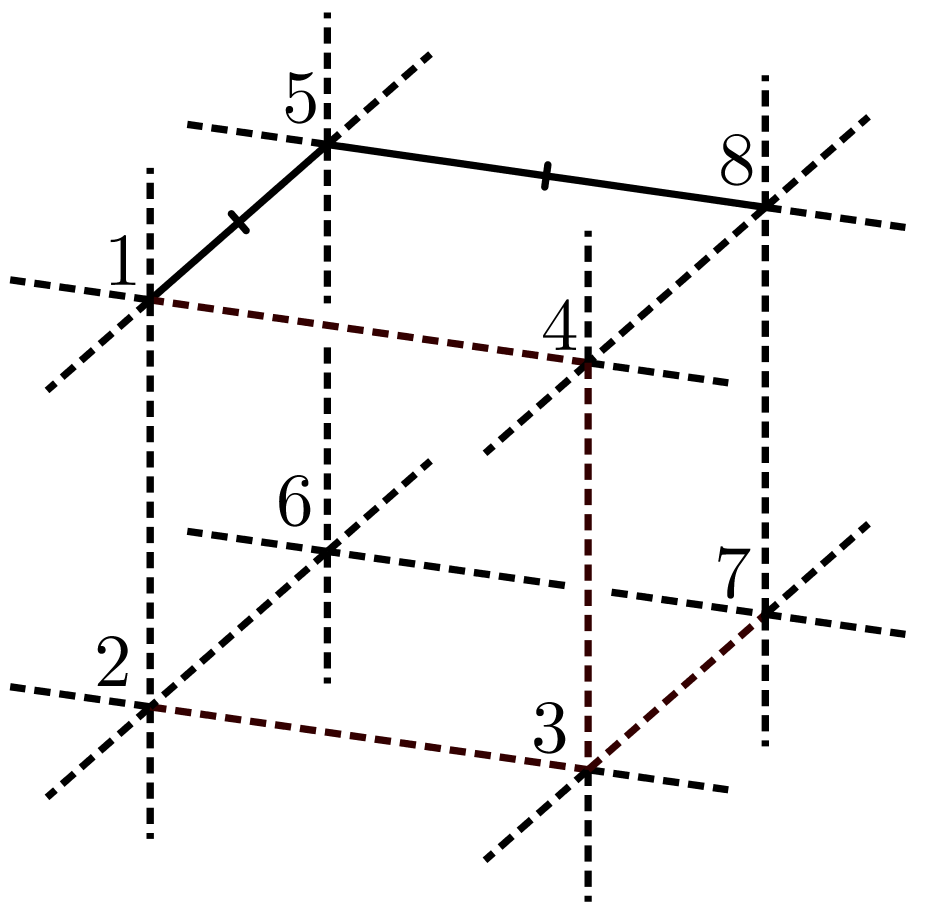}
\end{equation}
Clearly, $\widehat{\sigma_1} \in \mathrm{Con}$, as required by \eqref{compat3}.\footnote{In our example, $\widehat{\sigma_1}$ has two sources: vertices $1$ and $8$. Let us note in advance that this is not a general feature. There will be some $\widehat{\sigma_q} \in \mathrm{NSrc}$.}

We now present an initial, naive definition of tensors $\deltaB_v^{\sigma_1}$. Although this definition appears straightforward and resolves \eqref{NSrcEx_m}, we will see that it requires a subtle adjustment to satisfy \eqref{compat5} with $M=1$. This explains the term "naive".
\begin{itemize}
    \item For $v=1,8$, we define $\deltaB_v^{\sigma_1}$ as the tensors restricted according to $\widehat{\sigma_1}$, each with a single nonzero tensor element, given by:
          \begin{equation}\label{NSrc_def_18}
              \includegraphics[scale=0.75,valign=c]{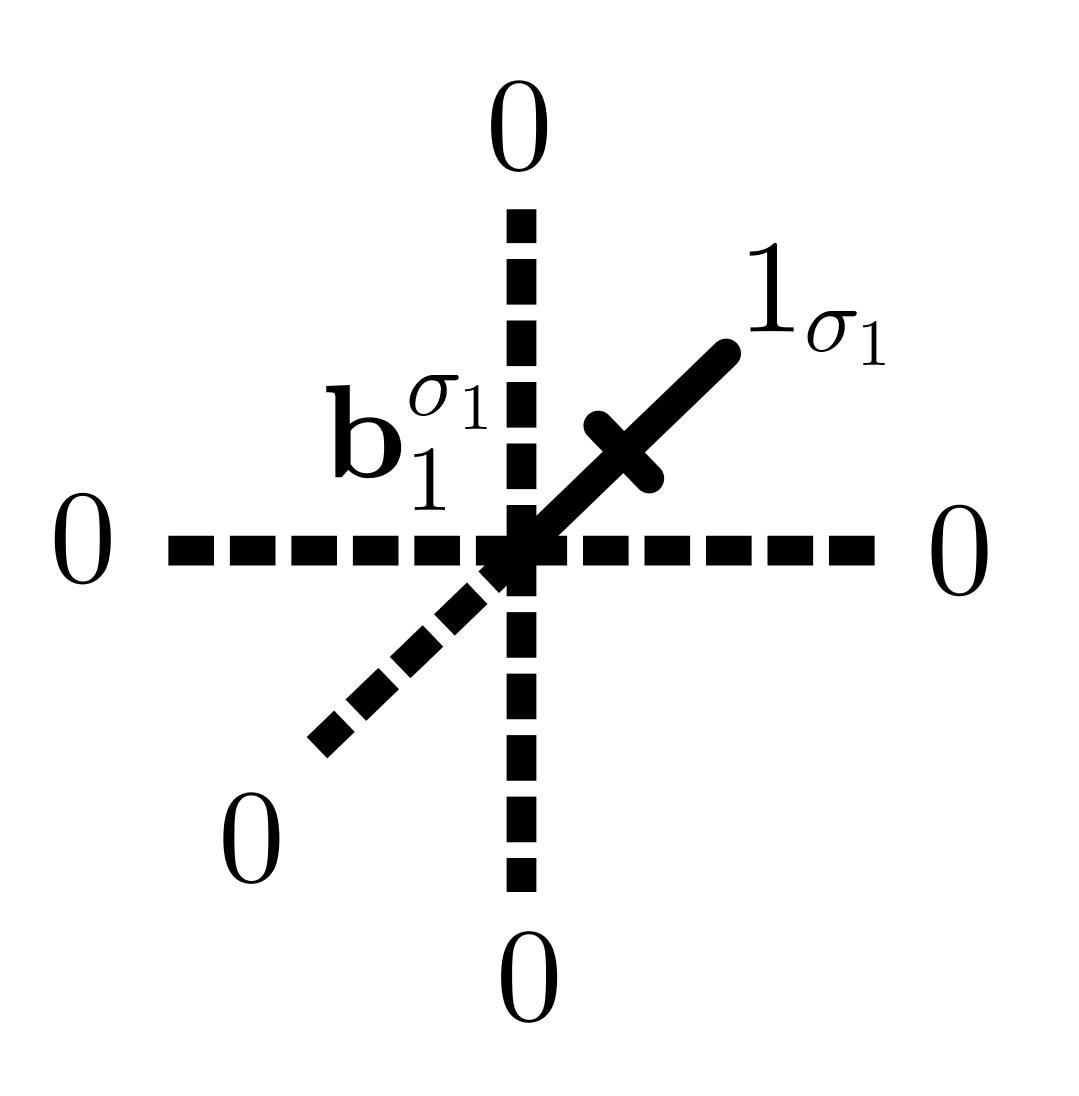}=\epsilon^{4/3}, \qquad \includegraphics[scale=0.75,valign=c]{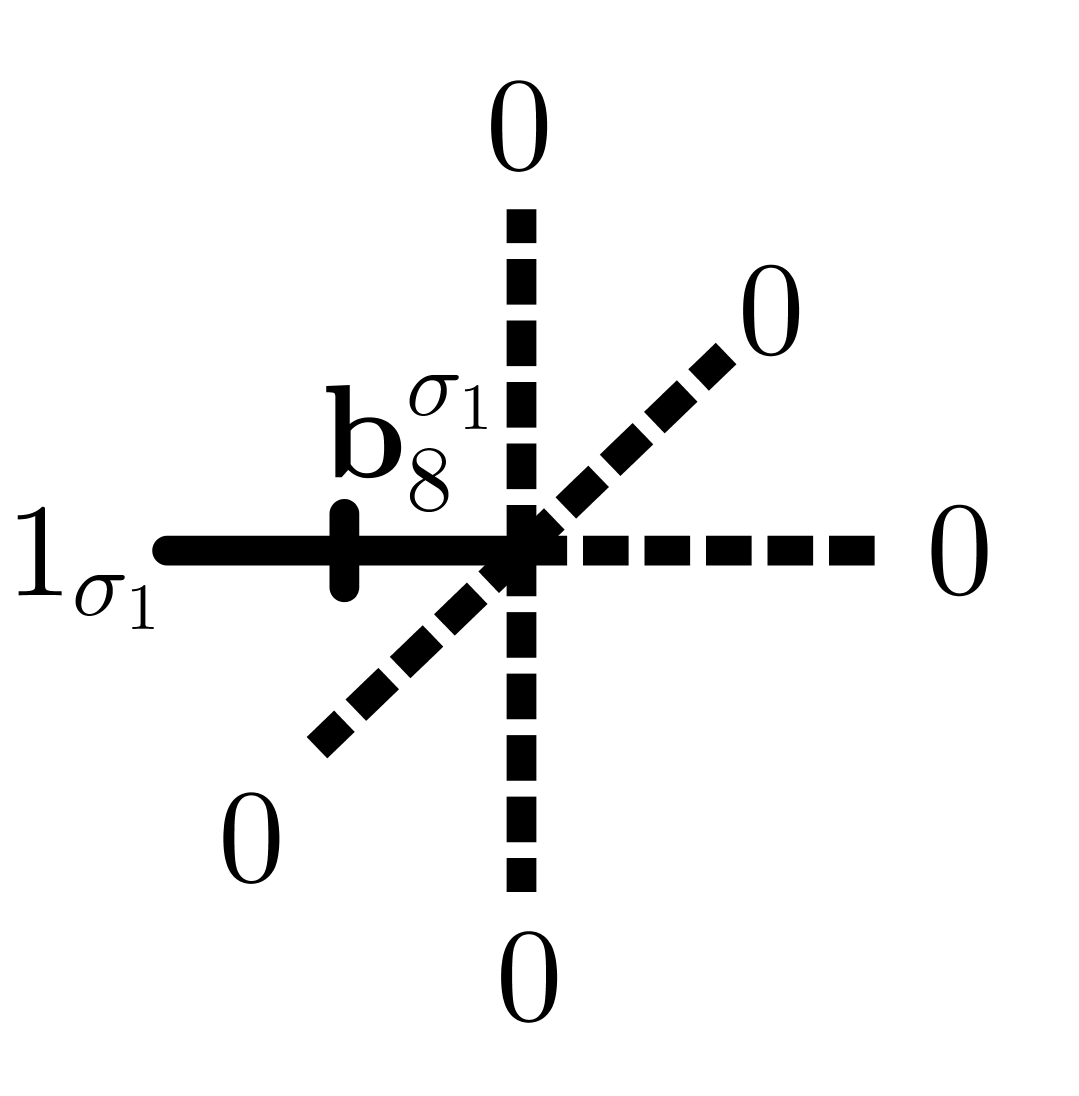}=\epsilon^{4/3},
          \end{equation}
          where $1_{\sigma_1}$ is an element of $\mathcal{D}$ introduced in Section~\ref{D0} (see also Footnote~\ref{footNg}).
    \item For $v=5$, we define $\deltaB_5^{\sigma_1}$ as the tensor restricted according to $\widehat{\sigma_1}$ with a single nonzero tensor element given by:
          \begin{equation}\label{NSrc_def_5}
              \includegraphics[scale=0.75,valign=c]{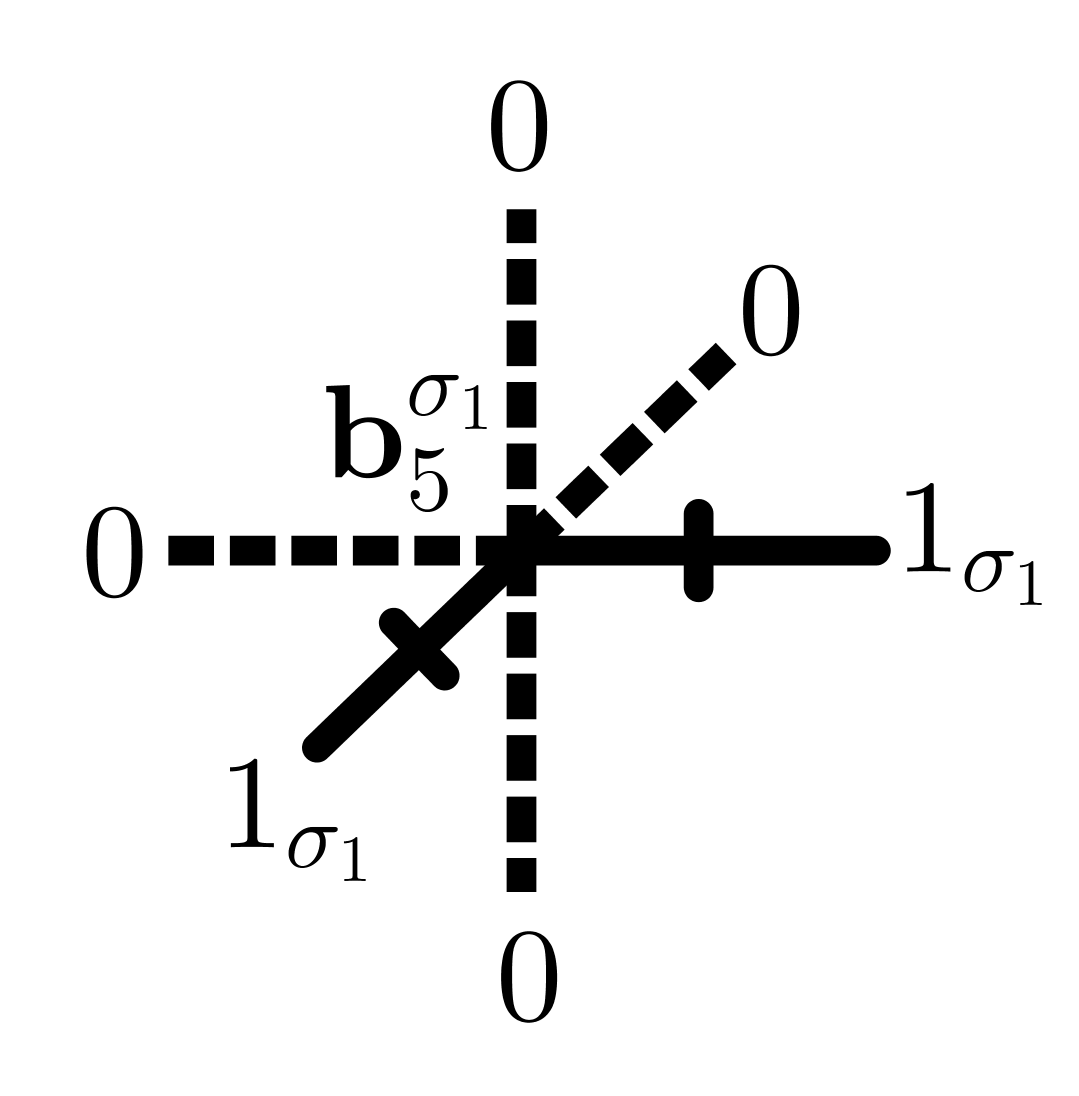}=C_{\sigma_1} \epsilon^{-8/3} (= O(\epsilon^{4/3})).
          \end{equation}
    \item For $v=2,3,4,6,7$ ($v \notin\graph (\widehat{\sigma_1})$), we define:
          \begin{equation}
              \deltaB_v^{\sigma_1}=0
          \end{equation}
\end{itemize}

Clearly, the so-defined tensors satisfy \eqref{NSrcEx_m}. Moreover, one can easily verify that \eqref{compat3}, including the second line, and \eqref{Pconds} are satisfied. We will discuss these properties after presenting the general construction. In the rest of this example, we will discuss the problem with \eqref{compat5} for $M=1$ posed by the naive definition and the modification of the naive definition which resolves this problem. We want to focus on \eqref{compat5} as its verification is the main challenge of Section~\ref{CON2} and it provides the motivation for the upcoming modification of the naive definition of $\deltaB_v^{\sigma_1}$.

To prove \eqref{compat5} for $M=1$, we aim to show that there is a collection of disjoint sets $D_\gamma \subset \mathrm{Disc}_1$, $\gamma \in \mathrm{Disc}$, (most of which are empty) such that
\begin{equation}\label{desire2}
    \mathrm{Disc}_1=\bigsqcup\limits_{\gamma \in \mathrm{Disc}} D_\gamma,
\end{equation}
and
\begin{equation}\label{desire}
    T_\gamma (\{\deltaB_v^1 + \deltaB_v^{\sigma_1}\}) = \sum_{\mu \in D_\gamma} T_\mu (\deltaA) \qquad  (\gamma \in \mathrm{Disc}).
\end{equation}
We assume that the sum over empty $D_\gamma$ is $0$. Clearly, \eqref{compat5} for $M=1$ follows from \eqref{desire} and \eqref{desire2}. Unfortunately, \eqref{desire} does not hold for the naive tensors $\deltaB_v^{\sigma_1}$. We will now demonstrate that there are templates $\gamma \in \mathrm{Disc}$ such that
\begin{equation}\label{NSrc_breaks}
    T_\gamma(\{\deltaB_v^1+\deltaB_v^{\sigma_1}\})=T_\gamma (\deltaA) + \text{"unwanted term"},
\end{equation}
where the "unwanted term" contradicts \eqref{desire}. Then, we will suggest a modification of the definition of $\deltaB_v^{\sigma_1}$ which cancels the unwanted term in the r.h.s. of \eqref{NSrc_breaks}. After this, we will show that \eqref{desire},\eqref{desire2} hold, and so, that \eqref{compat5}, for $M=1$, is satisfied.

Let us consider a disconnected template $\gamma$ of the form:
\begin{equation}\label{NSrc_gamma_def}
    \gamma = \widehat{\sigma_1} \star \lambda.
\end{equation}
where $\lambda$ is given by the following diagram:
\begin{equation}\label{NSrc_lambda_diag}
    \diag(\lambda)=\includegraphics[scale=0.75,valign=c]{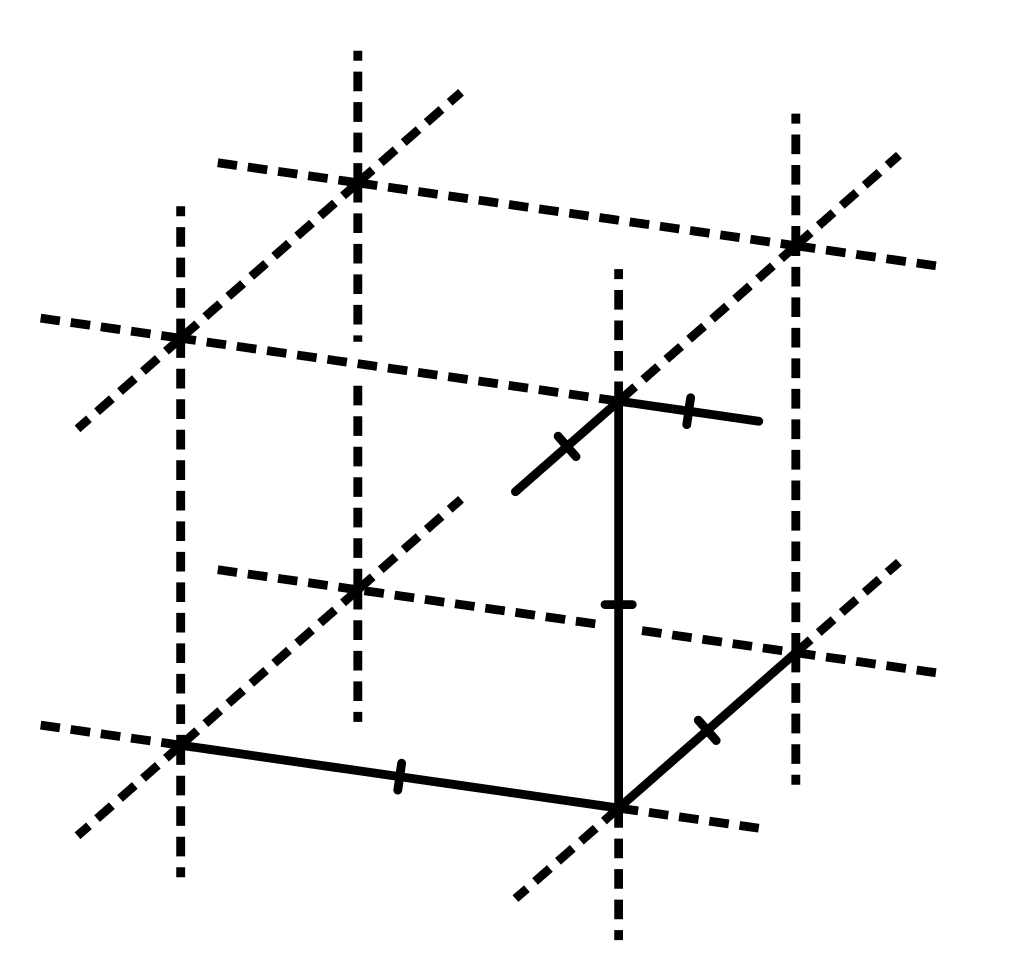}.
\end{equation}
So,
\begin{equation}\label{NSrc_gamma_diag}
    \diag(\gamma)=\includegraphics[scale=0.75,valign=c]{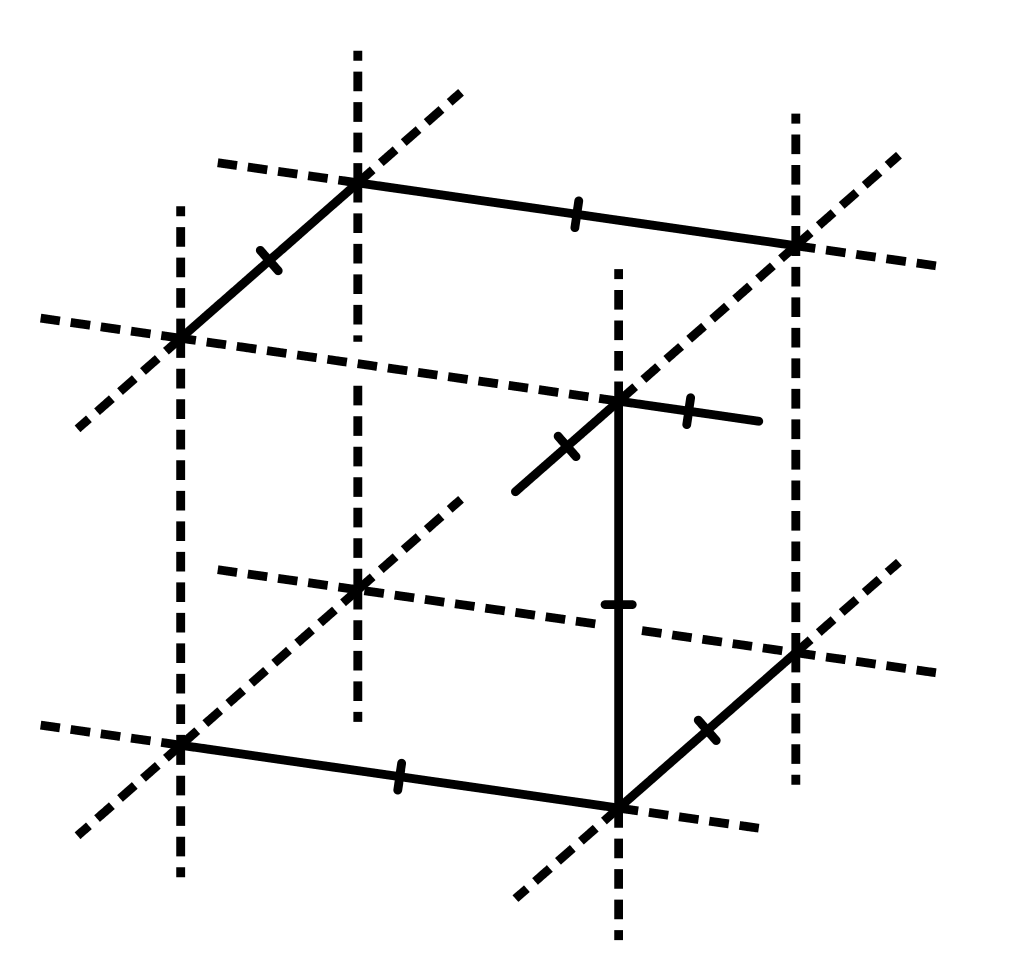}.
\end{equation}
Consider the corresponding $T_{\gamma}(\{\deltaB_v^1 +\deltaB_v^{\sigma_1}\})$:
\begin{equation}\label{NSrc_ex_disc1}
    T_{\gamma}(\{\deltaB_v^1 +\deltaB_v^{\sigma_1}\})= \includegraphics[scale=0.77,valign=c]{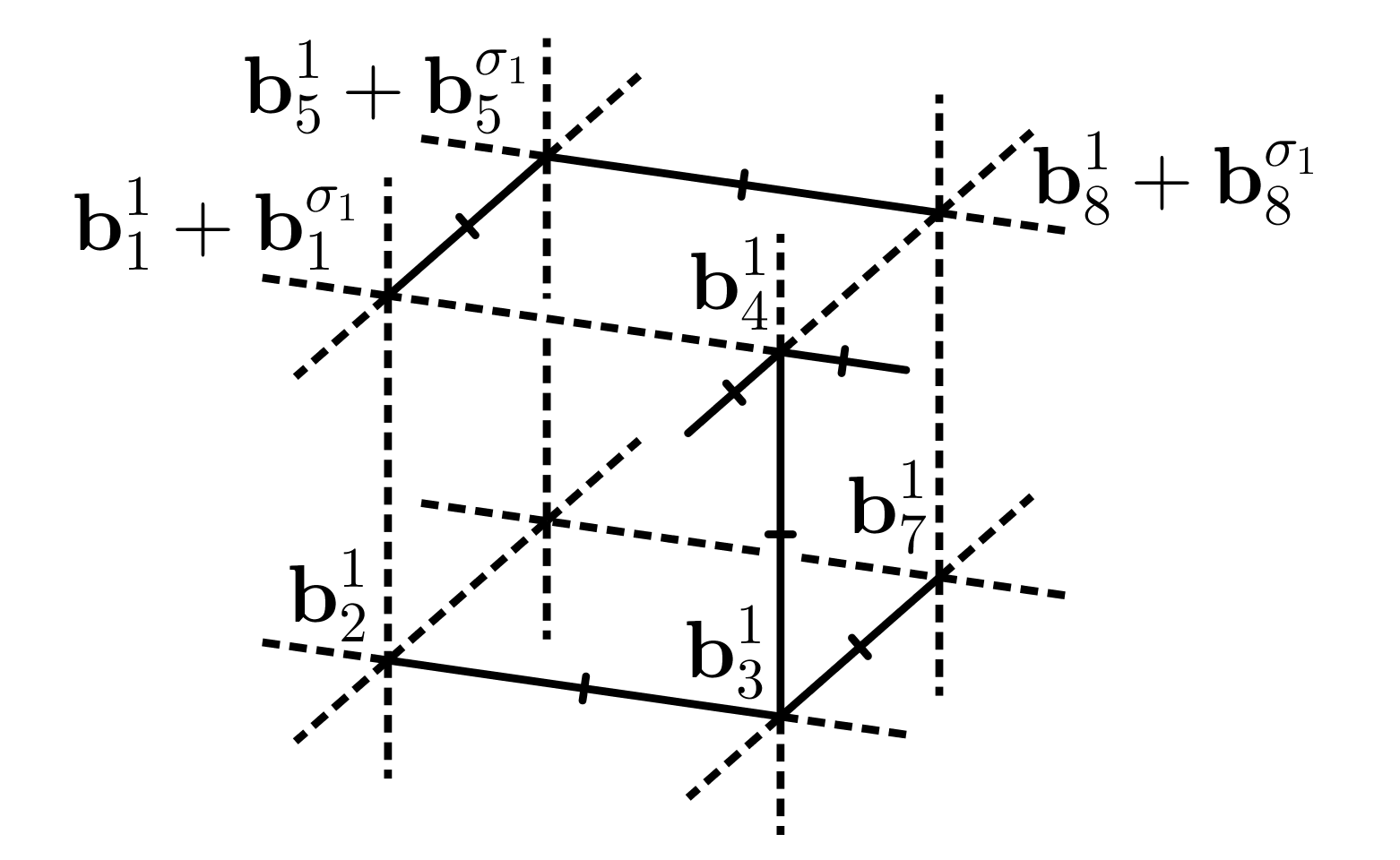}.
\end{equation}
We expand \eqref{NSrc_ex_disc1} by multilinearity of contractions and note that there are two nonzero terms:
\begin{equation}\label{NSrc_ex_disc3}
    T_\gamma(\{\deltaB_v^1 +\deltaB_v^{\sigma_1}\})=\includegraphics[scale=0.77,valign=c]{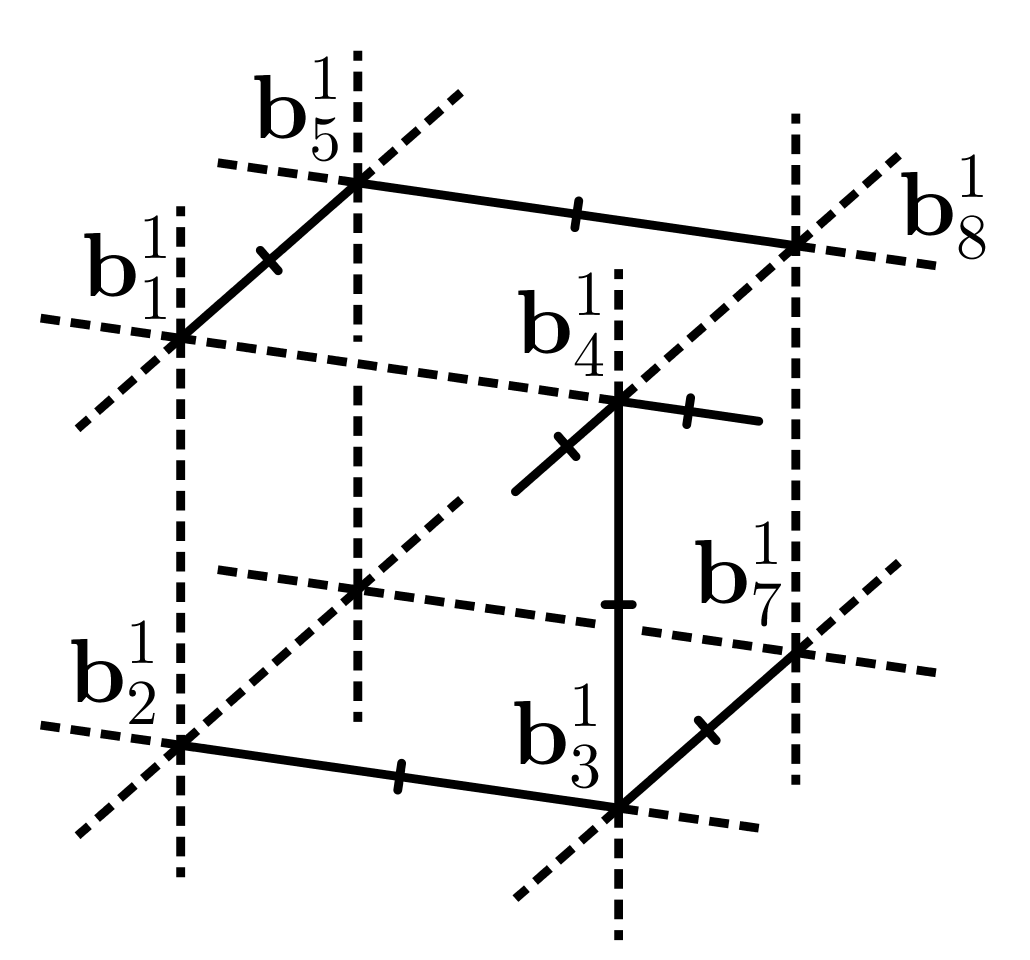}+\includegraphics[scale=0.77,valign=c]{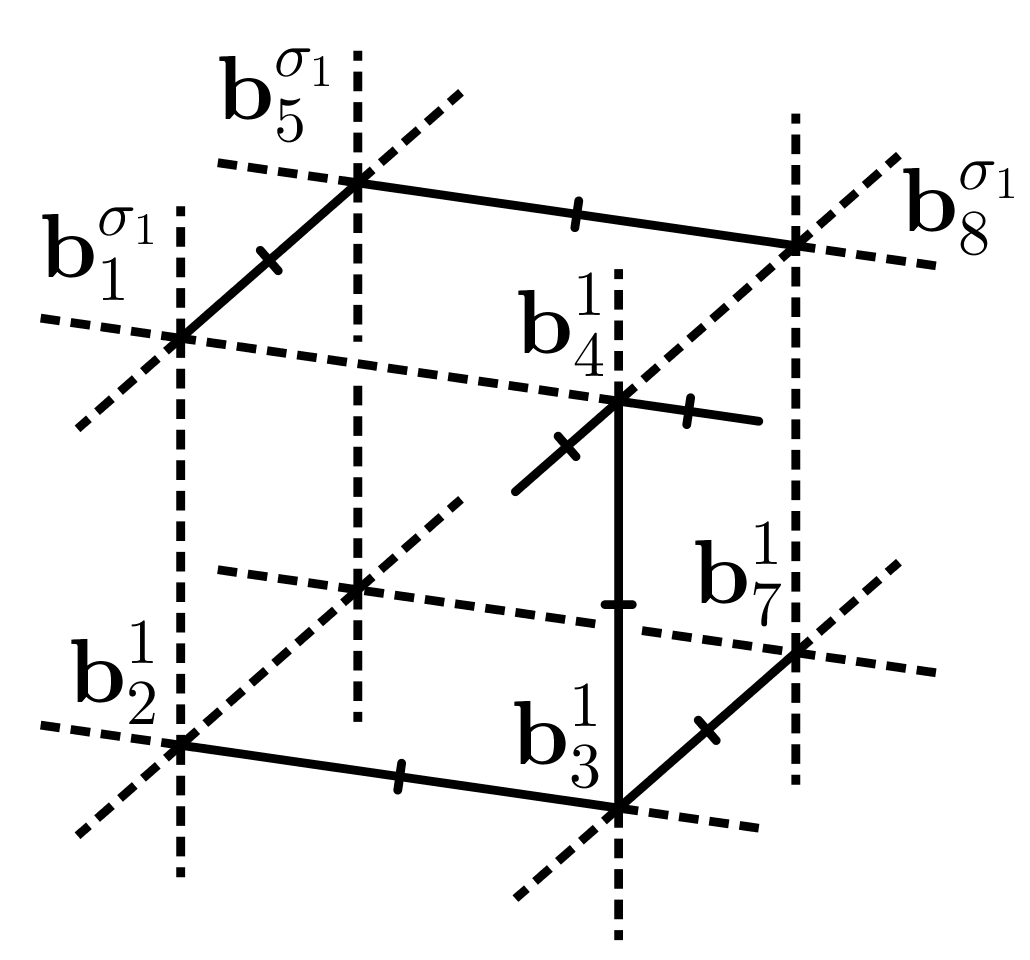}
\end{equation}
All other terms vanish as the contraction of $\deltaB_v^1$ and $\deltaB_v^{\sigma_1}$ along a ticked line vanishes due to their internal ticked legs being restricted to disjoint sets $\bigsqcup\limits_{\mu \in \mathrm{Src}} \N_\mu$ and $\{1_{\sigma_1}\}$, respectively. By \eqref{dbv1_goal}, the first term in \eqref{NSrc_ex_disc3} equals $T_\gamma (\deltaA)$. Consider the second term in \eqref{NSrc_ex_disc3}. Using \eqref{NSrc_def_18}, \eqref{NSrc_def_5}, we obtain:
\begin{equation}\label{NSrc_breaking_term}
    \includegraphics[scale=0.77,valign=c]{NSrc_tgamma3.png} = C_{\sigma_1} T_\lambda(\{\deltaB_v^1\})=C_{\sigma_1} T_\lambda(\deltaA),
\end{equation}
where $\lambda$ is given by \eqref{NSrc_lambda_diag}. This is the unwanted term mentioned earlier.

It can be shown that for generic $\deltaA$, there is no $\mu \in \mathrm{Disc}$ such that $T_\mu(\deltaA)= C_{\sigma_1} T_\lambda (\deltaA)$, and so, \eqref{desire} is violated for $\gamma$ given by \eqref{NSrc_gamma_def},\eqref{NSrc_gamma_diag}.\footnote{Strictly speaking, to show the contradiction with \eqref{desire}, one should prove that $C_{\sigma_1} T_\lambda (\deltaA)$ cannot be represented as a sum of contractions $T_{\mu} (\deltaA)$ with disconnected $\mu$'s. We omit this complicated discussion here, as our goal is to give a simple motivation for the upcoming modification of the naive definition, which, as we will see, establishes \eqref{desire}.} Let us give a rough explanation of this fact. Assume that there is a $\mu \in \mathrm{Disc}$ such that $T_\mu(\deltaA)=C_{\sigma_1} T_\lambda(\deltaA)$. Then, to get factor $C_{\sigma_1}$, $T_\mu (\deltaA)$ should insert $\deltaA$ in vertices $v \in\graph (\sigma_1)$ and restrict them according to $\sigma_1$. To get factor $T_\lambda (\deltaA)$, $T_\mu (\deltaA)$ should insert $\deltaA$ in vertices $v \in\graph (\lambda)$ and restrict them according to $\lambda$. Then, the tensor in vertex $4 \in\graph (\sigma_1) \cap\graph (\lambda)$ should be restricted simultaneously according to $\sigma_1$ and $\lambda$, which is impossible as there are bonds and legs $n^4_i$ such that $\sigma_1(n^4_i) \neq \lambda(n^4_i)$.\footnote{For the notation $n^v_k$ see \eqref{positions_of_ns} and the discussion around.}

We see from this discussion that a similar unwanted term will appear for any $\lambda$ in \eqref{NSrc_gamma_def} such that $4 \in\graph (\lambda)$ (in particular, $\lambda$ can be disconnected). If on the other hand $4 \notin\graph (\lambda)$, then it is easy to verify that $C_{\sigma_1} T_\lambda (\deltaA)=T_{\sigma_1 \star \lambda} (\deltaA)$ and so, such a term does not contradict \eqref{desire}.

We will cancel the unwanted terms by modifying the definition of $\deltaB_v^{\sigma_1}$ tensors as follows:
\begin{itemize}
    \item For $v=1,8$, we introduce a second tensor element for each $\deltaB^{\sigma_1}_v$, in addition to \eqref{NSrc_def_18}:
          \begin{equation}\label{NSrc_def_18(2)}
              \includegraphics[scale=0.75,valign=c]{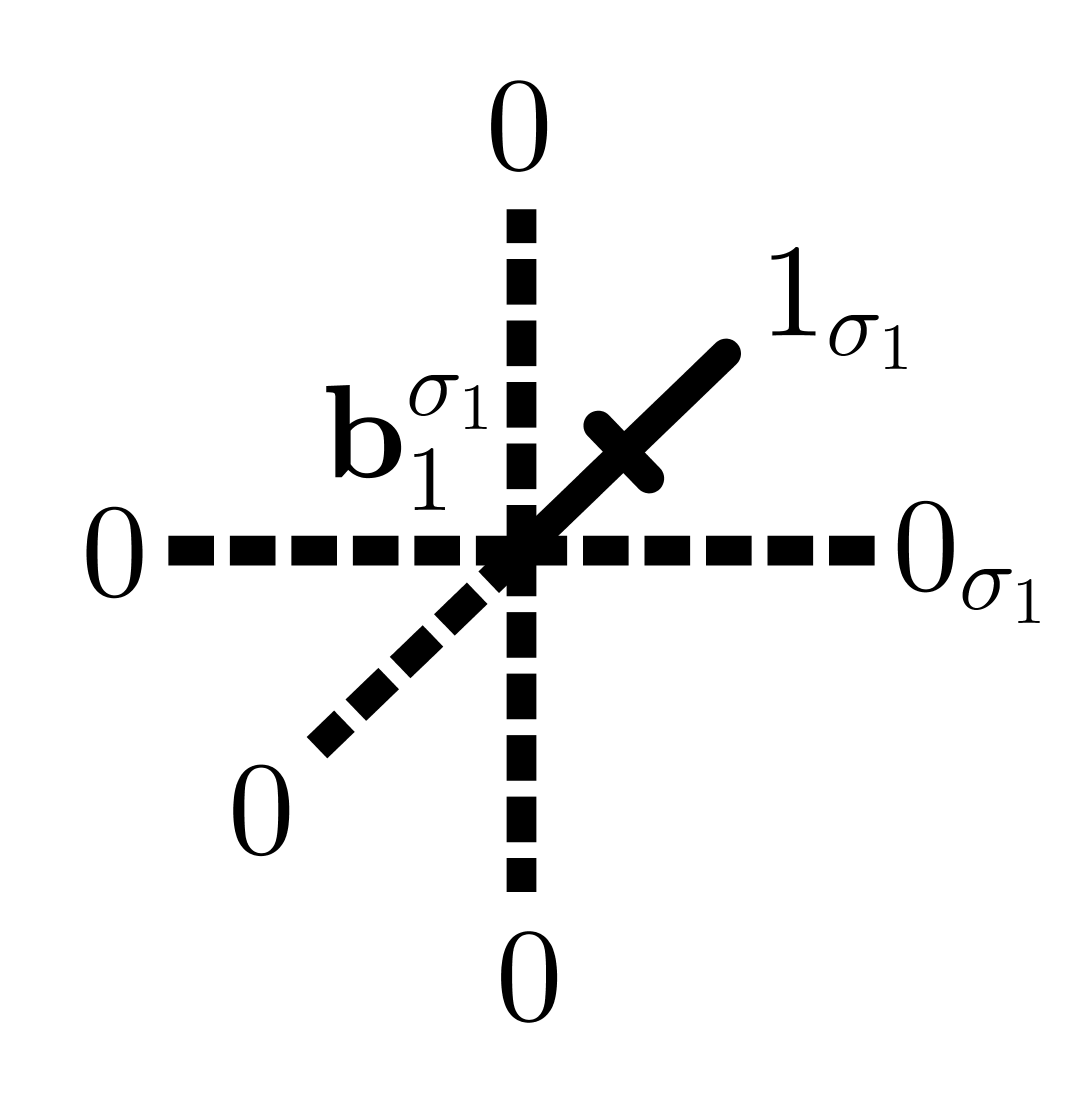}=\epsilon^{4/3}, \qquad \includegraphics[scale=0.75,valign=c]{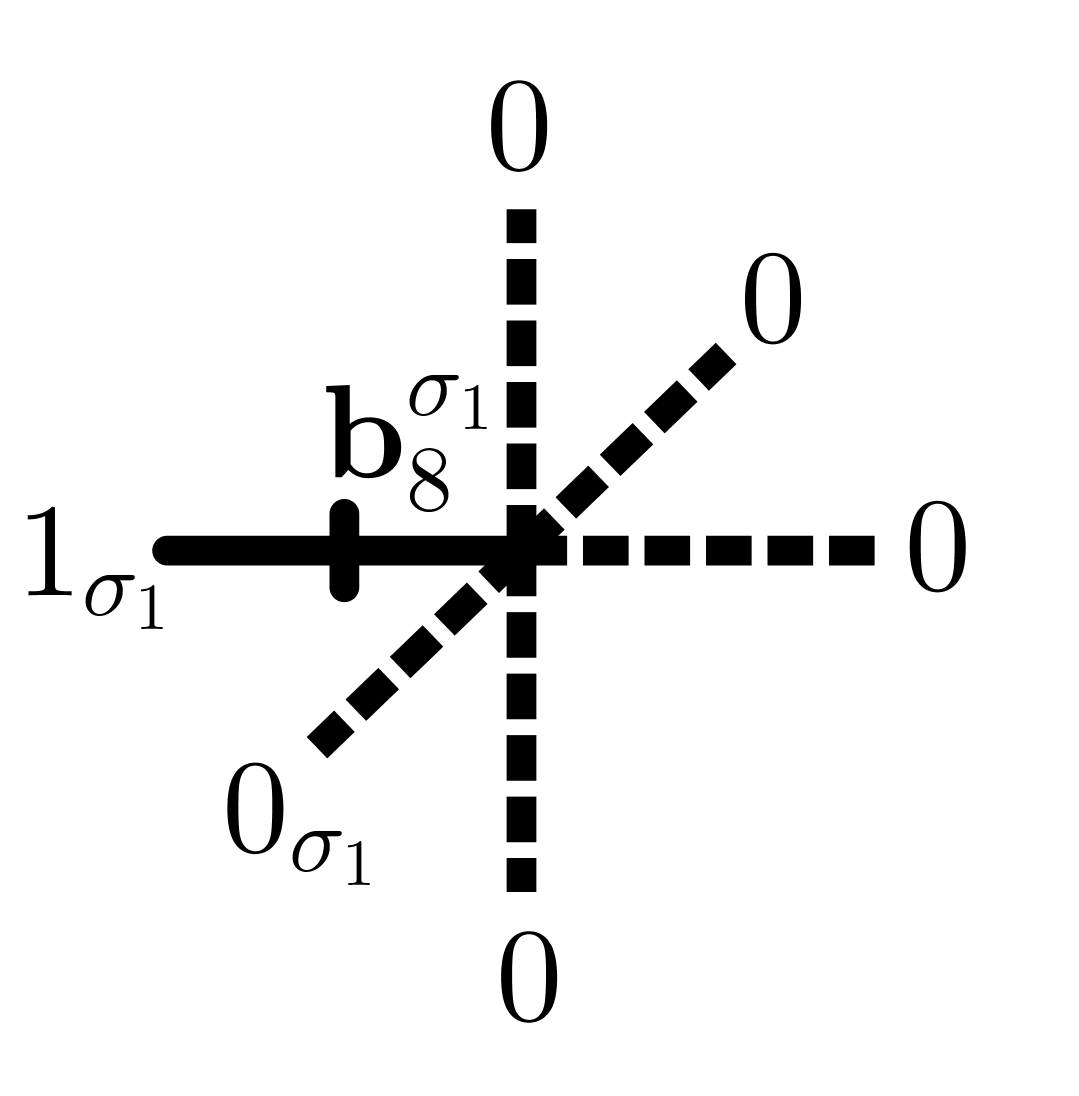}=\epsilon^{4/3},
          \end{equation}
          where $0_{\sigma_1}$ is an element of $\mathcal{O}$ introduced in Section~\ref{D0} (see also Footnote~\ref{footNg}). The difference between tensor elements in \eqref{NSrc_def_18} and \eqref{NSrc_def_18(2)} is that in the latter, legs which point towards vertex $4$ have index value $0_{\sigma_1}$. All nonzero tensor elements of $\deltaB_1^{\sigma_1}$ and $\deltaB_8^{\sigma_1}$ are given by \eqref{NSrc_def_18} and \eqref{NSrc_def_18(2)}.
    \item For $v=4$, we define $\deltaB_4^{\sigma_1}$ by the following formula:
          \begin{equation}\label{NSrc_ex_4}
              \includegraphics[scale=0.75,valign=c]{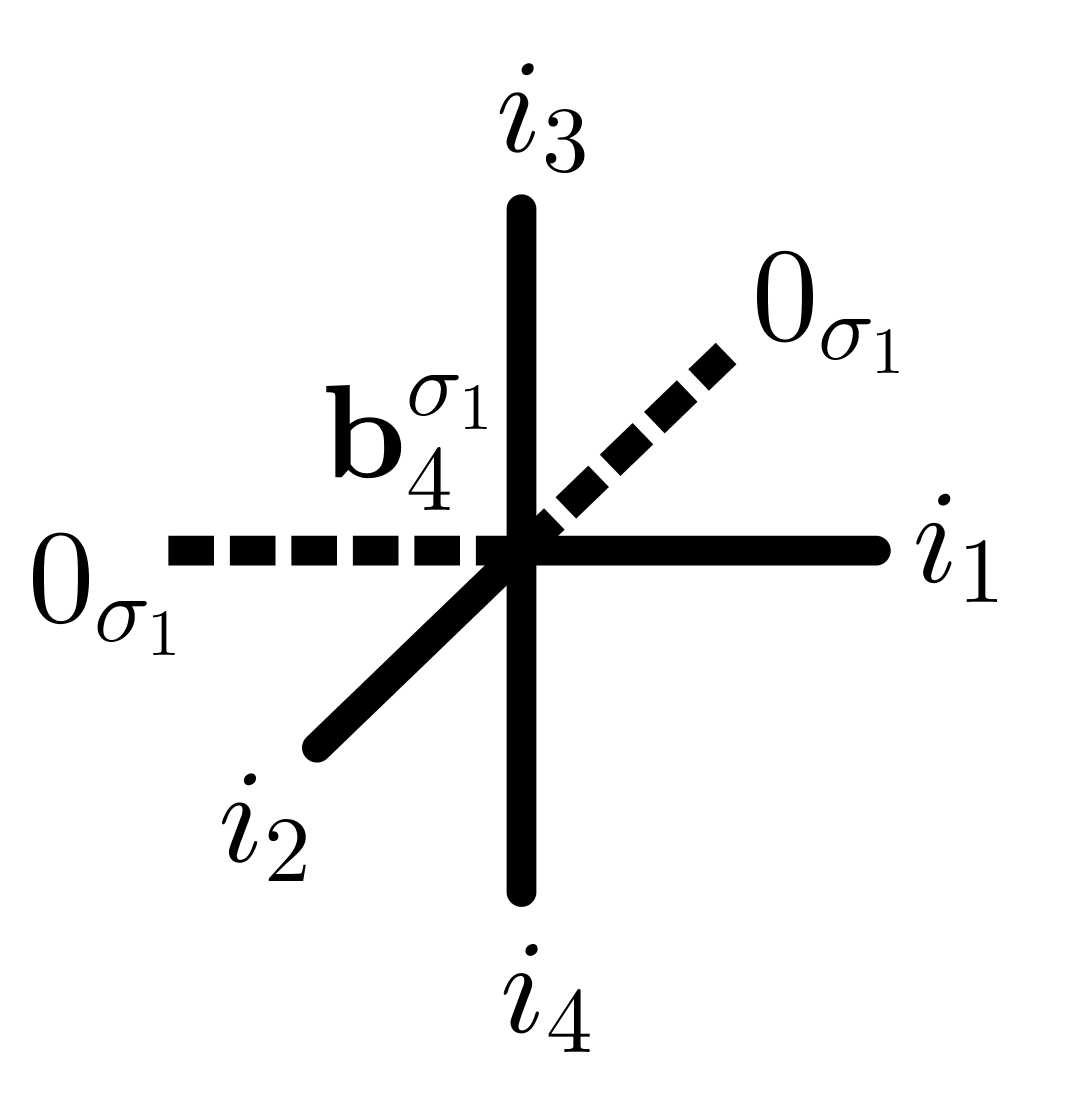}=-\left(\includegraphics[scale=0.75,valign=c]{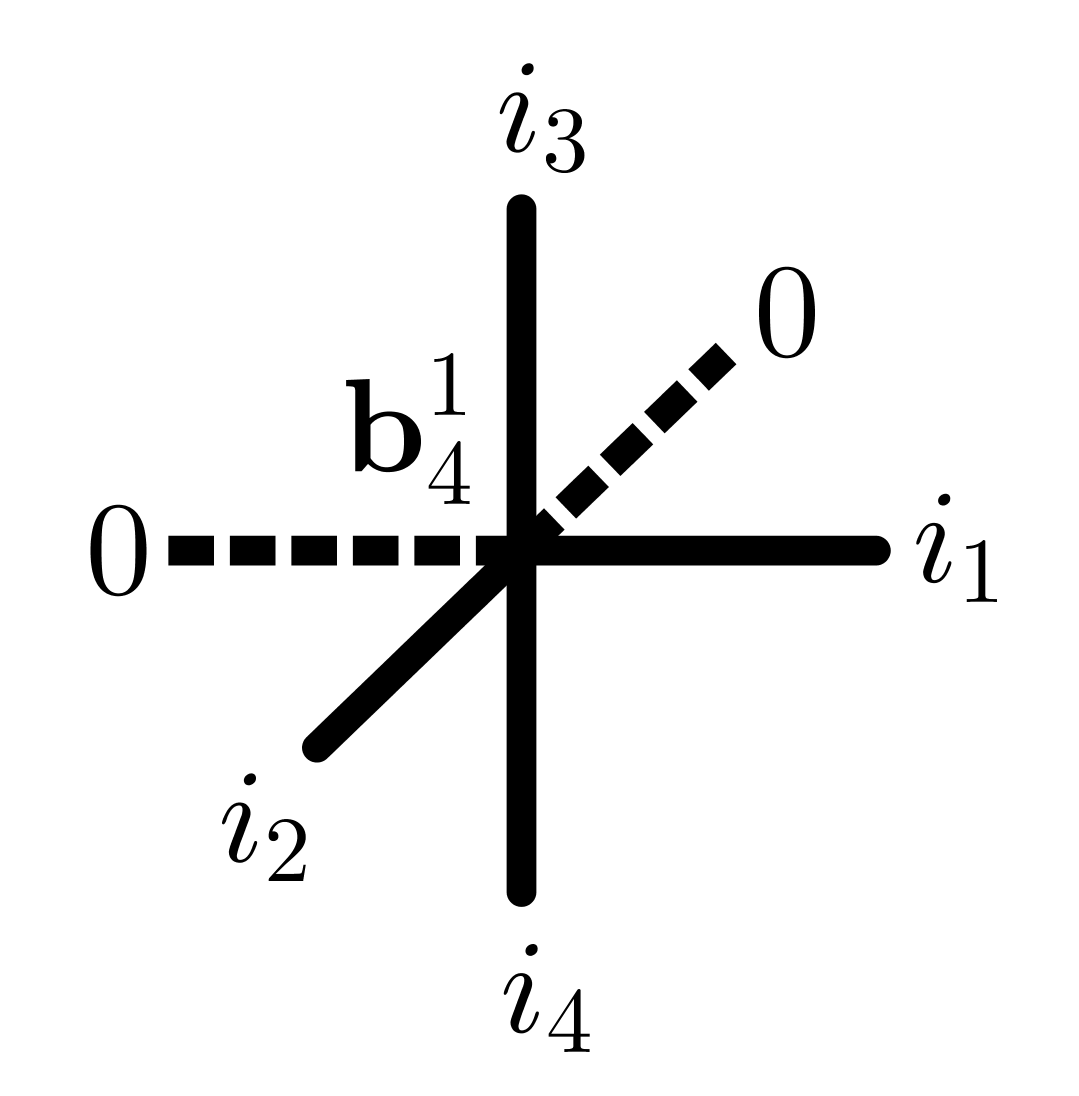}\right),
          \end{equation}
          where $\deltaB_4^1$ is the tensor defined in Section~\ref{dbv1}. Here, $i_1,i_2,i_3$ (external legs) run over $\N_0$, and $i_4$ (internal leg) runs over $\mathcal{D}_0$. Other tensor elements of $\deltaB_4^{\sigma_1}$ are zero. Note that the legs of $\deltaB_4^{\sigma_1}$ which point towards vertices $1$ and $8$ have index value $0_{\sigma_1}$, matching the index value of legs of $\deltaB_1^{\sigma_1}$ and $\deltaB_8^{\sigma_1}$ in \eqref{NSrc_def_18(2)}.\footnote{For our example, setting $i_1,i_2 \in \N$, $i_3=0$, and $i_4 \in \mathcal{D}$ in \eqref{NSrc_ex_4} would suffice to cancel the unwanted term in \eqref{NSrc_ex_disc3}. However, similar unwanted terms appear for all templates $\lambda$ in \eqref{NSrc_gamma_def} such that $4 \in\graph (\lambda)$. Different $\lambda$'s in \eqref{NSrc_gamma_def} may restrict the legs of $\deltaB_4^1$ in the corresponding unwanted terms in different ways. By using all possible $i_1,\ldots,i_4$ in \eqref{NSrc_ex_4}, we cover all possible cases.}
    \item Other tensors remain the same as in the naive definition.
\end{itemize}
It is easy to see that \eqref{Pconds} and the first line of \eqref{compat3} still hold after this modification. The second line of \eqref{compat3} is a bit trickier as $\deltaB_4^{\sigma_1}$ is not restricted according to $\widehat{\sigma}$. After the presentation of the general construction, we will show that $\deltaB_4^{\sigma_1}$ and its counterparts for other templates $\sigma_q \in \mathrm{NSrc}$ cannot appear in the contractions restricted according to connected templates\footnote{Expression "contraction restricted according to a template" is clarified in the discussion after \eqref{trm_dscrA}.} and, in particular, in $T_\gamma(\{\deltaB_v^{\sigma_1}\})$ in \eqref{compat3}. Thus, if \eqref{compat3} holds for the naive tensors, it should hold for the modified tensors as well.\footnote{Running ahead, we note that $\deltaB_4^{\sigma_1}$, though it has legs with index values in $\bigsqcup\limits_{\mu \in \mathrm{NSrc}} \N_\mu$, will not break the argument for \eqref{compat4} mentioned in the beginning of Section~\ref{CON2} as it will not appear in the contractions $T_\gamma (\{\deltaB_v^1+\deltaB_v^2\})$ with connected $\gamma$'s.}

Now we reconsider the example of $\gamma$ given by \eqref{NSrc_gamma_def}, \eqref{NSrc_gamma_diag}. We will show that after the above modification, we have the following:
\begin{equation}\label{nounwnt}
    T_\gamma(\{\deltaB_v^1+\deltaB_v^{\sigma_1}\})=T_\gamma(\deltaA),
\end{equation}
which coincides with \eqref{desire} for $D_\gamma =\{\gamma\}$.

The expansion by multilinearity of $T_\gamma (\{\deltaB_v^1 +\deltaB_v^{\sigma_1}\})$ now includes four terms:
\begin{equation}\label{NSrc_disc_4}
    T_\gamma (\{\deltaB_v^1 +\deltaB_v^{\sigma_1}\})=T_1+T_2+T_3+T_4
\end{equation}
The first two terms $T_1,T_2$ are identical to those in \eqref{NSrc_ex_disc3}:
\begin{equation}
    T_1=\includegraphics[scale=0.85,valign=c]{NSrc_tgamma2.png}, \qquad T_2=\includegraphics[scale=0.85,valign=c]{NSrc_tgamma3.png}.
\end{equation}
The third and the fourth terms $T_3,T_4$ are analogous to $T_1,T_2$, respectively, but with $\deltaB_4^1$ replaced with $\deltaB_4^{\sigma_1}$:
\begin{equation}\label{T4ref}
    T_3=\includegraphics[scale=0.85,valign=c]{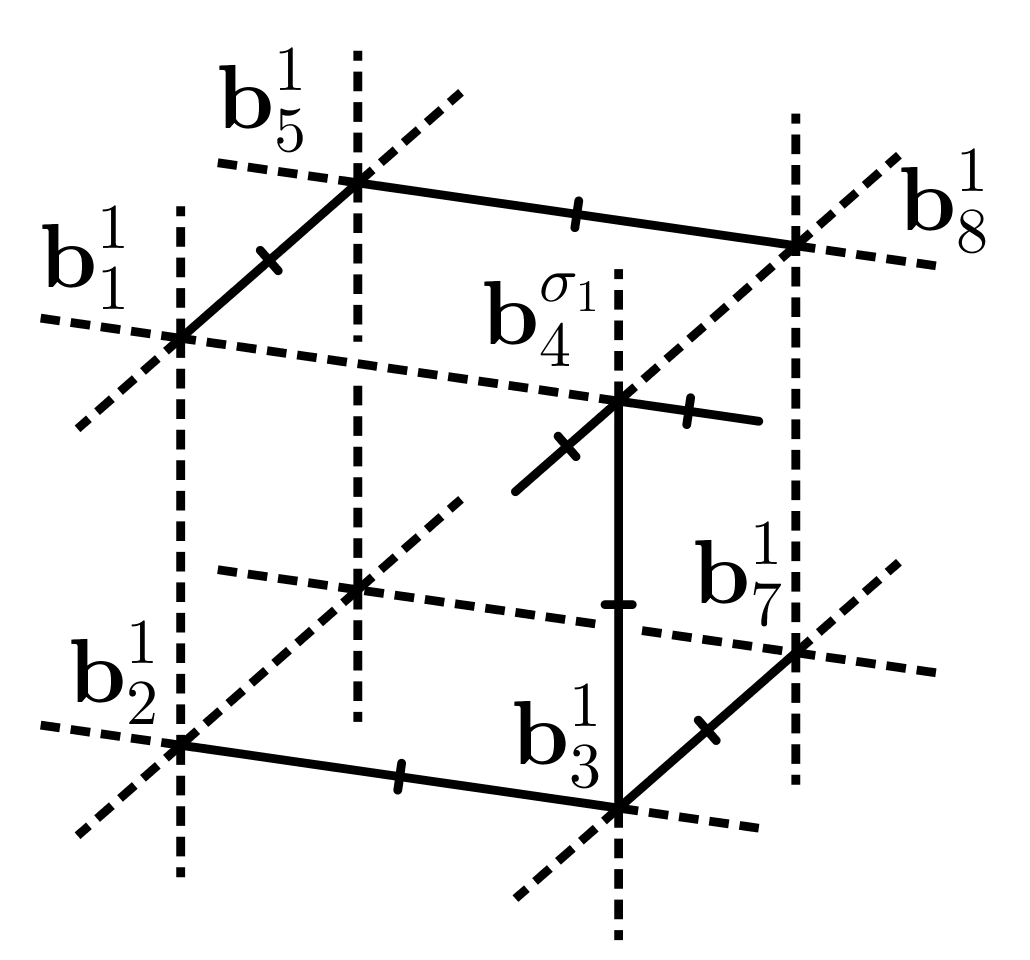}, \qquad T_4=\includegraphics[scale=0.85,valign=c]{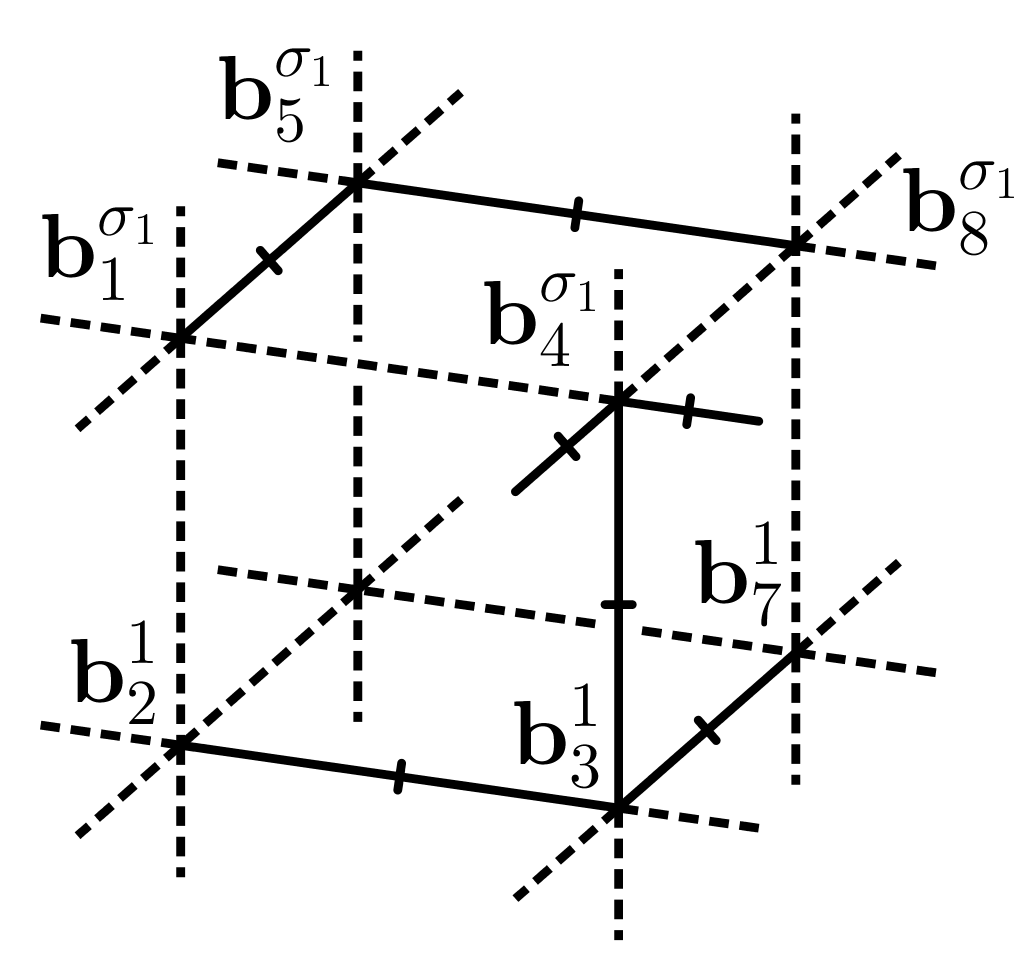}.
\end{equation}
Note that $T_3$ vanishes as the contracted legs of $\deltaB_4^{\sigma_1}$ and $\deltaB_1^1$ are restricted to disjoint sets $\{0_{\sigma_1}\}$ and $\{0\}$, respectively. For $T_4$, using \eqref{NSrc_def_18(2)},\eqref{NSrc_ex_4}, we find:
\begin{equation}\label{NSrc_cancel}
    T_4= -T_2.
\end{equation}
Then, $T_4$ cancels $T_2$ (the unwanted term), leaving only $T_1$ on the r.h.s. of \eqref{NSrc_disc_4}. By \eqref{dbv1_goal}, $T_1=T_\gamma(\deltaA)$. Therefore, we have \eqref{nounwnt}.

We conclude our example by checking that \eqref{desire}, \eqref{desire2} (and so \eqref{compat5} for $M=1$) hold for any $\gamma \in \mathrm{Disc}$. We consider three possible relations between generic $\gamma \in \mathrm{Disc}$ and $\widehat{\sigma_1}$:
\begin{subequations}\label{gammamu}
    \begin{align}
         & \text{$\widehat{\sigma_1}$ is not a connected component of $\gamma$;\footnotemark}\label{gammamuA}              \\
         & \text{$\widehat{\sigma_1}$ is a connected component of $\gamma$ and $4 \in\graph (\gamma)$;}\label{gammamuB}    \\
         & \text{$\widehat{\sigma_1}$ is a connected component of $\gamma$ and $4 \notin\graph (\gamma)$.}\label{gammamuC}
    \end{align}
\end{subequations}
Note that \eqref{gammamuA}-\eqref{gammamuC} are mutually exclusive and that each $\gamma \in \mathrm{Disc}$ satisfies one of these relations. We will show that in all cases \eqref{gammamuA}-\eqref{gammamuC}, \eqref{desire} holds with specific disjoint sets $D_\gamma$. Then, we will demonstrate that \eqref{desire2} holds. We will keep details to a minimum, as a similar discussion will follow in the general construction. \footnotetext{See the discussion after \eqref{unionExample} for the definition of a connected component of a template.}

\begin{itemize}
    \item[\eqref{gammamuA}] In this case, one can show that in the expansion by multilinearity of $T_\gamma(\{\deltaB_v^1+\deltaB_v^{\sigma_1}\})$, terms with insertions of $\deltaB_v^{\sigma_1}$ vanish. Thus,
        \begin{equation}\label{NSrc_nuhatnot}
            T_\gamma(\{\deltaB_v^1+\deltaB_v^{\sigma_1}\})=T_\gamma(\{\deltaB_v^1\})=[\text{by} \ \eqref{dbv1_goal}] \ \begin{cases}
                T_\gamma (\deltaA), & \text{if } \gamma \in \mathrm{DiscS}\footnotemark \\
                0,                  & \text{if } \gamma \notin \mathrm{DiscS}
            \end{cases}.
        \end{equation}
        Therefore, for $\gamma$ satisfying \eqref{gammamuA}, we have \eqref{desire} with:
        \begin{equation}\label{DgammaA}
            D_\gamma = \begin{cases}
                \{\gamma\}, & \text{if } \gamma \in \mathrm{DiscS}    \\
                \void,      & \text{if } \gamma \notin \mathrm{DiscS} \\
            \end{cases}.
        \end{equation}
        \footnotetext{See the definition of $\mathrm{DiscS}$ in \eqref{Discs_def}.}

    \item[\eqref{gammamuB}] This case was previously addressed with the example of $\gamma$ given by \eqref{NSrc_gamma_def}, \eqref{NSrc_gamma_diag}. For a generic $\gamma$ satisfying \eqref{gammamuB}, a derivation analogous to the one below \eqref{nounwnt} yields \eqref{NSrc_nuhatnot}. Therefore, for $\gamma$ satisfying \eqref{gammamuB}, we have \eqref{desire} with $D_\gamma$ given by \eqref{DgammaA}.

    \item[\eqref{gammamuC}] We can uniquely write each such $\gamma$ as $\gamma = \widehat{\sigma_1}\star \lambda$, where $\lambda$ is not necessarily connected. Note that, as $4 \notin\graph (\gamma)$, we have $4 \notin\graph (\lambda) \subset\graph (\gamma)$. The expansion by multilinearity of $T_\gamma (\{\deltaB_v^1+\deltaB_v^{\sigma_1}\})$ yields two terms similar to those in \eqref{NSrc_ex_disc3}:
        \begin{equation}\label{NSrc_ex_disc_5}
            T_\gamma(\{\deltaB_v^1+\deltaB_v^{\sigma_1}\})= T_\gamma (\{\deltaB_v^1\})+C_{\sigma_1} T_\lambda (\{\deltaB_v^1\}).
        \end{equation}
        Note that vertices $1$ and $8$ in $\widehat{\sigma_1}$ given by \eqref{NSrc_hatsigma_diag} are sources (see \eqref{source2}), and so $\widehat{\sigma_1} \in \mathrm{Src}$. Then, we have $\gamma = \widehat{\sigma_1} \star \lambda \in \mathrm{DiscS}$ if and only if $\lambda \in \mathrm{DiscS} \sqcup \mathrm{Src}$. Then, applying \eqref{dbv1_goal} to \eqref{NSrc_ex_disc_5}, we obtain:
        \begin{equation}
            T_\gamma(\{\deltaB_v^1+\deltaB_v^{\sigma_1}\})=\begin{cases}
                T_\gamma(\deltaA) + C_{\sigma_1} T_\lambda (\deltaA), & \text{if } \lambda \in \mathrm{DiscS} \sqcup \mathrm{Src}    \\
                0,                                                    & \text{if } \lambda \notin \mathrm{DiscS} \sqcup \mathrm{Src}
            \end{cases}
        \end{equation}
        As mentioned earlier\footnote{See the discussion after \eqref{NSrc_breaking_term}}, if $4 \notin\graph (\lambda)$, then $C_{\sigma_1} T_\lambda (\deltaA)=T_{\sigma_1 \star \lambda} (\deltaA)$. Therefore, for $\gamma=\widehat{\sigma_1} \star \lambda$ satisfying \eqref{gammamuC}, we have \eqref{desire} with:
        \begin{equation}
            D_{\widehat{\sigma_1} \star \lambda} = \begin{cases}\label{DgammaC}
                \{\widehat{\sigma_1} \star \lambda, \sigma_1 \star \lambda \}, & \text{ if } \lambda \in \mathrm{DiscS}\sqcup \mathrm{Src}    \\
                \void,                                                         & \text{ if } \lambda \notin \mathrm{DiscS}\sqcup \mathrm{Src}
            \end{cases}.
        \end{equation}
\end{itemize}

Thus, we see that \eqref{desire} holds for any $\gamma \in \mathrm{Disc}$. Clearly, $D_\gamma$'s from above are disjoint as required. Finally, we verify \eqref{desire2}:
\begin{equation}
    \begin{aligned}
        \bigsqcup_{\gamma \in \mathrm{Disc}} D_\gamma & =                                                                                                                                                                                                                                              \\
                                                      & =[\text{by \eqref{DgammaA}, \eqref{DgammaC}}] \ \bigsqcup_{\gamma \in \mathrm{DiscS}} \{\gamma\} \sqcup \bigsqcup_{\lambda \in \mathrm{DiscS} \sqcup \mathrm{Src}} \{\sigma_1 \star \lambda\}= [\text{by \eqref{DiscMdef}}] \ \mathrm{Disc}_1.
    \end{aligned}
\end{equation}

In this way, we constructed tensors $\deltaB_v^{\sigma_1}$ which satisfy \eqref{Pconds}, \eqref{compat3} and \eqref{compat5} for $M=1$.

\subsubsection{General construction}\label{dbv2genconstr}

\paragraph{Construction of $\widehat{\sigma}$} Let us start by providing the construction of $\widehat{\sigma}$ required in \eqref{compat3}.

Note that if all legs in $\diag(\mu)$ are dashed, we can identify a non-trivial template $\mu$ with its graph $\graph (\mu)$. In such cases, we will say that $\mu$ is \textbf{given by the graph} $\graph (\mu)$. In particular, each $\sigma \in \mathrm{NSrc}$ is given by the corresponding graph $\graph (\sigma)$.

We write $\graph (\sigma)-v$ for the subgraph of $\graph (\sigma)$ obtained by deleting the vertex $v$ and edges for which $v$ is an endpoint. For a connected graph $\graph (\sigma)$, a vertex $v$ is called a \textbf{cut-vertex} if $\graph (\sigma)-v$ is disconnected.\footnote{This is a special case of the general definition of a cut-vertex (suitable for any graph) from \cite{west_introduction_2000}.}

Now, let $\sigma \in \mathrm{NSrc} \subset \mathrm{Con}$. We choose and fix a vertex $w_\sigma \in\graph (\sigma)$ such that it is not a cut-vertex. [We can always find such a
$w_\sigma \in\graph (\sigma)$, see Proposition~1.2.29 in \cite{west_introduction_2000}.] Then, $\widehat{\sigma}$ is defined as the template given by the graph
\begin{equation}\label{hatdef1}
    \graph (\widehat{\sigma})=g(\sigma)-w_\sigma.
\end{equation}
As $w_\sigma$ is not a cut-vertex, $\graph (\widehat{\sigma})$ is connected and so $\widehat{\sigma} \in \mathrm{Con}$ as required by \eqref{compat3}.

In other words, we obtain $\diag(\widehat{\sigma})$ from $\diag(\sigma)$ by replacing all ticked bonds which touch the vertex $w_\sigma$ in $\diag(\sigma)$ with dashed ones. We will refer to $w_\sigma$ in \eqref{hatdef1} as the \textbf{special vertex} of $\sigma$.

\paragraph{Construction of $\deltaB_v^{\sigma_q}$} We number the elements of $\mathrm{NSrc}$ in some arbitrary fixed order (see \eqref{NSrc_ord}).\footnote{As already mentioned in Footnote~\ref{ftntnord}, in the general construction $\sigma_1$ can be any $\mathrm{NSrc}$ template, not necessarily the one from $q=1$ example.} Then, we construct tensors $\deltaB_v^{\sigma_q}$ in the increasing order of $q=1,\ldots, |\mathrm{NSrc}|$ as follows.
\begin{enumerate}
    \item\label{itm:step0dbv2} Let $N_q$ denote the number of vertices in $\graph (\sigma_q)$. [Note that $N_q \geq 4$ for any template $\sigma_q \in \mathrm{NSrc}$. Below, we will only use $N_q>1$.]
    \item As $\sigma_q \in \mathrm{NSrc}$, all legs of $T_{\sigma_q} (\deltaA)$ are restricted to $\{0\}$, and so $T_{\sigma_q} (\deltaA) \propto T_*$. Then, we define a scalar quantity $C_{\sigma_q}=C_{\sigma_q} (\deltaA)$ by the following formula:
          \begin{equation}\label{CsDEF}
              T_{\sigma_q} (\deltaA) = C_{\sigma_q} T_*.
          \end{equation}
          Note that $C_{\sigma_q}$ depends analytically on $\deltaA$, and that:
          \begin{equation}\label{CsigmaORD}
              C_{\sigma_q} = O (\epsilon^{N_q}).
          \end{equation}

          \item\label{itm:notations} Recall that $w_{\sigma_q}$ denotes the special vertex of $\sigma_q$, such that $\graph (\widehat{\sigma_q})=g(\sigma_q)-w_{\sigma_q}$, see \eqref{hatdef1}. To simplify the notation, we will write $w_q$ for $w_{\sigma_q}$. We denote by $V_q$ the set of neighbours of $w_q$ in $\graph (\sigma_q)$.\footnote{In the $q=1$ example, $V_1=\{1,8\} $ and $w_1=4$.} Let $S_{q}$ be the set of edges which were deleted from $\graph (\sigma_q)$ to get $\graph (\widehat{\sigma_q})$ (i.e., the edges of $\graph (\sigma_q)$ connecting $w_q$ with vertices in $V_q$). Note that, as $\sigma_q \in \mathrm{NSrc}$, $w_q\in\graph (\sigma_q)$ has degree no less than $2$ and so $|S_q|= |V_q| \geq 2$.

          \item\label{dbv2DEFS} For each $v \in\graph (\widehat{\sigma_q})$ we define two sequences of indices:
          \begin{equation}\label{sequence_nn}
              I_v=ind_{1}(n^v_1),\ldots,ind_1(n^v_6), \text{ where } ind_1(n^v_i)=\begin{cases}
                  0,            & \text{if } \widehat{\sigma_q} (n^v_i)=\mathcal{O} \\
                  1_{\sigma_q}, & \text{if } \widehat{\sigma_q} (n^v_i)=\mathcal{D}
              \end{cases}.
          \end{equation}
          \begin{equation}\label{sequence_n}
              J_v=ind_2(n^v_1),\ldots,ind_2(n^v_6), \text{ where } ind_2(n^v_i) = \begin{cases}
                  0,            & \text{if } \widehat{\sigma_q} (n^v_i)=\mathcal{O} \text{ and } n^v_i \notin S_q \\
                  0_{\sigma_q}, & \text{if } \widehat{\sigma_q} (n^v_i)=\mathcal{O} \text{ and } n_{i}^v\in S_q   \\
                  1_{\sigma_q}, & \text{if } \widehat{\sigma_q} (n^v_i)=\mathcal{D}
              \end{cases}.
          \end{equation}
          The notation $n^v_i$ was introduced in the discussion around \eqref{positions_of_ns}.

          In other words, an index in $I_v$ has value $0$ (resp. $1_{\sigma_q}$) if the corresponding leg is restricted by $\widehat{\sigma_{q}}$ to $\mathcal{O}$ (resp. $\mathcal{D}$). The sequence $J_v$ can be obtained from $I_v$ by replacing the index value $0$ in legs pointing towards vertex $w_q$ with $0_{\sigma_q}$ (these are exactly the legs $n^v_i \in S_q$).

          \item\label{itm:defsofBs} We distinguish three types of vertices $v \in\graph (\sigma_q)$:
          \begin{subequations}\label{vertextypes}
              \begin{align}
                   & \text{$v \in V_q$;}\label{vtypeA}                    \\
                   & \text{$v\notin V_q$ and $v \neq w_q$;}\label{vtypeB} \\
                   & v=w_{q}.\label{vtypeC}
              \end{align}
          \end{subequations}
          For each case \eqref{vtypeA}-\eqref{vtypeC}, we define $\deltaB_v^{\sigma_q}$ as follows.
          \begin{itemize}
              \item[\eqref{vtypeA}] We define $\deltaB_v^{\sigma_q}$ as the tensor restricted according to $\widehat{\sigma_q}$ with two nonzero tensor elements:
                  \begin{subequations}\label{vA}
                      \begin{align}
                          \left(\deltaB_v^{\sigma_q}\right)_{I_v}=\epsilon^{N_q/(N_q-1)};\label{vA1} \\
                          \left(\deltaB_v^{\sigma_q}\right)_{J_v}=\epsilon^{N_q/(N_q-1)}\label{vA2}.
                      \end{align}
                  \end{subequations}
                  [Recall that $N_q>1$ for any $\sigma_q \in \mathrm{NSrc}$.]

              \item[\eqref{vtypeB}] First, we arbitrarily choose a single vertex $v_*$ satisfying \eqref{vtypeB}. Then, for a generic vertex $v$ of this type, we define $\deltaB_v^{\sigma_q}$ as the tensor restricted according to $\widehat{\sigma_q}$ with a single nonzero tensor element:\footnote{In the $q=1$ example, $v_*=5$. We did not need to make a choice, as vertex $5$ was the only vertex satisfying \eqref{vtypeB}.}
                  \begin{equation}\label{vBdef}
                      (\deltaB_v^{\sigma_q})_{I_v}=\begin{cases}
                          \epsilon^{N_q/(N_q-1)},                  & \text{if } v \neq v_* \\
                          C_{\sigma_q} \epsilon^{N_q/(N_q-1)-N_q}, & \text{if } v=v_*
                      \end{cases},
                  \end{equation}
                  Note that, by \eqref{CsigmaORD}, we have:
                  \begin{equation}
                      C_{\sigma_q} \epsilon^{N_q/(N_q-1) -N_q}=O(\epsilon^{N_q/(N_q-1)}).
                  \end{equation}
                  Eq.~\eqref{vBdef} achieves two goals: together with \eqref{vA} it will make tensors $\deltaB_v^{\sigma_q}$ satisfy the first line of \eqref{compat3}; it ensures that the dependence of $\deltaB_v^{\sigma_q}$ on $C_{\sigma_q}$ (and so on $\deltaA$) is analytic. [The tensors which we will define for the case \eqref{vtypeC} will solve other problems.]

              \item[\eqref{vtypeC}] Note that, as we define $\deltaB_v^{\sigma_q}$ in increasing order of $q$, tensors $\deltaB_v^{\sigma_p}$ with $p=1,\ldots,q-1$ are already defined at this point. Then, we define tensor $W_{q}$ as:\footnote{In $q=1$ example, $W_1= \deltaB_{4}^1$.}
                  \begin{equation}\label{Wqdef}
                      W_q=\deltaB_{w_q}^1+\sum_{p=1}^{q-1} \deltaB_{w_q}^{\sigma_p}.
                  \end{equation}
                  Next, let $\mathcal{L}_{q}$ be the set of all index sequences which can label the tensor elements of a $6$-tensor inserted in the special vertex $w_q$. Namely, $\mathcal{L}_{q}$ consists of sequences $L=l_1,\ldots,l_6$ such that:
                  \begin{subequations}\label{genLdef}
                      \begin{align}
                           & \text{$l_i \in \N_0$ if the corresponding leg $n^{w_q}_i$ is external;}\label{LdefA}     \\
                           & \text{$l_i \in \mathcal{D}_0$ otherwise, i.e., if $n^{w_q}_i$ is internal.}\label{LdefC}
                      \end{align}
                  \end{subequations}
                  Then, we define a set $\bar{\mathcal{L}_q} \subset \mathcal{L}_q$ as the set of sequences  $L=l_1,\ldots,l_6$ satisfying the following condition:
                  \begin{equation}\label{Ldef}
                      \text{$L \in \mathcal{L}_q$ and $l_i=0$ if $n^{w_q}_i \in S_q$.}
                  \end{equation}
                  Recall that $n_i^{w_q} \in S_q$ are the internal legs pointing towards vertices belonging to $V_q$.

                  For each sequence $L\in \bar{\mathcal{L}_q}$ we define the sequence $L'$ by the following formula:
                  \begin{equation}\label{Lprimeseq}
                      L'=(l_1)',\ldots, (l_6)', \text{ where } (l_i)'=\begin{cases}
                          l_i, \text{ if } n^{w_q}_i \notin S_q \\
                          0_{\sigma_q}, \text{ if } n^{w_q}_i \in S_q.
                      \end{cases}.
                  \end{equation}
                  Finally, we define $\deltaB_{w_q}^{\sigma_q}$ by the following formula:
                  \begin{equation}\label{ErrTrmDef}
                      (\deltaB_{w_q}^{\sigma_q})_{L'}=-(W_q)_{L} \qquad (L \in \bar{\mathcal{L}_q}).
                  \end{equation}
                  Other tensor elements of $\deltaB_{w_q}^{\sigma_q}$ are zero. Note that
                  \begin{equation}\label{speciallegszeros}
                      \parbox{10cm}{the legs of $\deltaB_{w_q}^{\sigma_q}$ which point towards vertices in $V_q$ have index value $0_{\sigma_q}$,\footnotemark matching the index value of legs of $\deltaB_{v}^{\sigma_q}$ in \eqref{vA2}.}
                  \end{equation}
                  Let us emphasise that $\deltaB_{w_q}^{\sigma_q}$ may have nonzero tensor elements $(\deltaB_{w_q}^{\sigma_q})_{L'}$ where some $l_i' \in L'$ take arbitrary values from $\mathcal{D}$. This feature distinguishes $\deltaB_{w_q}^{\sigma_q}$ from other tensors defined in this paper, whose legs are typically restricted to proper subsets of $\mathcal{D}$ or $\mathcal{O}$. We will refer to $b_{w_q}^{\sigma_q}$, $q=1,\ldots,|\mathrm{NSrc}|$, as \textbf{special tensors} and to all other tensors as \textbf{nonspecial tensors}. \footnotetext{These are exactly the legs $n_i^w \in S_q$.}

                  Eq.~\eqref{ErrTrmDef} is the generalisation of \eqref{NSrc_ex_4} from the $q=1$ example. The only difference in $q>1$ case is that the unwanted terms\footnote{See the discussion about the unwanted terms in the $q=1$ example after \eqref{NSrc_breaks}.} inserts in $w_q$ all previously constructed tensors: $\deltaB_{w_q}^1$ and $\deltaB_{w_q}^{\sigma_p}, p<q$. Thus, to cancel those, we have to use $W_q$ in \eqref{ErrTrmDef}.

                  The definition of $\deltaB_{w_q}^{\sigma_q}$ may look quite complicated, but the idea behind it is straightforward. For each $q=1,\ldots, |\mathrm{NSrc}|$ we will encounter the unwanted terms which insert $W_q$ in $w_q$ and $\deltaB_v^{\sigma_q}$ in vertices $v \in V_q$ (E.g. see \eqref{NSrc_breaking_term}). To eliminate these terms, we introduce the tensors $\deltaB_{w_q}^{\sigma_q}$ and ensure they can appear only in the analogues of unwanted terms where $W_q$ is replaced with $\deltaB_{w_q}^{\sigma_q}$ (as $T_4$ in \eqref{NSrc_disc_4}, \eqref{T4ref}). This is achieved by using the index value $0_{\sigma_q}$ in \eqref{Lprimeseq}, which ensures that $\deltaB_{w_q}^{\sigma_q}$ can be inserted in $w_q$ only if there are $\deltaB_v^{\sigma_q}$ in vertices $v \in V_q$ as there are no other tensors whose legs can take value $0_{\sigma_q}$ (see Lemma~\ref{lem:nocrossvals} below). Here and below, we will use the following terminology:
                  \begin{itemize}
                      \item "leg of a tensor can take value $x$" means that there are potentially nonzero tensor elements with the corresponding index value;
                      \item "leg of a tensor cannot take value $x$" means that the corresponding tensor elements are guaranteed to be zero.
                  \end{itemize}

          \end{itemize}

    \item For $v\notin\graph (\sigma_q)$ we define:
          \begin{equation}\label{ZerosDef}
              \deltaB_v^{\sigma_q}=0.
          \end{equation}

\end{enumerate}

We now present the key properties of tensors $\deltaB_v^{\sigma_q}$ crucial for verification of \eqref{compat3},\eqref{compat4}, and \eqref{compat5}.

\begin{lem}\label{lem:nocrossvals}
    Tensors $\deltaB_{w_q}^{\sigma_q}$ and $\deltaB_v^{\sigma_q}$ with $v \in V_q$ are the only tensors among $\deltaB_v^1$, $\deltaB_v^{\sigma_p}$ ($p=1,\ldots,|\mathrm{NSrc}|$), $\deltaA$, $A_*$ whose legs can take value $0_{\sigma_q}$.
\end{lem}

\begin{proof}
    We need to check that legs of $\deltaB_{w_p}^{\sigma_p}$ ($p>q$) cannot take value $0_{\sigma_q}$ (for other tensors it is evident). By \eqref{ErrTrmDef}, it is enough to show that for any index sequence $L \in \bar{\mathcal{L}_p}$, $L=l_1,\ldots,l_6$ (see \eqref{Ldef}) such that there is an $l_i=0_{\sigma_q}$, we have $(W_p)_{L}=0$.

    Note that \eqref{Ldef} implies that $l_i=0_{\sigma_q}$ corresponds to an internal leg $n_i^{w_p} \notin S_p$. Recall that $|S_p|\geq 2$, so there are two possible situations (as there are only three internal legs): $|S_p|=2$ or $|S_p|=3$. If $|S_p|=3$, all internal legs $n_i^{w_q} \in S_p$, making it impossible for $L$ to contain $l_i=0_{\sigma_q}$. Hence, legs of $\deltaB_{w_p}^{\sigma_p}$ cannot take the index value $0_{\sigma_q}$ in this case.

    Suppose $|S_p|=2$. Then,
    \begin{equation}\label{twozeros}
        \text{we have two internal legs in $L$ with index value $0$ (see \eqref{Ldef}) and one with $0_{\sigma_q}$.}
    \end{equation}
    Let us show that $(W_p)_{L}=0$. By \eqref{Wqdef}, we have:
    \begin{equation}\label{WpCHECK}
        (W_p)_{L} = \left(\deltaB_{w_p}^1\right)_{L} + \sum_{r=1}^{p-1} \left(\deltaB_{w_p}^{\sigma_r}\right)_{L}.
    \end{equation}
    The first term in \eqref{WpCHECK} is zero as $\deltaB_{w_p}^1$ has no legs which can take index value $0_{\sigma_q}$. For terms $(\deltaB_{w_p}^{\sigma_r})_{L}$, we can have $w_p \notin \graph(\sigma_r)$ and $w_p \in \graph(\sigma_r)$. The former case is trivial as by \eqref{ZerosDef} we have $\deltaB_{w_p}^{\sigma_r}=0$. The latter case divides into two subcases: $w_p \in\graph (\widehat{\sigma_r})$ or $w_p =w_r$.

    If $w_p \in\graph (\widehat{\sigma_r})$, then $(\deltaB_{w_p}^{\sigma_r})_{L}=0$ as at least one internal leg of $\deltaB_{w_p}^{\sigma_r}$ is restricted to $\{1_{\sigma_r}\}$ (see \eqref{vA},\eqref{vBdef}) and so cannot take index values $0$ or $0_{\sigma_q}$.

    If $w_p=w_r$, then $b_{w_p}^{\sigma_r}=b_{w_r}^{\sigma_r}$. By \eqref{ErrTrmDef} and \eqref{Lprimeseq} with $q=r$, at least two internal legs of $b_{w_r}^{\sigma_r}$ are restricted to $\{0_{\sigma_r}\}$. Note that there are only three internal legs. Then, using \eqref{twozeros}, we may conclude that there is at least one internal leg restricted to $\{0_{\sigma_r}\}$ with the corresponding index value in $L$ equals $0$.\footnote{This follows from the fact that $2+2>3$.} Therefore, $(\deltaB_{w_r}^{\sigma_r})_{L}=0$.

    Thus, all terms in the r.h.s. of \eqref{WpCHECK} are zero, implying $(W_p)_{L}=0$. Therefore, legs of $\deltaB_{w_p}^{\sigma_p}$ cannot take the value $0_{\sigma_q}$.
\end{proof}

The following lemma establishes another key property of tensors $\deltaB_v^{\sigma_q}$. We will use it to verify \eqref{compat5}.
\begin{lem}\label{lem:nowqnew}
    Let $\gamma$ be a nontrivial template (connected or disconnected). We denote by $v_1,\ldots,v_m$ the vertices of $\graph (\gamma)$. Consider a contraction $T_\gamma (\deltaB_{v_1}^{K_1},\ldots, \deltaB_{v_m}^{K_m})$,\footnote{The notation $T_\gamma (\deltaB_{v_1}^{K_1},\ldots, \deltaB_{v_m}^{K_m})$ was introduced in \eqref{dep_on_r}. Recall that the tensor labelled with $v_i$ is assumed to be inserted in the vertex $v_i$. So, the order in which tensors are listed in the argument of $T_\gamma$ is irrelevant.} where each $K_i$ is either $1$ or a template belonging to $\mathrm{NSrc}$.\footnote{Note the difference between $K_i=1$ and $K_i \in \mathrm{NSrc}$. The tensor $\deltaB_{v_i}^1$ ($K_i=1$) is the tensor obtained in Section~\ref{dbv1} as the sum over $\mathrm{Src}$. At the same time, the tensor $\deltaB_{v_i}^{\sigma_q}$ ($K_i=\sigma_q\in\mathrm{NSrc}$) is only a summand in the definition of $\deltaB_{v_i}^2$ (see \eqref{dbv2Exp}).} If for some $i,j \in \{1,\ldots,m\}$ with $i\neq j$, and a $q \in \{1,\ldots,|\mathrm{NSrc}|\}$ we have:
    \begin{itemize}
        \item $v_i=w_q$, $K_i=\sigma_q$ (i.e., the special tensor $\deltaB_{w_q}^{\sigma_q}$ is inserted in $v_i=w_q$);
        \item $v_i, v_j$ belong to a connected component $\graph (\mu)$ of $\graph(\gamma)$;
        \item $K_j=\sigma_q$,
    \end{itemize}
    then
    \begin{equation}\label{nowqnew}
        T_\gamma (\deltaB_{v_1}^{K_1},\ldots, \deltaB_{v_m}^{K_m})=0.
    \end{equation}
\end{lem}

\begin{proof}[Proof of Lemma~\ref{lem:nowqnew}]
    Without loss of generality, assume $v_j=v_1, v_i=v_2$. Consider a path within $\graph(\mu)$ from $v_1$  to $v_2=w_q$ without backtracking and self-intersections. [There is such a path as $\graph(\mu)$ is connected.] Let $v_{i_1},\ldots, v_{i_l}$ be the vertices lying along this path ($i_1=1, i_{l}=2$). There are two cases to consider: (a) there is a $k \in \{2,\ldots,l-1\}$ such that $v_{i_k}=w_p, K_{i_k}=\sigma_p$ ($p\in \{1,\ldots,|\mathrm{NSrc}|\}$); (b) there is no such $k$ as in (a).

    \begin{itemize}
        \item[(a)] Note that vertices $v_{i_2},\ldots, v_{i_{l-1}}$ have degrees of at least two in $\graph (\mu)$. Then, for each $v_{i_k}$, $k=2,\ldots,l-1$, there are at least two bonds touching $v_{i_k}$, which are restricted to $\mathcal{D}$. Next, recall that $\deltaB_{w_p}^{\sigma_p}$ has $|S_p|\geq 2$ internal legs restricted to $\{0_{\sigma_p}\}$ (see \eqref{ErrTrmDef}, \eqref{Lprimeseq}). Then at least one out of three internal legs of $\deltaB_{w_p}^{\sigma_p}$ is restricted in $T_\gamma(\deltaB_{v_1}^{K_1},\ldots, \deltaB_{v_m}^{K_m})$ to $\mathcal{D} \cap \{0_{\sigma_p}\}=\void$. Therefore, \eqref{nowqnew} holds.

        \item[(b)] We divide this case into two subcases:(b.1) (a) does not hold and there is a $k \in\{2, \ldots,l-1\}$ such that $K_{i_k} \neq \sigma_q$; (b.2) (a) does not hold and for all $k \in \{1,\ldots,l\}$ we have $K_{i_k}=\sigma_q$.
            \begin{itemize}
                \item[(b.1)] Recall that internal ticked (restricted to $\mathcal{D}$) legs of $\deltaB_v^{\sigma_q}$ ($v\neq w_q$) are restricted to $\{1_{\sigma_q}\}$ (see \eqref{vA},\eqref{vBdef}). Consider the smallest $k$ for which $K_{i_k} \neq \sigma_q$. Then, in $v_{i_{k-1}}$ we have an insertion of $\deltaB_{v_{i_{k-1}}}^{\sigma_q}$. Then $\deltaB_{v_{i_{k-1}}}^{\sigma_q}$ is contracted along a ticked bond with $\deltaB_{v_{i_{k}}}^{1}$ or $\deltaB_{v_{i_{k}}}^{\sigma_p}$ ($p\neq q, v_{i_k}\neq w_p$ as (a) does not hold) whose internal ticked legs are restricted to the sets disjoint from $\{1_{\sigma_q}\}$ ($\bigsqcup\limits_{\mu \in \mathrm{Src}}\N_\mu$ and $\{1_{\sigma_p}\}$, respectively). Consequently, \eqref{nowqnew} holds.
                \item[(b.2)] In this subcase, $\deltaB_{v_{l-1}}^{\sigma_q}$ is contracted with $\deltaB_{w_q}^{\sigma_q}$ along a ticked bond. Consequently, \eqref{nowqnew} holds, as internal ticked legs of $\deltaB_{v_{l-1}}^{\sigma_q}$ and $\deltaB_{w_q}^{\sigma_q}$ are restricted to disjoint sets $\{1_{\sigma_q}\}$ (see \eqref{vA},\eqref{vBdef}) and $\bigsqcup\limits_{\mu \in \mathrm{NSrc}}\N_\mu \sqcup \{1_{\sigma_1},\ldots, 1_{\sigma_{q-1}}\} $ (see \eqref{ErrTrmDef},\eqref{Wqdef}), respectively.
            \end{itemize}
    \end{itemize}
\end{proof}

There is an important corollary from Lemmas~\ref{lem:nocrossvals},\ref{lem:nowqnew} which we will use for the verification of \eqref{compat3}, \eqref{compat4}.
\begin{cor}\label{cor:1}
    Let $\gamma$ be a connected template and $v_i, K_i$ be as in Lemma~\ref{lem:nowqnew}. If for some $i \in \{1,\ldots, m\}$ there is $q \in \{1, \ldots, \|\mathrm{NSrc}\|\}$ such that $v_i=w_q$ and $K_i=\sigma_q$, then Eq.~\eqref{nowqnew} holds.
\end{cor}
In other words, if $\deltaB_{w_q}^{\sigma_q}$, $q =1,\ldots,|\mathrm{NSrc}|$, (i.e., a special tensor) is inserted in a contraction $T_\gamma$, where $\gamma$ is connected, the contraction is zero.

\begin{proof}[Proof of Corollary~\ref{cor:1}]
    There are two cases to consider: (a) there is a $v \in V_q$ satisfying either $v=v_j, K_j \neq \sigma_q$ or $v \notin \graph(\gamma)$; (b) (a) does not hold, i.e., all $v \in V_q$ belong to $\graph(\gamma)$ and the corresponding $K_j$'s equal $\sigma_q$.
    \begin{itemize}
        \item[(a)] In this case, we have (if $v=v_j, K_j \neq \sigma_q$) $\deltaB_{v}^1$, $\deltaB_{v}^{\sigma_p}$, $p\neq q$,  or (if $v \notin \graph(\gamma)$) $A_*$ inserted in $v \in V_q$ in the contraction. By Lemma~\ref{lem:nocrossvals}, none of these tensors has legs which can take value $0_{\sigma_q}$. Then, as legs $n^{w_q}_k \in S_q$ of $\deltaB_{w_q}^{\sigma_q}$ (the legs pointing towards vertices in $V_q$) are restricted to $\{0_{\sigma_q}\}$ we have \eqref{nowqnew}.

        \item[(b)] In this case Lemma~\ref{lem:nowqnew} implies Eq.~\eqref{nowqnew}. [Note that there is only one connected component in $\graph (\gamma)$ as $\gamma \in \mathrm{Con}$.]
    \end{itemize}
\end{proof}

\subsubsection{Verification of \texorpdfstring{\eqref{Pconds}}{(\getrefnumber{Pconds})}}

It is evident that \eqref{P0} and \eqref{P3} are satisfied. Thus, we focus on \eqref{P1} and \eqref{P2}. We will show that \eqref{P1} and \eqref{P2} hold for any $a,b$ within the intervals:
\begin{equation}\label{final_intervals}
    1/2<a \leq \min (1-7c, 8/7)=1-7c, \qquad 1<b \leq \min (1+c,8/7)=1+c,
\end{equation}
where $c<1/14$ is the constant fixed at the beginning of Section~\ref{gendbv1}. Note that intervals \eqref{final_intervals} coincide with the intervals \eqref{gen_intervals} for $\deltaB_v^1$ tensors in Section~\ref{PcondsDBV1}. The lower bounds in \eqref{final_intervals} are the conditions on $a,b$ in Lemma~\ref{mainlem}. The condition $c<1/14$ ensures that the intervals are nonempty.

For $v \notin\graph (\sigma_q)$, tensors $\deltaB_v^{\sigma_q}=0$ trivially satisfy \eqref{P1} and \eqref{P2}. So, we will focus on $\deltaB_v^{\sigma_q}$ with $v \in\graph (\sigma_q)$. We need to consider two cases: $v\neq w_q$ and $v=w_q$.

If $v \neq w_q$, \eqref{vA} and \eqref{vBdef} yield $\deltaB_{v}^{\sigma_q}=O(\epsilon^{N_q/(N_q-1)})$. Thus, as $1<N_q\leq 8$, \eqref{P1} and \eqref{P2} hold with $a,b \leq 8/7$ in agreement with \eqref{final_intervals}.

If $v=w_q$, we prove the statement by induction. Consider the base case $q=1$. By \eqref{ErrTrmDef} and \eqref{Wqdef}, tensor elements of $\deltaB_{w_1}^{\sigma_1}$ and $\bar{\deltaB_{w_1}^{\sigma_1}}$ are given, up to a minus sign, by those of $\deltaB_{w_1}^1$ and $\bar{\deltaB_{w_1}^1}$, respectively. Then, since $\deltaB_{w_1}^1$ satisfies \eqref{P1}, \eqref{P2} with $a,b$ as in \eqref{final_intervals}, the same applies to $\deltaB_{w_1}^{\sigma_1}$.

Now, suppose the statement is proven for $q=1,\ldots,p$. Then, $W_{p+1}$ satisfies \eqref{P1}, \eqref{P2} with $a,b$ as in \eqref{final_intervals} because all terms in \eqref{Wqdef} satisfy these conditions. Therefore, by \eqref{ErrTrmDef},  $\deltaB_{w_{p+1}}^{\sigma_{p+1}}$ satisfies \eqref{P1}, \eqref{P2} with the same $a,b$.

Thus, by induction, $\deltaB_{w_q}^{\sigma_q}$ for $q=1,\ldots, |\mathrm{NSrc}|$ satisfy \eqref{P1}, \eqref{P2} with $a,b$ as in \eqref{final_intervals}.

\subsubsection{Verification of \texorpdfstring{\eqref{compat3}}{(\getrefnumber{compat3})}}

If $\gamma = \widehat{\sigma_q}$, we have $T_{\widehat{\sigma_q}} (\{\deltaB_v^{\sigma_q}\}) \propto T_*$ with the proportionality constant equals the product of tensor elements \eqref{vA1} and \eqref{vBdef} of tensors $\deltaB_v^{\sigma_q}$ ($v \in\graph (\widehat{\sigma_q})$). [For $v\in \graph(\widehat{\sigma_q}), v \notin V_q$, tensor $\deltaB_v^{\sigma_q}$ has only one nonzero tensor element \eqref{vBdef}. For $v \in V_q \subset \graph (\widehat{\sigma_q})$, the only nonzero tensor element of $\deltaB_v^{\sigma_q}$ appearing in the contraction is \eqref{vA1} as we have $A_*$ inserted in $w_q$.]\footnote{See definition of $V_q$ in step~\ref{itm:notations} of the general construction.} There are $N_q-2$ tensors with tensor elements $\epsilon^{(N_q-1)/N_q}$ and one with $C_{\sigma_q} \epsilon^{(N_q-1)/N_q-N_q}$.\footnote{See the definition of $N_q$ in step~\ref{itm:step0dbv2} of the general construction.} Putting all factors together, we obtain the first line of \eqref{compat3}:
\begin{equation}\label{compat3FL}
    T_{\widehat{\sigma_q}}(\{\deltaB_v^{\sigma_q}\})=C_{\sigma_q} \epsilon^{N_q/(N_q-1)-N_q} \left(\epsilon^{N_q/(N_q-1)}\right)^{N_q-2} T_* =C_{\sigma_q} T_* =[\text{by \eqref{CsDEF}}]\  T_{\sigma_q}(\deltaA).
\end{equation}

Assume next that $\gamma \in \mathrm{Con}$ and $\gamma \neq \widehat{\sigma_q}$. We need to show that
\begin{equation}\label{wanttoshow}
    T_\gamma(\{\deltaB_{v}^{\sigma_q}\})=0
\end{equation}
There are two cases to consider: $w_q \in \graph (\gamma)$ and $w_q \notin \graph (\gamma)$.

If $w_q \in \graph (\gamma)$, Corollary~\ref{cor:1} implies Eq.~\eqref{wanttoshow}.

If $w_q \notin \graph(\gamma)$, by \eqref{vA}, \eqref{vBdef}, all tensors $\deltaB_v^{\sigma_q}$ with $v \in\graph (\gamma)$ are restricted according to $\widehat{\sigma_q}$. Then, by Lemma~\ref{lem:incompatible_contractions} (with $m=1$, $\gamma_k=\gamma$, $\mu = \widehat{\sigma_q}$), we have \eqref{wanttoshow}.

\subsubsection{Verification of \texorpdfstring{\eqref{compat4}}{(\getrefnumber{compat4})}}

Let $\gamma \in \mathrm{Con}$. We denote by $v_1,\ldots, v_m$ the vertices of $\graph (\gamma)$. Then, we  substitute \eqref{dbv2Exp} into the l.h.s. of \eqref{compat4} and expand the result by multilinearity of contractions:\footnote{Here we use the notation $T_\gamma(\{\deltaB_v^1+\deltaB_v^2\})=T_\gamma(\deltaB_{v_1}^1+\deltaB_{v_1}^2,\ldots,\deltaB_{v_m}^1+\deltaB_{v_m}^2)$ introduced earlier in \eqref{dep_on_r}.}
\begin{equation}\label{compat4EXP}
    \begin{aligned}
        T_\gamma (\deltaB_{v_1}^1+\deltaB_{v_1}^2, \ldots, \deltaB_{v_m}^1+\deltaB_{v_m}^2) & =                                                                                                                                                              \\
                                                                                            & =T_\gamma (\deltaB_{v_1}^1+\sum_{q=1}^{|\mathrm{NSrc}|}\deltaB_{v_1}^{\sigma_q}, \ldots, \deltaB_{v_m}^1+\sum_{q=1}^{|\mathrm{NSrc}|}\deltaB_{v_m}^{\sigma_q}) \\
                                                                                            & =\sum_{K_i \in \{1\}\sqcup \mathrm{NSrc}}  T_\gamma (\deltaB_{v_1}^{K_1},\ldots, \deltaB_{v_m}^{K_m}).
    \end{aligned}
\end{equation}
By Corollary~\ref{cor:1}, terms in \eqref{compat4EXP} with insertions of special tensors $\deltaB_{w_q}^{\sigma_q}$ vanish. Hence, only tensors $\deltaB_v^1$ and $\deltaB_v^{\sigma_q}$ ($v \neq w_q$) appear in the nonzero terms of expansion \eqref{compat4EXP}. Then the legs of each $\deltaB_{v_i}^{K_i}$ contracted along the edges of $\graph (\gamma)$ belong to the sets which are disjoint for different $K_i$'s: $\{1_{\sigma_q}\}$ for $K_i=\sigma_q$ and $\bigsqcup\limits_{\mu \in \mathrm{Src}} \N_\mu$ for $K_i=1$. Consequently, $T_\gamma (\deltaB_{v_1}^{K_1},\ldots, \deltaB_{v_m}^{K_m})$ is nonzero only for $K_1=\ldots=K_m$ in which case it equals either $T_\gamma(\{\deltaB_v^1\})$ ($K_i=1$) or $T_\gamma(\{\deltaB_v^{\sigma_q}\})$ ($K_i=\sigma_q$). This implies \eqref{compat4}.

\subsubsection{Verification of \texorpdfstring{\eqref{compat5}}{(\getrefnumber{compat5})}}

Now we will prove \eqref{compat5}, concluding the construction of $\deltaB_v^2$ tensors and the proof of Lemma~\ref{mainlem}.

We aim to show that for every $M=1,\ldots,|\mathrm{NSrc}|$, there exists a collection of disjoint sets $D_\gamma^{M}$, $\gamma \in \mathrm{Disc}$, (some of which can be empty) such that:\footnote{See the definition of $\mathrm{Disc}_M$ in \eqref{DiscMdef}.}
\begin{equation}\label{Dsets}
    \mathrm{Disc}_M=\bigsqcup_{\gamma \in \mathrm{Disc}} D_\gamma^M
\end{equation}
and
\begin{equation}\label{Dsums}
    T_\gamma \left(\left\{ \deltaB_v^1+\sum_{q=1}^{M} \deltaB_v^{\sigma_q}\right\} \right)= \sum_{\lambda \in D^M_\gamma} T_\lambda(\deltaA) \qquad (\gamma \in \mathrm{Disc}).
\end{equation}
Clearly, \eqref{compat5} follows from \eqref{Dsums}, \eqref{Dsets}. This strategy generalises \eqref{desire2}, \eqref{desire} from the $q=1$ example.

We will start by defining the sets $D_\gamma^M$ and discussing the relations between $D_\gamma^{M}$ and $D_\gamma^{M-1}$. After this, we will prove \eqref{Dsums} by induction in $M$.

\paragraph{Sets $D_\gamma^M$} Note that each $\lambda \in \mathrm{Disc}_M$ can be uniquely expressed as:
\begin{equation}\label{lambdaexp}
    \lambda = \nu_1 \star \ldots \star \nu_{m_1} \star \mu_1 \star \ldots \star \mu_{m_2},
\end{equation}
where $\nu_i \in \mathrm{Src}$, $\mu_i \in \{\sigma_1,\ldots, \sigma_M\}$, $m_1,m_2 \geq 0$, and $m_1+m_2\geq 2$. We define $\widehat{\lambda}$ as the template obtained from $\lambda$ by replacing all $\mu_i$'s with $\widehat{\mu_i}$'s:
\begin{equation}
    \widehat{\lambda}=\nu_1 \star \ldots \star \nu_{m_1} \star \widehat{\mu_1} \star \ldots \star \widehat{\mu_{m_2}}.
\end{equation}
Note that if there are no $\mu_i$'s in \eqref{lambdaexp} ($m_2=0$), we have $\lambda \in \mathrm{DiscS}$ and $\widehat{\lambda}=\lambda$.

We define $D_\gamma^M \subset{\mathrm{Disc}_M}$ for each $\gamma \in \mathrm{Disc}$ as the set of all disconnected templates $\lambda\in \mathrm{Disc}_M$ satisfying $\widehat{\lambda} = \gamma$:
\begin{equation}\label{DgammaMdef}
    D_\gamma^M=\{\lambda| \lambda \in \mathrm{Disc}_M, \ \widehat{\lambda}=\gamma\}.
\end{equation}
Clearly, such sets $D_\gamma^M$ are disjoint (for the same $M$ and different $\gamma$'s) and \eqref{Dsets} is satisfied.

Also, we define $D_\gamma^M \subset \{\sigma_1,\ldots, \sigma_M\}$ for each $\gamma\in \mathrm{Con}$ as follows:
\begin{equation}\label{DgammaMdef2}
    D_\gamma^M = \{\lambda| \lambda \in \{\sigma_1,\ldots,\sigma_M\}, \ \widehat{\lambda}=\gamma\}.
\end{equation}
Although the sets $D_\gamma^M$ with $\gamma \in \mathrm{Con}$ are not present in \eqref{Dsets} and \eqref{Dsums}, they will be involved in the analysis of templates of the form $\widehat{\sigma_M} \star \nu$, where $\nu$ is connected.

To simplify the upcoming proof of \eqref{Dsums}, we will now express $D_\gamma^M$ with $\gamma \in \mathrm{Disc}$ in terms of sets $D^{M-1}_\nu$, where $\nu$ can be either connected or disconnected. We need to examine three possible relations between generic $\gamma$ and $\widehat{\sigma_M}$:
\begin{subequations}\label{gen_gammasigma}
    \begin{align}
         & \text{$\widehat{\sigma_M}$ is not a connected component of $\gamma$;}\label{gen_gammasigmaA}                             \\
         & \text{$\widehat{\sigma_M}$ is a connected component of $\gamma$ and $w_M \in\graph (\gamma)$;} \label{gen_gammasigmaB}   \\
         & \text{$\widehat{\sigma_M}$ is a connected component of $\gamma$ and $w_M \notin\graph (\gamma)$;}\label{gen_gammasigmaC}
    \end{align}
\end{subequations}
Note that \eqref{gen_gammasigmaA}-\eqref{gen_gammasigmaC} are mutually exclusive and that each $\gamma \in \mathrm{Disc}$ satisfies one of these relations.
\begin{itemize}
    \item[\eqref{gen_gammasigmaA}] Clearly, $\lambda \in \mathrm{Disc}_M$ may satisfy $\widehat{\lambda} = \gamma$ only if $\sigma_M$ is not a connected component of $\lambda$. [Otherwise $\widehat{\sigma_M}$ is a connected component of $\gamma$, which contradicts \eqref{gen_gammasigmaA}.] In other words, $\lambda$ should belong to $\mathrm{Disc}_{M-1} \subset \mathrm{Disc}_M$. Thus, by \eqref{DgammaMdef}, we have:
        \begin{equation}\label{MMm11}
            D^M_\gamma = D^{M-1}_\gamma.
        \end{equation}

    \item[\eqref{gen_gammasigmaB}] Again, $\lambda \in \mathrm{Disc}_M$ may satisfy $\widehat{\lambda}=\gamma$ only if $\sigma_M$ is not a connected component of $\lambda$. Otherwise, vertex $w_M$ belongs to the single connected component $\graph (\sigma_M)$ of $\graph(\lambda)$. In $\widehat{\lambda}$, $w_M$ were removed from $\graph(\sigma_M)$ to get $\graph(\widehat{\sigma_M})\subset \graph(\widehat{\lambda})$, so $w_M \notin \graph(\widehat{\lambda})= \graph(\gamma)$, which contradicts \eqref{gen_gammasigmaB}. Thus, we have \eqref{MMm11} again.

    \item[\eqref{gen_gammasigmaC}] As $\widehat{\sigma}_M$ is a connected component of $\gamma$, we have:
        \begin{equation}\label{gammadecnu}
            \gamma = \widehat{\sigma_M} \star \tilde{\gamma},
        \end{equation}
        where $\tilde{\gamma}$ can be either connected or disconnected. Consequently, there are two types of $\lambda \in \mathrm{Disc}_M$ satisfying $\widehat{\lambda}=\gamma$: $\lambda \in D_\gamma^{M-1}$ and $\lambda = \sigma_M \star \mu$, where $\mu \in D^{M-1}_{\tilde{\gamma}}$.

        Let us focus on the second case. For the union $\sigma_M \star \mu$ to be defined, we need $\graph (\sigma_M) \cap\graph (\mu)=\void$.\footnote{Recall that $\graph(\sigma_M)\cap \graph(\mu)=\void$ is required by the definition of union (see Def.~\ref{def:union}).} Let us show that this property holds for all $\mu \in D^{M-1}_{\tilde{\gamma}}$.

        Assume the contrary, i.e., that there is a $\mu \in D^{M-1}_{\tilde{\gamma}}$ such that
        \begin{equation}\label{inter1}
            \graph (\sigma_M) \cap\graph (\mu)\neq\void.
        \end{equation}
        Note that, as $\widehat{\sigma_M} \star \widehat{\mu}$ is defined (it equals $\gamma$), we have
        \begin{equation}\label{inter2}
            \graph (\widehat{\sigma_M}) \cap\graph (\widehat{\mu})=\void.
        \end{equation}
        Since conditions \eqref{inter1} and \eqref{inter2} are compatible, one of the two possibilities must hold: (a) $w_M \in\graph (\mu)$; (b) there is a connected component $\sigma_q \in \mathrm{NSrc}$ of $\mu$ such that $w_q \in\graph (\sigma_M)$.
        \begin{itemize}
            \item[(a)] Suppose $w_M \in\graph (\mu)$. Note that, by \eqref{gen_gammasigmaC}, $w_M \notin\graph (\widehat{\mu}) \subset\graph (\gamma)$. Then, there is a connected component $\sigma_q$ of $\mu$ such that $w_M=w_q \in\graph (\sigma_q)$. Then $w_q=w_M \in g(\sigma_M)$ and this case is reduced to the case (b).

            \item[(b)] Let $nbrs(w_q)$ be the set of neighbours of $w_q$ in $\graph (\sigma_M)$. Note that, as $\sigma_M \in \mathrm{NSrc}$, degree of $w_q$ in $\graph (\sigma_M)$ is $\geq 2$, and so $|nbrs(w_q)| \geq 2$. Then, as $|V_q| \geq 2$, and there are only three vertices in the cube that can belong to $V_q \subset \graph(\widehat{\sigma_q})$ or $nbrs(w_q) \subset \graph (\sigma_M)$, we can find a $v \in nbrs(w_q) \cap V_q \neq \void$. As $v \in nbrs(w_q) \subset \graph(\sigma_M)$, there are two options: either $v \in \graph(\widehat{\sigma_M})$ or $v=w_M$.

                If $v \in\graph (\widehat{\sigma_M})$, this contradicts \eqref{inter2} as $v \in V_q \subset \graph(\widehat{\sigma_q}) \subset \graph(\widehat{\mu})$. If $v = w_M$, this contradicts \eqref{gen_gammasigmaC} as $v \in V_q \subset \graph (\widehat{\mu}) \subset \graph(\gamma)$ (recall $\gamma = \widehat{\sigma_M} \star \widehat{\mu}$).
        \end{itemize}
        Thus, all unions $\sigma_M \star \mu$, $\mu \in D_{\tilde{\gamma}}^{M-1}$ are defined. Therefore, we have:
        \begin{equation}\label{MMm12}
            D^{M}_{\gamma} = D^{M-1}_\gamma \sqcup \left(\sigma_M \star D^{M-1}_{\tilde{\gamma}}\right),
        \end{equation}
        where $\sigma_M \star D^{M-1}_{\tilde{\gamma}}$ denotes the set of all unions of $\sigma_M$ with the elements of $D^{M-1}_{\tilde{\gamma}}$.
\end{itemize}

\paragraph{Proof of \eqref{Dsums}} We are going to use induction in $M$.

It is convenient to consider $M=0$ as the base case. For this, we define $\mathrm{Disc}_0=\mathrm{DiscS}$, $D^0_\gamma = \{\gamma\}$ for $\gamma \in \mathrm{DiscS}$, and $D^0_\gamma=\void$ for $\gamma \in \mathrm{Disc} \setminus \mathrm{DiscS}$. Then, \eqref{Dsums} for $M=0$ follows from \eqref{dbv1_goal}.

Now, assume \eqref{Dsums} with $M$ replaced by an $M-1 \in \{0, \ldots, |\mathrm{NSrc}|-1\}$ is proven. We will show now that this implies \eqref{Dsums}. To simplify the notation, we introduce tensors $\deltaB_v^{(M-1)}$ as follows:
\begin{equation}\label{dbD}
    \deltaB_v^{(M-1)}=\deltaB_v^1+\sum\limits_{q=1}^{M-1} \deltaB_v^{\sigma_q}.
\end{equation}
Thus, we want to prove that
\begin{equation}\label{compat5stp1}
    T_\gamma \left(\left\{ \deltaB_v^{(M-1)} +\deltaB_v^{\sigma_{M}} \right\}\right)=\sum_{\lambda \in D_\gamma^{M}} T_\lambda(\deltaA),
\end{equation}
having that
\begin{equation}\label{compat5stp2}
    T_\gamma \left( \left\{\deltaB_v^{(M-1)}\right\}\right)=\sum_{\lambda \in D_\gamma^{M-1}} T_\lambda(\deltaA).
\end{equation}

We will consider relations \eqref{gen_gammasigmaA}-\eqref{gen_gammasigmaC} between $\gamma$ and $\widehat{\sigma_{M}}$ and prove \eqref{compat5stp1} in each case separately. In what follows, we denote by $v_1,\ldots,v_m$ the vertices of $\graph (\gamma)$.

\begin{itemize}
    \item[\eqref{gen_gammasigmaA}]
        We expand $T_\gamma \left( \left\{\deltaB_v^{(M-1)} +\deltaB_v^{\sigma_{M}} \right\}\right)$ by multilinearity of contractions and obtain:
        \begin{equation}\label{final_exp1}
            T_\gamma \left( \left\{\deltaB_v^{(M-1)} +\deltaB_v^{\sigma_{M}}\right\}\right)=\sum_{K_i \in \{(M-1), \sigma_{M}\}} T_\gamma \left( \deltaB_{v_1}^{K_1},\ldots, \deltaB_{v_m}^{K_m} \right).
        \end{equation}
        We will now show that each term on the r.h.s. of \eqref{final_exp1} such that there is a $K_i = \sigma_{M}$ vanishes. This implies:
        \begin{equation}\label{dropD}
            T_\gamma \left( \left\{\deltaB_v^{(M-1)} +\deltaB_v^{\sigma_{M}}\right\}\right)=T_\gamma \left(\left\{ \deltaB_v^{(M-1)}\right\}\right),
        \end{equation}
        which, by \eqref{compat5stp2} and \eqref{MMm11}, yields \eqref{compat5stp1}.

        Consider a term $T_\gamma \left( \deltaB_{v_1}^{K_1},\ldots, \deltaB_{v_m}^{K_m} \right)$ from the r.h.s. of \eqref{final_exp1} such that there is a $K_i = \sigma_{M}$. We want to show that
        \begin{equation}\label{itiszero}
            T_\gamma \left( \deltaB_{v_1}^{K_1},\ldots, \deltaB_{v_m}^{K_m} \right)=0.
        \end{equation}
        Without loss of generality, we assume $i=1$. Note that there are two possible situations: (a) $v_1=w_{M}$ and (b) $v_1 \neq w_{M}$. Suppose we are in the situation (a). This means that we have a special tensor $\deltaB_{w_M}^{\sigma_M}$ inserted in vertex $v_1=w_M$. Then Lemma~\ref{lem:nocrossvals} and Eq.~\eqref{speciallegszeros} imply that either \eqref{itiszero} holds (what we wanted to show), or that all $v \in V_{M}$ belong to $\graph(\gamma)$ and the corresponding $K_j$'s equal $\sigma_{M}$. The latter case reduces to the situation (b) by relabeling vertices $v_1,\ldots,v_m$ such that some $v_j \in V_M$ becomes $v_1$. Thus, it is left to consider the situation (b).

        Let $\nu$ be the connected component of $\gamma$ such that $v_1 \in \graph(\nu)$. We should, again, consider two situations: (b.1) the special tensor $\deltaB_{w_M}^{\sigma_M}$ is inserted in $\graph(\nu)$; (b.2) (b.1) does not hold.

        \begin{itemize}
            \item[(b.1)] Lemma~\ref{lem:nowqnew} implies \eqref{itiszero}.
            \item[(b.2)] We may assume that for all $v_j \in \graph(\nu)$ we have $K_j=\sigma_{M}$ as otherwise it is clear that \eqref{itiszero} holds. [Contraction of nonspecial $\deltaB_{v}^{\sigma_M}$ ($v \neq w_M$) and $\deltaB_{v}^{(M-1)}$ along a ticked line is zero as their ticked legs are restricted to disjoint sets $\{1_{\sigma_M}\}$ and $\bigsqcup_{\mu \in \mathrm{NSrc}}\N_\mu \sqcup \{1_{\sigma_1}, \ldots,1_{\sigma_{M_1}}\} $, respectively.] Next, note that condition \eqref{gen_gammasigmaA} implies $\nu \neq \widehat{\sigma}_M$. Also, recall that nonspecial tensors $\deltaB_{v_j}^{\sigma_M}$ are restricted according to $\widehat{\sigma_M}$. Then, Lemma~\ref{lem:incompatible_contractions} (with $\gamma_k=\nu$ and $\mu=\widehat{\sigma_M}$) implies \eqref{itiszero}.
        \end{itemize}
        Thus, we have \eqref{dropD}, and so \eqref{compat5stp1} holds.

    \item[\eqref{gen_gammasigmaB}] In this case, it is clear that the expansion \eqref{final_exp1} has only three nonzero terms:
        \begin{equation}\label{dropDD}
            T_\gamma \left(\left\{ \deltaB_v^{(M-1)} +\deltaB_v^{\sigma_{M}}\right\}\right)=t_1+t_2+t_3.
        \end{equation}
        The first term $t_1$ is the contraction without insertions of $\deltaB_v^{\sigma_{M}}$ tensors, i.e., $t_1=T_\gamma(\{\deltaB_v^{(M-1)}\})$. The second term $t_2$ is the contraction with tensors $\deltaB_v^{\sigma_{M}}$ in $v \in\graph (\widehat{\sigma_{M}})$ and $\deltaB^{(M-1)}_{v}$ in all other vertices, including $w_M$. The third term $t_3$ is the contraction with tensors $\deltaB_v^{\sigma_{M}}$ in $v \in\graph (\sigma_{M})$ and $\deltaB^{(M-1)}_{v}$ in all other vertices. Note that the only difference between $t_2$ and $t_3$ is that $t_2$ inserts $\deltaB_{w_M}^{(M-1)}$ in $w_M$ and $t_3$ inserts $\deltaB_{w_M}^{\sigma_M}$ in the same vertex. All the other terms in the expansion \eqref{dropDD} vanish because the contraction of nonspecial $\deltaB_v^{\sigma_M}$ with $\deltaB_v^{M-1}$ along a ticked line is zero as their ticked legs are restricted to disjoint sets $\{1_{\sigma_M}\}$ and $\bigsqcup_{\mu \in \mathrm{NSrc}}\N_\mu \sqcup \{1_{\sigma_1}, \ldots,1_{\sigma_{M_1}}\} $, respectively.

        Note that $\deltaB^{(M-1)}_{w_{M}}=W_{M}$ (see \eqref{Wqdef} and \eqref{dbD}). Next, note that  tensors $\deltaB_v^{\sigma_M}$, $v \in V_M$, has legs $n^{v}_i \in S_M$ (the legs pointing towards $w_M$) restricted to $\{0, 0_{\sigma_M}\}$ (see \eqref{vA}). Therefore, the only tensor elements of $\deltaB^{(M-1)}_{w_{M}}$ appearing in the contraction $t_2$ are $(\deltaB^{(M-1)}_{w_{M}})_L=(W_M)_L$, $L \in \bar{\mathcal{L}_M}$ (see \eqref{Ldef}). [Note that Lemma~\ref{lem:nocrossvals} implies that legs of $\deltaB^{(M-1)}_{w_M}$ cannot take value $0_{\sigma_M}$.] The same statement with $L$ replaced by $L'$ (see \eqref{Lprimeseq}) is true for the special tensor $\deltaB_{w_M}^{\sigma_M}$ in $t_3$.

        We note that as all tensor elements of $\deltaB_v^{\sigma_M}$, $v \in V_M$ (the tensors which are directly contracted with $(\deltaB_{w_M}^{M-1})_L=(W_M)_L$ and $(\deltaB_{w_M}^{\sigma_M})_{L'}$ in $t_2$ and $t_3$, respectively), appearing in the contractions $t_2$, $t_3$, are equal to each other (see \eqref{vA1} and \eqref{vA2}), \eqref{ErrTrmDef} implies that $t_2+t_3=0$. Thus, we have \eqref{dropD}, which, by \eqref{MMm11} and \eqref{compat5stp2}, yields \eqref{compat5stp1}.

    \item[\eqref{gen_gammasigmaC}] We represent $\gamma$ as in \eqref{gammadecnu}. Then, the expansion \eqref{final_exp1} has two nonzero terms:
        \begin{equation}\label{dontDropD}
            T_\gamma \left( \left\{\deltaB_v^{(M-1)} +\deltaB_v^{\sigma_{M}} \right\}\right)=t_1+t_2.
        \end{equation}
        Terms $t_1, t_2$ in \eqref{dontDropD} are analogous to those in  \eqref{dropDD}. The first term $t_1$ is the contraction without tensors $\deltaB_v^{\sigma_{M}}$, i.e., $t_1=T_\gamma(\{\deltaB_v^{(M-1)}\})$. The second term $t_2$ is the contraction with $\deltaB_v^{\sigma_{M}}$ in $v \in\graph (\widehat{\sigma_{M}})$ and $\deltaB_v^{(M-1)}$ in other vertices ($v \in\graph (\tilde{\gamma}) \subset\graph (\gamma)$). There is no $t_3$ as, by \eqref{gen_gammasigmaC}, $w_M \notin \graph(\gamma)$.

        By \eqref{vA1}, \eqref{vBdef}, we have $t_2=C_{\sigma_{M}} T_{\tilde{\gamma}} (\{\deltaB_v^{(M-1)}\})$. [Note that, as $w_{M}\notin \graph(\gamma)$, by \eqref{trm_dscr}, we have $A_*$ inserted in $w_{M}$, which restricts the legs $n^v_i \in S_M$ of tensors inserted in $v \in V_M$ to $\{0\}$, excluding tensor elements \eqref{vA2} from the contraction.] Then, by \eqref{compat5stp2}, we have:
        \begin{equation}\label{prelast}
            T_\gamma \left( \deltaB_v^{(M-1)} +\deltaB_v^{\sigma_{M}}\right)=\sum_{\lambda \in D_\gamma^{M-1}} T_\lambda(\deltaA)+C_{\sigma_{M}} \sum_{\mu \in D_{\tilde{\gamma}}^{M-1}} T_\mu(\deltaA).
        \end{equation}
        As proved before Eq.~\eqref{MMm12}, $\sigma_{M} \star \mu$ is defined for any $\mu \in \mathrm{D}_{\tilde{\gamma}}^{M-1}$. Then, it follows from \eqref{CsDEF} that
        \begin{equation}\label{last}
            C_{\sigma_{M}} T_\mu (\deltaA)=T_{\sigma_{M} \star \mu}(\deltaA).
        \end{equation}
        Using \eqref{last} and \eqref{MMm12} in \eqref{prelast} we obtain \eqref{compat5stp1}. This finishes the proof of \eqref{Dsums} and thereby the verification of \eqref{compat5}.
\end{itemize}

This section concludes the verification that tensors $\deltaB_v=\deltaB_v^1+\deltaB_v^2$ solve \eqref{maineq2} and satisfy \eqref{Pconds}, thereby confirming that Lemma~\ref{mainlem} holds.

\subsection{The exponent in Theorem~\ref{maintheor}}\label{theexponent}

As we demonstrated in Section~\ref{rRG}, Lemma~\ref{mainlem} implies Theorem~\ref{maintheor} and provides the exponent $h$ in \eqref{contr}. Namely, we obtained $h=\min(2a,b)$ (see \eqref{dApbound}), where $a$ and $b$ are the constants from Lemma~\ref{mainlem}.

We constructed tensors $\deltaB_v$ which satisfy the conditions of Lemma~\ref{mainlem} with any $a,b$ within the intervals \eqref{final_intervals}, where $c<1/14$. Maximising $\min (2a,b)$ over the intervals \eqref{final_intervals} under the condition $c<1/14$, we obtain the exponent $h$ in \eqref{contr}:
\begin{equation}\label{exponent_of_step}
    h=\max_{\overset{\text{intervals \eqref{final_intervals},}}{c<1/14}} \min (2a,b) = 16/15.
\end{equation}
Value \eqref{exponent_of_step} is reached at:
\begin{equation}
    c=1/15, \qquad a=1-7c=8/15, \qquad b=1+c=16/15
\end{equation}

\section{Final remarks}\label{finalremarks}

In this paper, we gave a rigorous formulation of a tensor RG approach for the high-temperature phase of $3D$ lattice models. We explicitly constructed an RG map $R_\epsilon: A \mapsto A'$, acting in $U_\epsilon (A_*)$ (the ball of radius $\epsilon>0$ around the high-temperature fixed point tensor $A_*$). The map $R_\epsilon$ can be expressed as the following composition:
\begin{equation}
    R_\epsilon = R^\epsilon_{rearrangement} \circ R_{simple} \circ R_{gauge},
\end{equation}
where $R_{gauge}$ and $R_{simple}$ are the gauge transformation and simple RG step, respectively, discussed in Section~\ref{cornstr}, and $R^\epsilon_{rearrangement}$ is the rearrangement RG step discussed in Section~\ref{rRG}. The RG map $R_\epsilon$ has the lattice rescaling factor $s=4$.

We demonstrated that for a small enough $\epsilon$, $R_\epsilon$ is a contraction in the sense of Eq.~\eqref{contr} with $h=16/15$. The constant $t$ in Eq.~\eqref{contr} could be extracted by following the proofs of Theorem~\ref{maintheor} and Lemma~\ref{mainlem}; we will not do it here. We will show now how we can use $R_\epsilon$ to map any tensor in the vicinity of $A_*$ to a tensor arbitrarily close to $A_*$ while preserving the partition function.

Let $\epsilon_0>0$, $A_0 \in U_{\epsilon_0} (A_*)$, and $R_{\epsilon_0}(A_0)=\z_1 (A_*+\deltaA_1)$. In Remark~\ref{factorrem}, we mentioned that the normal factor $\z_1$ can always be factored out from the partition function. Because of this, we are interested in the RG map followed by the normalisation of the tensor, which makes the normal factor equal to one. Let $\tilde{R}_{\epsilon_0}$ be $R_{\epsilon_0}$ followed by such normalisation, i.e., $\tilde{R}_{\epsilon_0} (A_0)=\frac{1}{z_1} R_{\epsilon_0}(A)= A_*+\deltaA_1$. Let $A_1=\tilde{R}_{\epsilon_0} (A)= A_*+\deltaA_1$. By Eq.~\eqref{contr} ($h=16/15$), $\|\deltaA_1\| < t \epsilon_0^{16/15}$ and so $A_1 \in U_{\epsilon_1}$, where $\epsilon_1=t \epsilon_0^{16/15}$. Next, we get $A_2$ by applying $\tilde{R}_{\epsilon_1}$ to $A_1$. As a result, $A_2 \in U_{\epsilon_2}$, where $\epsilon_2 \in t \epsilon_1^{16/15}$. Repeating this procedure for $A_2, A_3, \ldots $, we obtain the following sequence:
\begin{subequations}\label{sequence}
    \begin{align}
         & A_0 \in U_{\epsilon_0} (A_*), \qquad \epsilon_0>0;\label{sequence0}                                                       \\
         & A_{n}=\tilde{R}_{\epsilon_{n-1}} (A_{n-1}), \qquad \epsilon_n=t\epsilon_{n-1}^{16/15} \qquad (n \geq 1).\label{sequencen}
    \end{align}
\end{subequations}
Clearly, each $A_n$ belongs to $U_{\epsilon_n} (A_*)$. Then, for $\epsilon_0$ satisfying $t \epsilon_0^{1/15}<1$, $A_n$ converges to $A_*$ super-exponentially fast. Thus, we can map $A_0$ arbitrarily close to $A_*$. Note that as $R_\epsilon: A \mapsto A'$ is an RG map with the lattice rescaling factor $s=4$, it satisfies \eqref{RG} with $s=4$. Then, in terms of $A_n$ tensors, the preservation of partition can be expressed as follows:
\begin{equation}\label{part01}
    Z(A_n, N)=\z_{n+1}^{(N/4)^3} Z(A_{n+1},N/4).
\end{equation}

Let us mention that a similar convergence result to the high-temperature fixed point was obtained in \cite{Kashapov1980} for Wilson-Kadanoff RG. There, the RG procedure was studied in terms of Hamiltonians by analysis of cluster expansions of the Gibbs measure. Our result, however, demonstrates properties of a completely different type of RG (tensor RG), which has its advantages over the Wilson-Kadanoff RG:
\begin{itemize}
    \item The Wilson-Kadanoff RG is ill-defined at low temperatures (see \cite{noLowWilson}). In contrast, tensor RG is well-defined in this regime. In \cite{Kennedy2022a}, tensor RG was successfully applied to demonstrate the existence of the first-order phase transition at the low-temperature regime of the $2D$ Ising model. The study of the low-temperature regime of the $3D$ Ising model is one possible direction for future work.

    \item Tensor RG procedures are more accessible numerically than the Wilson-Kadanoff RG. Many different algorithms are present in the literature, see \cite{TNR, TRG, Xie, EntRen, PhysRevLett.118.110504, Gilt}.
\end{itemize}

We also showed that $R_\epsilon$ is an analytic RG map. Namely, it satisfies \eqref{analit}. Eq.~\eqref{analit}, together with \eqref{contr}, allows one to show that there exists an $\epsilon>0$ such that:
\begin{itemize}
    \item The free energy density $f(A, N)=\frac{1}{N^3} \ln Z (A, N)$ is a bounded and analytic function of $A \in U_{\epsilon} (A_*)$.
    \item The limit $f_{\infty} (A)=\lim_{N \rightarrow \infty} f(A,N)$ exists and is itself a bounded analytic function of $A \in U_{\epsilon} (A_*)$.
\end{itemize}
The proof of these statements goes along the lines of the proof of Proposition~4.3 in \cite{Kennedy2022a}.

There are some features of $R_\epsilon$ which may need to be clarified. As we saw above, it is possible to achieve super-exponential convergence to $A_*$ by using different $\epsilon_n$'s at each iteration of the RG procedure (see \eqref{sequencen}). However, if we choose for all $n$, $\epsilon_n=\epsilon_0$ in \eqref{sequencen}, the sequence $A_n$ would not necessarily converge to $A_*$ even for an arbitrarily small $\epsilon_0>0$. Moreover, $A_*$ is not a fixed point for $R_{\epsilon}$, where $\epsilon>0$ (for $\epsilon=0$ the RG map is not defined). This is because the tensors $\deltaB_v^2$ used in the construction of $R_\epsilon$ do not vanish when $A=A_*$ ($\deltaA=0$), see \eqref{vA} and \eqref{vBdef}.

We can define a less confusing RG map if we allow it to violate analyticity. The construction is simple. We define $R: A \mapsto A'=z' (A_*+\deltaA')$ as $R=R_{\|\deltaA\|}$. Assuming $R$ is defined on $U_\epsilon(A_*)$ with a small enough $\epsilon$, we get:
\begin{subequations}\label{non_analyt_bounds}
    \begin{align}
         & \|\deltaA'\|<t\|\deltaA\|^{16/15},\label{non_analyt_bound} \\
         & z'=1+O(\|\deltaA\|^{16/15}),
    \end{align}
\end{subequations}
where $t$ is some universal constant. Although $R_0$ is not defined, we can still define $R(A_*)$ by continuity of $R$. Then, Eqs.~\eqref{non_analyt_bounds} imply that $A_*$ is the fixed point of $R$: $R(A_*)=A_*$. Eq.~\eqref{non_analyt_bound} implies that $A_n=\tilde{R}^n(A_0)$, where $\tilde{R}=\frac{1}{z'}R: A \mapsto A_* + \deltaA'$, converges to $A_*$ super-exponentially fast.

Let us indicate a few directions for future work.
\begin{itemize}
    \item We utilised the fact that each vertex $v$ has only three neighbours in the cube in a few places: Lemma~\ref{lem:nocrossvals}, Lemma~\ref{lem:nowqnew}, and derivation of \eqref{MMm12}. These were the only places where the lattice geometry was relevant. In light of this, it is natural to ask whether the techniques from this paper can be generalised to tensor networks on other lattices.

    \item An interesting challenge would be applying the techniques from this paper to a tensor network representing the $3D$ Ising model at low temperatures.

    \item The most exciting direction is the search for nontrivial tensor RG fixed points describing the critical points of the Ising model in $2D$ and $3D$. A rigorous proof of the existence of such points, perhaps, requires computer assistance. Obtaining the fixed points with good precision would be the first step toward this proof.
\end{itemize}

\paragraph{Acknowledgements} The author thanks Roman Gaidarov and Lev Yung for discussions on the initial draft. The author thanks Slava Rychkov and Tom Kennedy for their assistance in proof verification and advice on refining the text of this paper. The author thanks Antoine Tilloy and Clément Delcamp for communications about HOTRG. The author thanks the organisers of the research program "Quantum Information Theory" (ICMAT, March 2023), where this material was presented. The author thanks the anonymous referee for their qusetion, which led to the inclusion of  Footnote~\ref{refereequestion}.

%\paragraph{Funding} The author is supported by the Simons Foundation grant 733758.

%\paragraph{Conflicts of Interest} The author declares that there are no conflicts of interest to disclose.

\begin{appendices}
    
\section{Overview of \texorpdfstring{$2D$}{2D} results}\label{2Drd}

Here, we briefly discuss some results from \cite{Kennedy2022, Kennedy2022a}. Additionally, we provide an overview of the main ideas from \cite{Kennedy2022b} and explain how we adapted these ideas to the $3D$ case.

We consider a $2D$ tensor network on a rectangular lattice. At each site of the lattice, we have a $4$-tensor $A$. The edges of the lattice represent the contracted legs of $A$ tensors. We assume that $A$ has the following form:
\begin{equation}
    A=A_*+\deltaA,
\end{equation}
where:
\begin{itemize}
    \item $A_*$ is the tensor with single tensor element $(A_*)_{0000}=1$;
    \item $\deltaA$ is a small deviation tensor:
          \begin{equation}
              \deltaA=O(\epsilon) \qquad (\epsilon>0).
          \end{equation}
\end{itemize}

\paragraph{Result from \cite{Kennedy2022}}
The main result of \cite{Kennedy2022} is Theorem 2.1, which states that there is an RG map $A\mapsto A'=z' (A_*+\deltaA')$ such that:
\begin{equation}\label{app:eq:1}
    \|\deltaA'\|\leq C \|\deltaA\|^{3/2},
\end{equation}
where $C$ is a universal constant. Note that analyticity was not discussed in \cite{Kennedy2022}.

The suitable RG map was constructed in \cite{Kennedy2022} as a composition of three RG steps named type $0$, type I, and type II. Type $0$ and type I RG steps are the gauge transformation and the simple RG step, respectively. They were used to obtain tensor $A=A_*+\deltaA$ with $\deltaA$ having the corner structure (similar to what we did in Section~\ref{cornstr}):
\begin{equation}\label{app:eq:2}
    \deltaA=O(\epsilon), \qquad \deltaA=\includegraphics[valign=c, scale=0.8]{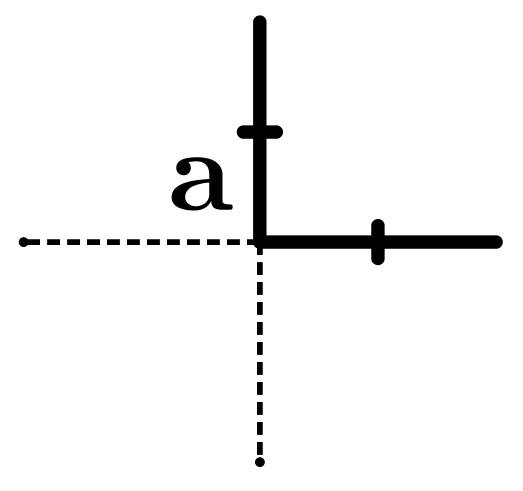}\!\!+3+O(\epsilon^2).
\end{equation}
The type II RG step (inspired by the TNR algorithm \cite{TNR}) is more complicated than the previous two. It involves the application of disentanglers, which is basically a gauge transformation acting on pairs of legs, followed by specific reconnection of tensors within the network. We refer the reader to \cite{Kennedy2022} for more details.

Let us mention that if we abandon the analyticity condition \eqref{analit} in Theorem~\ref{maintheor}, then we can find an RG map with \eqref{contr} replaced with the property analogous to \eqref{app:eq:1}, see \eqref{non_analyt_bound}.

\paragraph{Result from \cite{Kennedy2022a}}

In \cite{Kennedy2022a}, it was demonstrated how tensor RG can be used to study first-order phase transitions. In particular, it was made clear that the analyticity of an RG map plays a crucial role in the free energy analysis. Though the main focus of the work was on the low-temperature phase of the Ising model, the high-temperature result from \cite{Kennedy2022} was revisited as well. It was shown that it is possible to tweak the RG map from \cite{Kennedy2022} so that it becomes analytic. The cost is that \eqref{app:eq:1} should be replaced by a slightly weaker condition. This result was formulated in Theorem 4.1 in \cite{Kennedy2022a}, which states that there exists a tensor RG map $A \mapsto A'=z' (A_*+\deltaA')$ such that:
\begin{subequations}\label{TS_TH}
    \begin{align}
         & \text{$z', \deltaA'$ are analytic, taking real values for real $\deltaA$;}        \\
         & z'=1+O(\|\deltaA\|^2);                                                            \\
         & \|\deltaA'\|\leq C \epsilon^{1/2} \|\deltaA\| \label{contraction_Tom_and_Slavas}.
    \end{align}
\end{subequations}
Note that this result is stronger in spirit than our Theorem~\ref{maintheor}. Namely, the contraction property \eqref{contraction_Tom_and_Slavas} implies the $2D$ analogue of \eqref{contr}: $\|\deltaA'\| < C \epsilon^{3/2}$, but not other way around.

\paragraph{Overview of \cite{Kennedy2022b}}
The starting point in notes \cite{Kennedy2022b} was the tensor $A=A_*+\deltaA$ with $\deltaA$ having the corner structure \eqref{app:eq:2}. The goal was to prove the $2D$ analogue of Lemma~\ref{mainlem}. Namely, it was shown in \cite{Kennedy2022b} that there are tensors $B_1,B_2,B_3,B_4$ such that:
\begin{itemize}
    \item The following equation holds
          \begin{equation}\label{maineq2d}
              T\equiv\includegraphics[valign=c, scale=0.8]{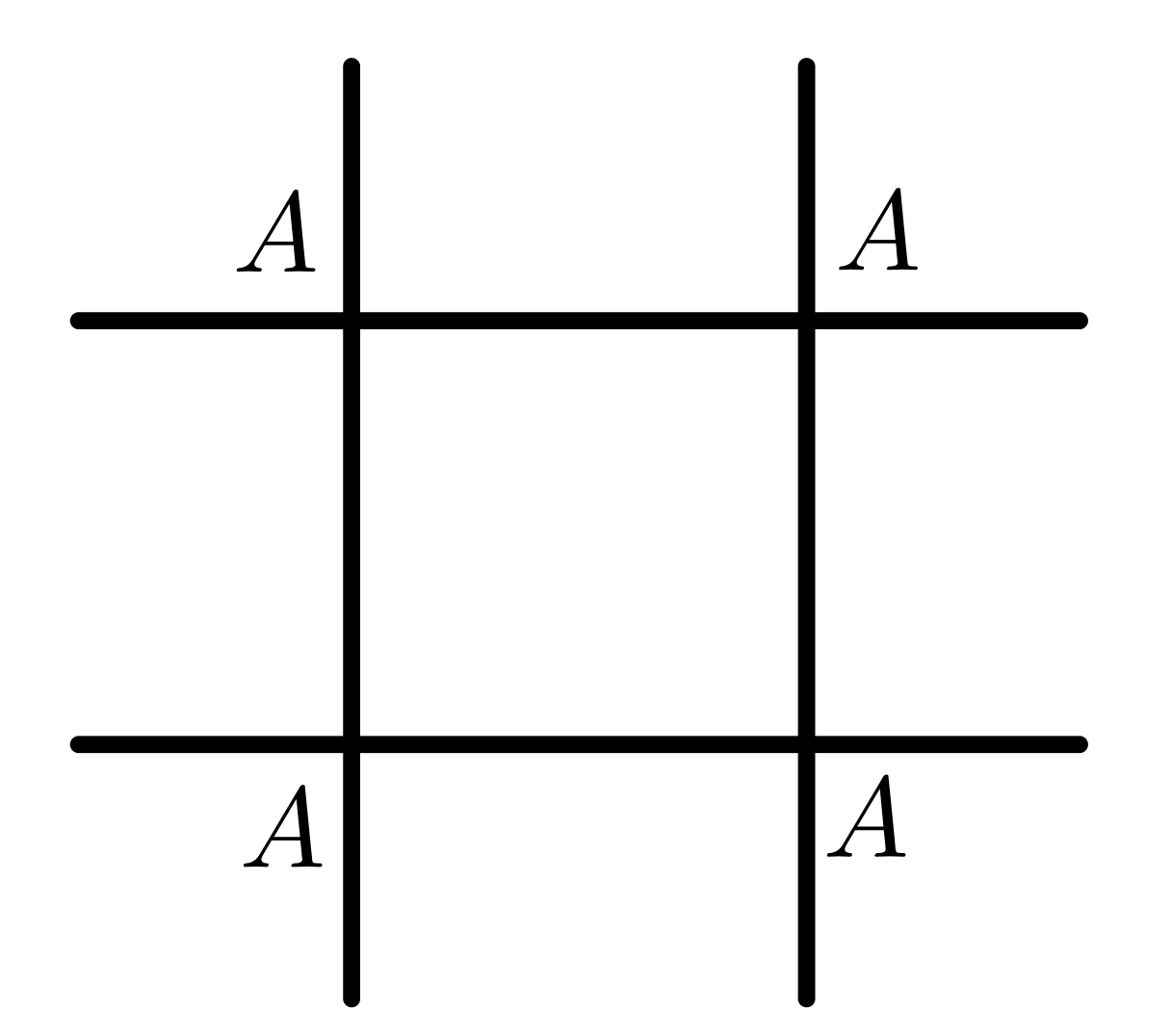}=\includegraphics[valign=c, scale=0.8]{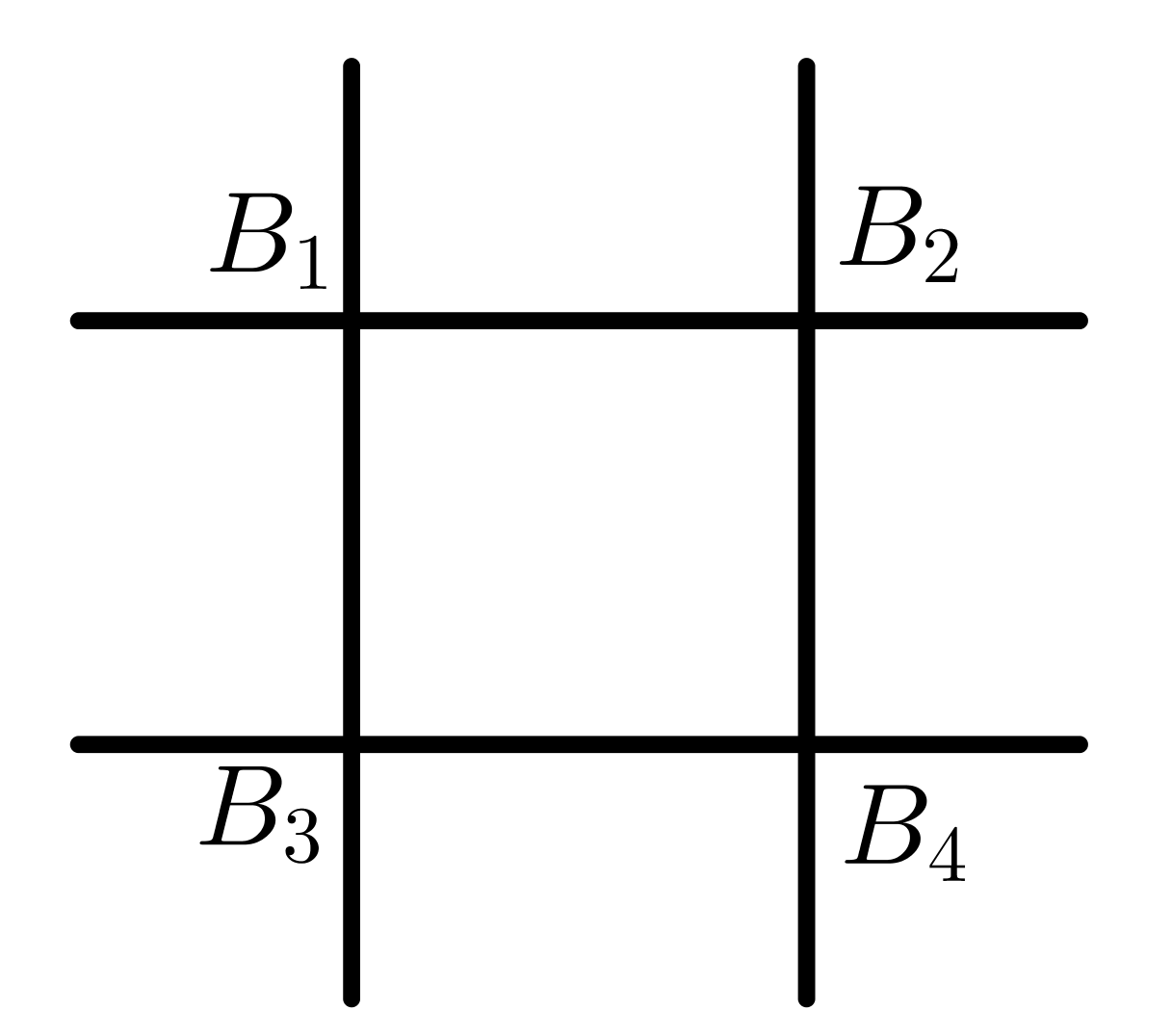} \equiv T'.
          \end{equation}
    \item Tensors $B_v$, $v=1,\ldots,4$, have the form
          \begin{equation}
              B_v=A_*+\deltaB_v,
          \end{equation}
          where $\deltaB_v$ satisfy the $2D$ analogues of \eqref{P1},\eqref{P2},\eqref{P3} (analyticity was not required in \cite{Kennedy2022b}).\footnote{Note that $T'$-external legs can be defined in $2D$ analogously to how we did it in the discussion above Lemma~\ref{mainlem}.}
\end{itemize}
Once such $B_v$ tensors are constructed, one can define the RG map analogously to how we did it in the proof of Theorem~\ref{maintheor}.

The construction of $B_v$ tensors provided in \cite{Kennedy2022b} is straightforward. The contraction $T$ (the l.h.s. of \eqref{maineq2d}) was expanded in the number of insertions of $\deltaA$.\footnote{Similarly to expansion \eqref{dUdec} in the proof of Theorem~\ref{maintheor}} Then, tensors $\deltaB_v$ were constructed by solving \eqref{maineq2d} order-by-order in the number of $\deltaA$ insertions.

\paragraph{Our contribution} The method from \cite{Kennedy2022b} presumably can be extended to $3D$. However, it requires drawing and analysing many diagrams, which is doable in $2D$ but becomes tedious in $3D$. We made the $3D$ case tractable by noting that the number of insertions of $\deltaA$ is not as significant as the "geometry" of contraction, which we encoded in the notion of a template (Section~\ref{tmplandD}). Then, we distinguished the main "geometrical" features: connectivity of $\graph(\gamma)$ and the presence of sources \eqref{source1},\eqref{source2}. This allowed us to avoid the separate analysis of dozens of diagrams. Instead, it was enough to consider only two cases: connected templates with sources (Section~\ref{dbv1}) and connected templates without sources (Section~\ref{CON2}).

There is yet another important observation that we made. In the original notes \cite{Kennedy2022b}, the issue of the unwanted terms (see Section~\ref{CON2}) was overlooked. We addressed this issue in the construction of $\deltaB_v^2$ tensors in Sections~\ref{Con2Example},~\ref{dbv2genconstr}.

Finally, we considered the analyticity properties of the construction, which was not done in \cite{Kennedy2022b}. We provided two versions of the RG map: analytic (satisfies \eqref{analit}) and non-analytic (violates \eqref{analit}). These two maps have different contraction properties. The analytic map satisfies \eqref{contr}. The non-analytic map satisfies a stronger condition \eqref{non_analyt_bound}. In the $2D$ case, the same techniques as we used in $3D$ would give the analytic RG map $A\mapsto z' (A_*+\deltaA')$ satisfying
\begin{equation}
    \|\deltaA'\| < t \epsilon^{8/7}, \quad \z'=1+O(\epsilon^{8/7}),
\end{equation}
and the non-analytic RG map satisfying
\begin{equation}
    \|\deltaA'\| < t \|\deltaA\|^{8/7}, \ \z'=1+O(\|\deltaA\|^{8/7}),
\end{equation}
where $t$ is a universal constant.

\end{appendices}

\bibliographystyle{utphys}
\bibliography{TensRG}

\end{document}